\documentclass[11pt]{article}
\usepackage{booktabs}
\usepackage[ruled, linesnumbered]{algorithm2e}

\SetAlFnt{\small}
\SetAlCapFnt{\small}
\SetAlCapNameFnt{\small}
\SetAlCapHSkip{0pt}
\IncMargin{-\parindent}

\usepackage{xcolor}

\definecolor{color1}{HTML}{CC7B2C}
\definecolor{color2}{HTML}{28607D}
\definecolor{color3}{HTML}{3F8F9B}
\definecolor{color4}{HTML}{AC4334}

\colorlet{alternativebarcolor}{black!15}
\colorlet{lightgray}{black!50!white}
\colorlet{darkgray}{black!80!white}

\usepackage[colorlinks=true, allcolors=color2]{hyperref}
\usepackage{doi}

\usepackage[T1]{fontenc}
\usepackage[margin=1in]{geometry}
\usepackage{natbib}
\usepackage[dvipsnames]{xcolor}
\usepackage{amsmath}
\usepackage{amsthm}
\usepackage{thm-restate}
\usepackage{amssymb}
\usepackage{xfrac}
\usepackage{array}
\usepackage{subcaption}
\usepackage{adjustbox}
\usepackage{makecell}
\usepackage{cleveref}
\usepackage{enumitem}
\usepackage{booktabs}
\usepackage{longtable}
\usepackage{xspace}
\usepackage{MnSymbol}

\usepackage{wrapfig}
\usepackage{thm-restate}

\crefname{theorem}{Theorem}{Theorems}
\crefname{lemma}{Lemma}{Lemmas}

\usepackage{tikz}
\usetikzlibrary{arrows.meta, positioning, calc, patterns, patterns.meta, shapes.geometric}
\tikzset{%
    phantom/.style={%
        draw, color=color1, thick
    },
    phantoms/.style={%
        phantom, ultra thick,
    },
    mov/.style={%
        color=color1, pos=0.5, anchor=south,
    },
    phlabel/.style={%
        color=color1, anchor=west, font=\small,
    }, 
    myStar/.style={
    star,
    star points=5,
    star point ratio=2.5,
    fill=yellow!40,
    minimum size=0.4cm
  }
}

\DeclareMathOperator*{\argmax}{arg\,max}

\DeclareMathOperator*{\med}{med}

\newcommand{\F}{\mathcal{F}}

\newcommand{\G}{\mathcal{G}}
\newcommand{\A}{\mathcal{A}}
\newcommand{\B}{\mathcal{B}}

\newcommand{\N}{\mathbb{N}}
\newcommand{\R}{\mathbb{R}}

\newcommand{\p}{p}   %

\newcommand{\opt}{\textsc{Opt}}
\newcommand{\optd}{\textsc{UD}}
\newcommand{\overlap}[2]{\sqcap(#1, #2)}

\newcommand{\IM}{\textsc{IndependentMarkets}\xspace}
\newcommand{\IMs}{\textsc{IM}\xspace}
\newcommand{\Ladder}{\textsc{Ladder}\xspace}
\newcommand{\GreedyMax}{\textsc{GreedyMax}\xspace}
\newcommand{\Fan}{\textsc{Fan}\xspace}
\newcommand{\Constant}{\textsc{Constant}\xspace}
\newcommand{\PiecewiseUniform}{\textsc{PiecewiseUniform}\xspace}
\newcommand{\PiecewiseUniforms}{\textsc{PU}\xspace}
\newcommand{\Ut}{\textsc{Util}\xspace}
\newcommand{\UtP}{\textsc{UtilProp}\xspace}
\newcommand{\Uts}{\textsc{U}\xspace} %
\newcommand{\GrdDec}{\textsc{GreedyDecomp}\xspace}
\newcommand{\GrdDecs}{\textsc{GD}\xspace}%
\newcommand{\UtDec}{\textsc{UtilDecomp}\xspace}
\newcommand{\smp}{single-minded proportionality\xspace}
\newcommand{\dec}{decomposability\xspace}

\newcommand{\changed}[1]{#1}

\newtheorem{theorem}{Theorem}
\newtheorem{lemma}{Lemma}
\newtheorem{definition}{Definition}
\newtheorem{corollary}{Corollary}
\newtheorem{proposition}{Proposition}

\newtheorem{example}{Example}

\usepackage{authblk}

\author[1]{Javier Cembrano}
\author[2]{Rupert Freeman}
\author[3]{\\ Ulrike Schmidt-Kraepelin}
\author[3]{Markus Utke}
\affil[1]{Universidad de Chile, Santiago, Chile}
\affil[2]{Darden School of Business, University of Virginia, VA, USA}
\affil[3]{TU Eindhoven, The Netherlands}
\title{Social Welfare in Budget Aggregation}
\date{}

\begin{document}

\begin{titlepage}

\maketitle
\vspace{0.0cm}
\begin{abstract}
We study budget aggregation under $\ell_1$-utilities, a model in which agents with heterogeneous preferences allocate a public budget across alternatives. Each agent reports their preferred allocation, and a mechanism selects an allocation. Early work has focused on social welfare maximization, which admits truthful mechanisms but may underrepresent minority groups, motivating the study of \emph{proportional} mechanisms. However, the dominant proportionality notion, \emph{single-minded proportionality}, is weak, as it only constrains outcomes when agents hold extreme preferences. To understand proportionality and its interaction with welfare and truthfulness, we address three questions.

First, how much welfare must be sacrificed to achieve proportionality? We formalize this via the \emph{price of proportionality}, the best worst-case welfare ratio between a proportional mechanism and the welfare-maximizing allocation. We introduce a single-minded proportional and truthful mechanism, \UtP, and show it achieves the optimal worst-case ratio. Second, how do proportional mechanisms compare in welfare? We define an instance-wise welfare dominance relation and compare existing mechanisms, finding that \UtP welfare-dominates all previously known single-minded proportional and truthful mechanisms. Third, can stronger proportionality notions be achieved without compromising welfare guarantees? We answer affirmatively by studying \emph{decomposability} and proposing \GrdDec, a decomposable mechanism with optimal worst-case welfare ratio.
\end{abstract}

\end{titlepage}

\section{Introduction}\label{sec:Intro}

Participatory budgeting, public investment planning, and resource allocation within organizations share a common challenge: agents with typically conflicting preferences must collectively decide how to allocate a public budget. In such settings, well-defined decision mechanisms are essential for transparency and legitimacy. Designing an appropriate mechanism, however, is nontrivial, as several desiderata must be balanced. Ideally, a mechanism should achieve high overall agreement (efficiency), ensure meaningful representation of all agents (proportionality), and incentivize truthful preference reporting (truthfulness). These goals, however, are often in tension.

We consider the model of budget aggregation under $\ell_1$-utilities \citep{lindner2008midpoint,goel2019knapsack}. Agents report their preferred budget allocations, and an agent’s utility for an outcome is modeled by the amount of budget spent in their interest; equivalently, disutility is measured by the $\ell_1$-distance between the outcome and the reported allocation. 
Unlike many voting settings in which truthfulness is unattainable, budget aggregation under $\ell_1$-utilities admits truthful mechanisms. In particular, \citet{freeman2021truthful} introduce a broad class of truthful mechanisms known as moving-phantom mechanisms.
Earlier work on this model \citep{lindner2008midpoint,goel2019knapsack} focuses on maximizing social welfare, a goal that is compatible with truthfulness. In particular, \citet{freeman2021truthful} define $\Ut$ as a neutral, anonymous, and truthful welfare-maximizing mechanism. They further characterize $\Ut$ as the unique Pareto-optimal moving-phantom mechanism. A common concern with pure welfare maximization, however, is that it may disregard minority opinions.

To address this concern, \citet{freeman2021truthful} propose the \IM mechanism, which satisfies \emph{single-minded proportionality}. This proportionality notion requires that when voters are \emph{single-minded}—seeking to allocate the entire budget to a single alternative—the budget is split proportionally to the sizes of the corresponding supporter groups. Since single-minded proportionality is relatively weak, there exists a large class of truthful mechanisms satisfying it, motivating subsequent work that seeks to further distinguish among them, for example by bounding the distance to the \emph{average} report \citep{caragiannis2022truthful,freeman2024project}.

In this work, we pursue two orthogonal approaches to refine the comparison of single-minded proportional mechanisms and to reconnect the study of proportional mechanisms with the ideal of welfare maximization. First, we evaluate these mechanisms through the lens of welfare guarantees, which naturally leads to the study of the \emph{price of proportionality}, that is, the best worst-case welfare ratio achievable by a single-minded proportional mechanism. Related questions have been studied for other utility models, for example by \citet{MPS20a}. Beyond worst-case analysis, we also introduce a welfare dominance relation that enables direct welfare comparisons between mechanisms.  
Second, we propose \emph{decomposability}, a notion of proportionality that is substantially stronger than single-minded proportionality in the sense that it constrains outcomes on all profiles rather than only single-minded ones. Overall, we present a systematic study of welfare, proportionality, and truthfulness in budget aggregation, uncovering two new mechanisms that strike a particularly favorable balance between welfare guarantees and proportionality.

Our first mechanism, \UtP, achieves the optimal worst-case approximation of social welfare among all single-minded proportional mechanisms, with welfare ratio $\textstyle \frac{n}{2\sqrt{n}-1}$ and \changed{$\textstyle \frac{\sqrt{m}+1}{2}$}, where $n$ is the number of voters and $m$ is the number of alternatives. As it is a moving-phantom mechanism, \UtP is truthful. In addition to the worst-case guarantee, \UtP maximizes welfare in every instance among all single-minded proportional moving-phantom mechanisms, and hence among all known proportional and truthful mechanisms.
To establish this result, we introduce a dominance relation: mechanism $A$ dominates mechanism $B$ if, for every instance, $A$ returns an allocation with welfare at least as high as that of $B$. This induces a \changed{preorder} over moving-phantom mechanisms and allows us to characterize the relationships among those studied in the literature. For example, perhaps surprisingly, \PiecewiseUniform \citep{caragiannis2022truthful} and \Ladder \citep{freeman2024project} welfare-dominate the more established \IM mechanism. We refer to \Cref{fig:phantom_domination} in \Cref{sec:phantom_domination} for an overview.

\begin{wrapfigure}{rt}{0.53\linewidth}
\vspace{-1.5cm}
\scalebox{.9}{
    \begin{tikzpicture}[
        mech/.style={draw=black,circle,fill,inner sep=2pt},
        every label/.append style={font=\footnotesize, fill=none, text opacity=1, inner sep=2pt},
        every node/.style={inner sep=0},
        dist/.style={very thick, dotted},
    ]
        \def\xdist{0.07}
        \def\ydist{0.03}
    
        \fill [fill=color1, opacity =0.5, rotate around = {300:(20*\xdist, 0)}] (0, 0) ellipse ({30*\xdist} and {40*\ydist});

        \fill [fill=color3,opacity =0.5, rotate around = {240:(20*\xdist, 0)}] (0, 0) ellipse ({30*\xdist} and {40*\ydist});

        \begin{scope}
         \clip [rotate around = {300:(20*\xdist, 0)}] (0, 0) ellipse ({30*\xdist} and {40*\ydist});

          \fill [fill=color4, opacity =0.5, rotate around = {240:(20*\xdist, 0)}] (0, 0) ellipse ({30*\xdist} and {40*\ydist});
        \end{scope}

        \begin{scope}
            \clip[rotate around = {240:(20*\xdist, 0)}] (-50*\xdist, -70*\ydist) rectangle (-12*\xdist,70*\ydist); 

            \fill [fill=color2, opacity=0.8, pattern=north east lines, rotate around = {240:(20*\xdist, 0)}] (0, 0) ellipse ({30*\xdist} and {40*\ydist});
        \end{scope}

        \node[color=color1] (phantoms) at (-15*\xdist,-5*\ydist) {moving-phantom};
        \node[color=color1] (phantoms) at (-15*\xdist,-15*\ydist) {mechanism};
        \node[color=color3,  align=center] (proportional) at (55*\xdist,-5*\ydist) {single-minded};
        \node[color=color3,  align=center] (proportional) at (55*\xdist,-15*\ydist) {proportional};
        \node[color=color2,  align=center] (decomposable) at (40*\xdist,115*\ydist) {decomposable};

        \node (utilprop)[myStar,fill=color1!80!color3!50!black,label={[anchor=north,fill=white, label distance=2mm]below:\UtP}] at (20*\xdist,0){};
        \node (util)[myStar,fill=color1,label={[anchor=east,fill=white, label distance=2mm]left:\Ut}] at (0*\xdist,80*\ydist){};
        \node (utildecomp)[myStar,fill=color2, label={[anchor=west, fill=white, label distance=2mm]right:\UtDec}] at (40*\xdist,70*\ydist){};
        \node (greedydecomp)[mech, draw = color2, fill=color2,label={[anchor=west,fill=white, label distance=2mm]right:\GrdDec}] at (40*\xdist,90*\ydist){};

        \draw[very thick] (util) to (utilprop);
        \draw[very thick] (util) to (utildecomp);
        \draw[very thick] (util) to (greedydecomp);
        \draw[dist] (utildecomp) to (greedydecomp);

    \end{tikzpicture}}

    \caption{Mechanisms that are instance-wise welfare-optimal within their respective class are marked with stars. The solid lines correspond to a worst-case welfare loss of $\frac{n}{2\sqrt{n}-1}$ {and $\frac{\sqrt{m}+1}{2}$} compared to \Ut. The dotted line indicates a welfare loss of at most $2 - \frac{1}{n-1}$.}\label{fig:intro}
\end{wrapfigure}
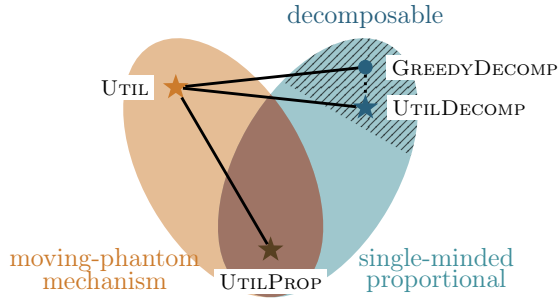

To strengthen single-minded proportionality, we draw inspiration from work on dichotomous preferences, where \emph{decomposability} has been proposed as a proportionality notion \citep{brandl2021distribution}. Decomposability requires the total budget to be partitioned into $n$ agent-specific allocations, each spent only on alternatives approved by the corresponding agent. 
We adapt this idea to the $\ell_1$-disutility model by requiring that the aggregate allocation can be decomposed into $n$ agent-specific pieces of size $1/n$, such that each agent contributes only to alternatives whose funding level does not exceed that agent's preferred level. Equivalently, no agent would prefer to reallocate their share of the budget given the contributions of the others. This condition mirrors a Nash equilibrium in an associated budget-allocation game and coincides with the \emph{equilibrium} notion introduced recently and independently by \citet{becker2025efficiently}.

We show, however, that decomposability is incompatible with truthfulness and neutrality, and therefore drop truthfulness in this part of the paper. We introduce a new decomposable mechanism, \GrdDec, which---like \UtP---minimizes worst-case welfare loss (with respect to $n$) within the class of single-minded proportional mechanisms. In particular, this demonstrates that strengthening proportionality beyond the single-minded case does not increase the price of proportionality. Although \GrdDec is worst-case optimal with respect to welfare loss, it is not optimal on every instance among decomposable mechanisms. In fact, we show that computing the outcome of \UtDec, the mechanism that maximizes welfare among all decomposable mechanisms, is NP-hard. On the positive side, we prove that \GrdDec approximates the welfare of \UtDec with a factor of $2$. We illustrate our approximation results in \Cref{fig:intro}. 

\subsection{Related Work}

\paragraph{Budget Aggregation under $\ell_1$-Utilities} We contribute to the growing literature on budget aggregation under $\ell_1$-utilities. As mentioned above, \citet{lindner2008midpoint} and \citet{goel2019knapsack} initiate the study of welfare-maximizing mechanisms and \citet{freeman2021truthful} introduce the class of moving-phantom mechanisms, including \Ut. As noted before, $\Ut$ is not single-minded proportional, and in fact, \citet{brandt2025optimal} prove that no mechanism can simultaneously be truthful, Pareto optimal, and single-minded proportional, thereby strengthening an impossibility for the class of moving-phantom mechanisms by \citet{freeman2021truthful}. One may wonder whether there exist compelling truthful mechanisms outside the class of moving-phantom mechanisms. Indeed, \citet{berg2024truthful} demonstrate the existence of other truthful mechanisms that are also anonymous and neutral. However, none of these mechanisms satisfies single-minded proportionality, hence, they are not suitable for our purposes. In terms of fairness notions that go beyond single-minded proportionality, \citet{caragiannis2022truthful} introduce the idea of mean approximation for the budget aggregation problem, which was later followed up by \citet{freeman2024project}. These two papers identify the two proportional moving-phantom mechanisms \PiecewiseUniform and \Ladder as particularly well-performing in terms of mean approximation. Recently, \citet{becker2025efficiently} introduced budget-aggregation games, in which agents control fractions of a public budget and allocate their shares to maximize their utility. They show that for $\ell_1$-utilities, Nash equilibria can be computed in polynomial time. While these equilibria correspond to our fairness notion of decomposability, their proposed algorithm is not neutral. In contrast, our algorithm \GrdDec is polynomial-time computable, neutral,  anonymous, and worst-case welfare-optimal. 
\citet{elkind2026settling} compare budget-aggregation mechanisms with respect to a broad range of axioms, including both inter-profile and fairness axioms. \citet{goyal2023low} study distortion and sample complexity in a setting similar to ours, except that each alternative’s funding is capped by a predefined cost.

\changed{The proportionality notions that we study are also related to justified-representation (JR) axioms from approval-based multiwinner voting and participatory budgeting. JR, and its strengthenings such as PJR and EJR, provide representation guarantees to sufficiently large and cohesive groups of voters~\citep{ABC+15a,SEL+26}. Single-minded proportionality can be viewed as a very weak analogue in that it guarantees proportionality only on profiles where every voter places their entire budget on a single project. Decomposability strengthens this substantially by constraining outcomes on arbitrary profiles. Unlike JR-style axioms, however, decomposability is formulated through individual contributions rather than through the representation of cohesive groups.}

\paragraph{Price of Proportionality in Social Choice}
While the price of proportionality~\citep{bertsimas2011price,caragiannis2012efficiency} has not been studied in our model, we discuss similar results in related settings. 
\changed{Notably, \citet{caragiannis2012efficiency} show that the utilitarian price of proportionality in cake cutting is also $\Theta(\sqrt{n})$, matching the asymptotic dependence on the number of voters that we obtain in our model.}
In budget division under approval preference, \citet{MPS20a} show that asking for the property of \emph{individual fair share (IFS)} (which is incomparable to single-minded proportionality in our setting) necessarily yields a worst-case welfare ratio of $\textstyle\frac{\sqrt{m}}{2}$, where $m$ is the number of projects, and show that three common algorithms achieve a worst-case welfare ratio between $\frac{\sqrt{m}}{2}$ and $\frac{m}{2\sqrt{m} -2}$.\footnote{To allow comparisons, we invert the bounds from the literature, which are reported as efficiency guarantees $\min(\text{ALG}/\text{OPT})$, where $\text{ALG}$ is the mechanism’s welfare and $\text{OPT}$ is the optimal welfare. Our objective instead is $\max(\text{OPT}/\text{ALG})$.} 
\changed{Building on the lower-bound construction of \citet{MPS20a}, we show that the exact price of single-minded proportionality in terms of $m$ is $\frac{\sqrt m+1}{2}$, which is attained by both \UtP and \GrdDec.}
\citet{TWZ20b} study the price of fairness with respect to egalitarian welfare.

In approval-based committee voting, the study of the trade-off between proportionality and utilitarian welfare was initiated by \citet{LaSk20b}, who, among other results, showed that any rule satisfying justified representation must have a welfare approximation ratio of at least 
$\textstyle\frac{k}{2\sqrt{k}-1}$, where $k$ is the committee size. \citet{EFI+24b} later demonstrated that this lower bound is asymptotically tight by providing a matching upper bound. \citet{BrPe24b} further improved this result by showing that an approximation of $\textstyle\frac{k}{2\sqrt{k}-1}$ can be achieved by completing the \emph{Greedy Justified Candidate Rule} \citep{BrPe23a}.

\subsection{Outline} In \Cref{sec:prelim}, we introduce our model and moving-phantom mechanisms. \Cref{sec:worst_case_price} defines the proportionality notions of single-minded proportionality and decomposability, and presents the algorithms \UtP and \GrdDec. To establish their optimal welfare guarantees, we first derive a lower bound and then prove a meta-theorem showing that any mechanism satisfying \changed{upper} range respect and a newly defined property called \emph{proportional spending} achieves worst-case optimal welfare among all single-minded proportional mechanisms. \Cref{sec:phantom_domination} introduces our welfare dominance relation and characterizes the welfare ordering among eight moving-phantom mechanisms. In \Cref{sec:optimal_decomposable}, we show that computing \UtDec is NP-hard, while \GrdDec achieves a $2$-approximation in welfare. 
Finally, Section 6 concludes with a discussion of relevant extensions and directions for future work.

\section{Preliminaries}\label{sec:prelim}

For $n \in \N$, we let $[n] = \{k \in \N \mid 1 \le k \le n\}$ and $[n]_0 = \{0\}\cup [n]$. For $k \in \N$, we write $\Delta^{(k)} = \{ v \in [0,1]^{k+1} \mid \sum_{k' = 1}^{k+1} v_{k'} = 1 \}$ for the standard $k$-simplex. 

\paragraph{Budget Aggregation} In budget aggregation, $n$ \textit{voters} want to decide how to distribute a normalized budget of~$1$ over $m \ge 2$ \textit{alternatives}. Each voter $i \in [n]$ reports their favored allocation $\p_i \in \Delta^{(m-1)}$ as their \textit{vote}, and all votes together form the \textit{profile} $P = (\p_1, \dots, \p_n)$. We let $\mathcal{P}_{n,m} = (\Delta^{(m-1)})^n$ denote the set of all possible profiles with $n$ voters and $m$ alternatives. 
For a profile $P \in \mathcal{P}_{n,m}$ and $k \in [n]$, let the vector $\mu^k \in [0,1]^m$ be defined by $\mu^k_j$ being the $k$-th smallest element among $\{p_{i,j} \mid i \in [n]\}$ for each $j \in [m]$. Thus, $\mu^k$ generally does not correspond to the vote of any single agent; for instance, $\mu^1$ corresponds to the vector of minimum votes per alternative. 

A \emph{budget-aggregation mechanism} (or simply a \emph{mechanism}) is a family of functions $\A_{n,m}\colon \mathcal{P}_{n,m} \rightarrow \Delta^{(m-1)}$, one for each pair $n,m \in \N$ with $m \ge 2$, that each map a profile $P \in \mathcal{P}_{n,m}$ to an \emph{aggregate} $a = \A_{n,m}(P) \in \Delta^{(m-1)}$. When $n$ and $m$ are clear from the context, we frequently write $\mathcal{A}$ instead of $\mathcal{A}_{n,m}$.

\paragraph{Utilities and Social Welfare}

We define the \textit{overlap} of two vectors $v,w \in \mathbb{R}^m$ to be $\overlap{v}{w} = \sum_{j\in [m]} \min(v_j, w_j)$ and the utility of a voter $i$ with preference $p_i \in  \Delta^{(m-1)}$ for an allocation $a \in \Delta^{(m-1)}$ as $u_i(a) = \overlap{p_i}{a}$. This utility function has been studied previously~\citep{goel2019knapsack,nehring2019resource} and induces the same ordinal rankings over budget allocations as $\ell_1$-disutilities (i.e., the $\ell_1$ distance between the vote and the aggregate). However, our results on welfare approximations depend on the specific utility formulation \changed{(see also the discussion in \Cref{sec:discussion})}. Later in the paper, we occasionally use the informal expression that a voter $i$ \emph{approves} increasing an allocation from $a_j$ to $a'_j > a_j$. This is merely shorthand for the condition $p_{i,j} \geq a'_j$. We denote by $w(a) = \sum_{i \in [n]} u_i(a)$ the (utilitarian) social welfare of an allocation $a$; we often refer to this value simply as the \textit{welfare} of an allocation, for conciseness. 
For a fixed profile $P \in \mathcal{P}_{n,m}$, we let $a^\opt\in \arg\max \{w(a)\mid a\in \Delta^{(m-1)}\}$ denote a welfare-optimal allocation and $w^\opt = w(a^\opt)$ denote the corresponding optimal welfare. 
\changed{We say that a mechanism $\A$ achieves a welfare approximation of $\alpha: \N \times \N \rightarrow \mathbb{R}_{\geq 1}$ if, for all $n,m\in \N$ with $m \ge 2$ and any instance $P \in \mathcal{P}_{n,m}$, it holds that $\frac{w^\opt}{w(\A(P))} \le \alpha(n, m)$. We slightly abuse notation and say a mechanism achieves a welfare approximation of $\alpha(n)$ (or $\alpha(m)$, respectively), when the function $\alpha$ only depends on $n$ (or $m$, respectively).}

\paragraph{Axioms}
Most mechanisms considered in this paper satisfy two basic symmetry notions with respect to voters and alternatives.
A mechanism is \textit{anonymous} if its output is invariant under permutations of the voters and \textit{neutral} if a permutation of the alternatives in all individual reports permutes the outcome in the same way. 

For a profile $P \in \mathcal{P}_{n,m}$, we say that an allocation $a$ \textit{Pareto dominates} another allocation $a'$ if $u_i(a) \ge u_i(a')$ for all voters $i \in [n]$ and $u_i(a) > u_i(a')$ for at least one voter $i \in [n]$. We call an allocation \textit{Pareto optimal} if it is \textit{not} Pareto dominated by any other allocation.   

A budget-aggregation mechanism $\A$ is \emph{truthful} if, for every $n, m \in \N$ with $m \ge 2$, profile $P = (\p_1, \dots, \p_n) \in \mathcal{P}_{n,m}$, voter $i \in [n]$, and misreport $\p_i' \in \Delta^{(m-1)}$, we have $u_i(\A(P)) \ge u_i(\A(p'_i,P_{-i}))$, where $(p'_i,P_{-i})$ is a shorthand for $(\p_1, \dots, \p_{i-1}, \p_i', \p_{i+1}, \dots, \p_n)$. 
\citet{freeman2021truthful} introduced a class of truthful mechanisms, which we define next. 

\paragraph{Moving-Phantom Mechanisms.}
A moving-phantom mechanism is defined by one \textit{phantom system} for each number of voters. 
For $n \in \N$, a \textit{phantom system} $\F_n = \{f_k \mid k=0,\dots,n\}$ is a collection of $(n+1)$ non-decreasing, continuous functions $f_k\colon [0,1] \rightarrow [0,1]$ with $f_k(0) = 0$, $f_k(1) \ge 1- \frac{k}{n}$ and $f_0(t) \ge \dots \ge f_n(t)$ for all $t \in [0,1]$.\footnote{\changed{The original definition of \citet{freeman2021truthful} required $f_k(1) \ge 1$ for all $k \in \{0, \ldots, n\}$, but it follows from the discussion in their Section 5.1 that $f_k(1) \ge 1-\frac{k}{n}$ is sufficient to guarantee normalization.}}
A family of phantom systems $\F = \{\F_n \mid n \in \N\}$ defines the moving-phantom mechanism $\A^\F$ in the following way. 
For any given profile $P = (\p_1, \dots, \p_n)$ with $p_i = (p_{i,1}, \dots, p_{i,m})$ for each $i \in [n]$, let $t^\star \in [0,1]$ be a value for which
$$
    \sum_{j \in [m]} \med(f_0(t^\star), \dots, f_n(t^\star), \p_{1,j}, \dots, \p_{n,j}) = 1. 
$$
Note that the median is well-defined since its input consists of an odd number of values. The moving-phantom mechanism $\A^\F$ induced by $\F$ returns the allocation $a=(a_1,\ldots,a_m)$, where $a_j = \med(f_0(t^\star), \dots, f_n(t^\star), \p_{1,j}, \dots, \p_{n,j})$ for each $j \in [m]$. Although $t^\star$ is not always unique, the induced allocation is, and we refer to $t^\star$ as a \emph{time of normalization}. \citet{freeman2021truthful} showed that any moving-phantom mechanism is truthful, anonymous, and neutral. 

In Appendix~\ref{app:mp-mechs}, we provide formal definitions and visualizations of all moving-phantom mechanisms introduced in previous literature and in this paper; we here give an intuitive description of those that are most relevant for our work. 
The \Ut mechanism \citep{freeman2021truthful} 
is induced by moving each phantom from 0 to 1 consecutively, one after another.  
\IM~\citep{freeman2021truthful} is 
induced by moving all phantoms simultaneously but at different speeds, namely, phantom $k$ moves with speed $\textstyle \frac{n-k}{n}$. 
The \PiecewiseUniform mechanism~\citep{caragiannis2022truthful} is induced by a two-phase moving procedure: a first phase that spreads the uppermost half of the phantoms uniformly across the $[0,1]$ interval while the lower half remains at zero, followed by a second phase that spreads the lower half of the phantoms across $\big[0, \frac{1}{2}\big]$ while concentrating the upper half of the phantoms into $\big[\frac{1}{2}, 1\big]$. 
The \Ladder mechanism~\citep{freeman2024project} is induced by moving each phantom at the same speed, but with phantom $k$ starting to move only when phantom $k-1$ reaches position $\frac{1}{n}$.
Finally, \GreedyMax~\citep{berg2024truthful} is induced by moving the top $n$ phantoms together from $0$ to $1$, while the last one remains at $0$.

\section{Price of Proportionality} \label{sec:worst_case_price}

In this section, we define two notions of proportionality and discuss the extent to which social welfare must be sacrificed to satisfy them. \changed{We quantify this welfare loss via the best welfare approximation functions $\alpha(n)$ and $\beta(m)$ that a mechanism satisfying a certain proportionality notion can guarantee, which we refer to as the \textit{price} of the corresponding proportionality notion with respect to $n$ or $m$.} 
After defining \smp and \dec in \Cref{subsec:prop-notions}, we present mechanisms satisfying them in \Cref{subsec:prop-mechanisms}, and analyze their welfare guarantees in \Cref{subsec:welfare_approximation}. 

\subsection{Proportionality Notions}\label{subsec:prop-notions}

\paragraph{Single-minded proportionality} 
{Although there is no single “correct” definition of proportionality in budget aggregation, there are some preference profiles for which proportionality is arguably uncontroversial. In particular, when each voter wishes to allocate the entire budget to a single alternative, a proportional allocation assigns to each alternative a budget share equal to its fraction of supporters. \citet{freeman2021truthful} formalize this intuition with the property of \emph{\smp}.}\footnote{While \citet{freeman2021truthful} call this property \textit{proportionality}, we reserve the term proportionality as an umbrella term.} 
A vote $p_i \in \Delta^{(m-1)}$ is called \textit{single-minded} if there is an alternative $j \in [m]$ with $p_{i,j} = 1$ and a profile $P \in \mathcal{P}_{n,m}$ is called single-minded if all votes are single-minded.
A budget aggregation mechanism $\A$ is \textit{single-minded proportional} if, for any single-minded profile $P=(p_1, \dots, p_n) \in \mathcal{P}_{n,m}$, it returns the average of the votes, i.e., $\A(P) = \frac{1}{n} \sum_{i \in [n]} p_i$.

Several mechanisms in the literature are single-minded proportional and truthful. \citet{freeman2021truthful} introduce the \IM mechanism; 
\citet{caragiannis2022truthful} and \citet{freeman2024project} introduce the \PiecewiseUniform and \Ladder mechanisms, respectively, and bound their distance from the average allocation. 
In \Cref{subsec:prop-mechanisms}, we introduce \UtP, which achieves best-possible welfare approximation \changed{with respect to $n$ and $m$} among single-minded proportional mechanisms. 

\paragraph{Decomposability}
The main issue with 
single-minded proportionality is that it only restricts the outcome of a mechanism for very specific profiles.
To obtain a stronger proportionality notion, 
we draw inspiration from the concept of \textit{decomposability}, defined by \citet{brandl2021distribution} for the approval preference setting. 
It requires that the outcome can be decomposed into voter-specific partial allocations, called contributions, each summing to $\tfrac{1}{n}$, such that every voter is satisfied with their contribution. In the approval setting, this simply means that voters contribute only to alternatives that they approve. In our setting, however, satisfaction depends on the contributions of others. We interpret a voter as satisfied if they do not fund an alternative that receives more total funding than they voted for. Otherwise, the voter would prefer to redirect that portion of their contribution to an alternative that is currently underfunded from their perspective. Such an alternative must exist by normalization of the votes and the outcome.

Formally, an aggregate $a$ is \textit{decomposable} for a profile $P$ if there exist voter \textit{contributions} $c_i \in \big[0, \frac{1}{n}\big]^m$ with $\sum_{j \in [m]} c_{i,j} = \frac{1}{n}$ for each voter $i \in [n]$, such that
\begin{itemize}
    \item $\sum_{i \in [n]} c_{i,j} = a_j$ for each alternative $j \in [m]$, and
    \item $a_j \le p_{i,j}$ for all voters $i$ and alternatives $j$ with $c_{i,j} > 0$.
\end{itemize}
A mechanism $\A$ is called decomposable if it returns a decomposable aggregate for every profile. Another way to think about decomposability is as a Nash equilibrium of the game in which each voter can distribute their fair share as they like. \citet{becker2025efficiently} call these \textit{budget-aggregation games} and show that such an equilibrium always exists and is computable in polynomial time, implying the existence of a polynomial-time decomposable mechanism.  
In \Cref{subsec:prop-mechanisms}, we present a different polynomial-time decomposable mechanism, \GrdDec, that achieves the best-possible welfare approximation \changed{with respect to $n$ and $m$}, even within the class of single-minded proportional mechanisms. Note that decomposability implies single-minded proportionality, since any single-minded voter can only contribute to their supported alternative.
However, when restricting to mechanisms satisfying our baseline property of neutrality, this strengthening comes at the cost of truthfulness. 

\begin{restatable}{proposition}{propDecompTruthfulness}\label{prop:decomp-truthfulness}
    No neutral and decomposable budget aggregation mechanism is truthful.
\end{restatable}\begin{figure}
    \centering
    \def\axispos{0.1}
    \def\coorddist{0.29}
    \def\coordwidth{0.06}
    \newcommand{\drawvoteannotated}[5]{
        \draw[ultra thick, #3] (#1*\coorddist-\coordwidth,#2) -- (#1*\coorddist+\coordwidth,#2)
        node[pos=0, anchor=east] {\large #4}
        node[pos=1, anchor=west] {\large #5};
    }
    \newcommand{\drawaggregate}[2]{
        \fill[pattern={Lines[angle=45*0.8,distance=4pt,line width=2pt]},pattern color=lightgray, draw=lightgray] (#1*\coorddist-\coordwidth,0) rectangle (#1*\coorddist+\coordwidth,#2);
    }
    \begin{subfigure}[t]{.48\textwidth}
        \begin{tikzpicture}[yscale=3,xscale=6]
            \draw (\axispos,0) -- (\axispos,1);  
            \draw (\axispos-.02,0) -- (\axispos+0.02,0); 
            \draw (\axispos-.02,1) -- (\axispos+0.02,1); 
            \node[anchor=east, xshift=-5px] at (\axispos,0) {$0$}; 
            \node[anchor=east, xshift=-5px] at (\axispos,1) {$1$}; 
            \node[anchor=south, rotate=90] at (\axispos,0.5) {\small Budget}; 
            \filldraw[fill=alternativebarcolor,draw=none] (1*\coorddist-\coordwidth,0) rectangle (1*\coorddist+\coordwidth,1); 
            \filldraw[fill=alternativebarcolor,draw=none] (2*\coorddist-\coordwidth,0) rectangle (2*\coorddist+\coordwidth,1); 
            \filldraw[fill=alternativebarcolor,draw=none] (3*\coorddist-\coordwidth,0) rectangle (3*\coorddist+\coordwidth,1); 
            \node[anchor=north east, outer sep=4.4pt] at (1*\coorddist,0) {\small Alternative:};
            \node[anchor=north, outer sep=5pt] at (1*\coorddist,0) {$1$};
            \node[anchor=north, outer sep=5pt] at (2*\coorddist,0) {$2$};
            \node[anchor=north, outer sep=5pt] at (3*\coorddist,0) {$3$};
            \drawaggregate{1}{1/3}
            \drawaggregate{2}{1/3}
            \drawaggregate{3}{1/3}
            \def\shift{0.012}
            \drawvoteannotated{1}{1/2}{color1}{$\frac{1}{2}$}{}
            \drawvoteannotated{2}{1/2}{color1}{$\frac{1}{2}$}{}
            \drawvoteannotated{3}{0}{color1}{$0$}{}
            \drawvoteannotated{1}{1/3}{color2}{}{$\frac{1}{3}$}
            \drawvoteannotated{2}{1/3}{color2}{}{$\frac{1}{3}$}
            \drawvoteannotated{3}{1/3}{color2}{}{$\frac{1}{3}$}
        \end{tikzpicture}
        \caption{Profile $P$. The aggregate of any neutral and decomposable mechanism is $a = \big(\frac{1}{3}, \frac{1}{3}, \frac{1}{3}\big)$.}
        \label{fig:decomposability_truthfulness1}
    \end{subfigure}\hfill
    \begin{subfigure}[t]{.48\textwidth}
        \centering
        \begin{tikzpicture}[yscale=3,xscale=6]
            \draw (\axispos,0) -- (\axispos,1);  
            \draw (\axispos-.02,0) -- (\axispos+0.02,0); 
            \draw (\axispos-.02,1) -- (\axispos+0.02,1); 
            \node[anchor=east, xshift=-5px] at (\axispos,0) {$0$}; 
            \node[anchor=east, xshift=-5px] at (\axispos,1) {$1$}; 
            \node[anchor=south, rotate=90] at (\axispos,0.5) {\small Budget}; 
            \filldraw[fill=alternativebarcolor,draw=none] (1*\coorddist-\coordwidth,0) rectangle (1*\coorddist+\coordwidth,1); 
            \filldraw[fill=alternativebarcolor,draw=none] (2*\coorddist-\coordwidth,0) rectangle (2*\coorddist+\coordwidth,1); 
            \filldraw[fill=alternativebarcolor,draw=none] (3*\coorddist-\coordwidth,0) rectangle (3*\coorddist+\coordwidth,1); 
            \node[anchor=north east, outer sep=4.4pt] at (1*\coorddist,0) {\small Alternative:};
            \node[anchor=north, outer sep=5pt] at (1*\coorddist,0) {$1$};
            \node[anchor=north, outer sep=5pt] at (2*\coorddist,0) {$2$};
            \node[anchor=north, outer sep=5pt] at (3*\coorddist,0) {$3$};
            \drawaggregate{1}{5/12}
            \drawaggregate{2}{7/24}
            \drawaggregate{3}{7/24}
            \drawvoteannotated{1}{5/12}{color1}{$\frac{5}{12}$}{}
            \drawvoteannotated{2}{7/24}{color1}{$\frac{7}{24}$}{}
            \drawvoteannotated{3}{7/24}{color1}{$\frac{7}{24}$}{}
            \drawvoteannotated{1}{1/3}{color2}{}{$\frac{1}{3}$}
            \drawvoteannotated{2}{1/3}{color2}{}{$\frac{1}{3}$}
            \drawvoteannotated{3}{1/3}{color2}{}{$\frac{1}{3}$}
        \end{tikzpicture}
        \caption{Profile $P'$. The aggregate of any neutral and decomposable mechanism is $a' = \big(\frac{5}{12}, \frac{7}{24}, \frac{7}{24}\big)$.}
        \label{fig:decomposability_truthfulness2}
        \end{subfigure}
    \caption{Visualization of the profiles from the proof of \Cref{prop:decomp-truthfulness}. Voters are drawn as horizontal (blue and orange) lines and aggregates are drawn as (gray) hatched areas.}
    \label{fig:decomposability_truthfulness}
\end{figure}
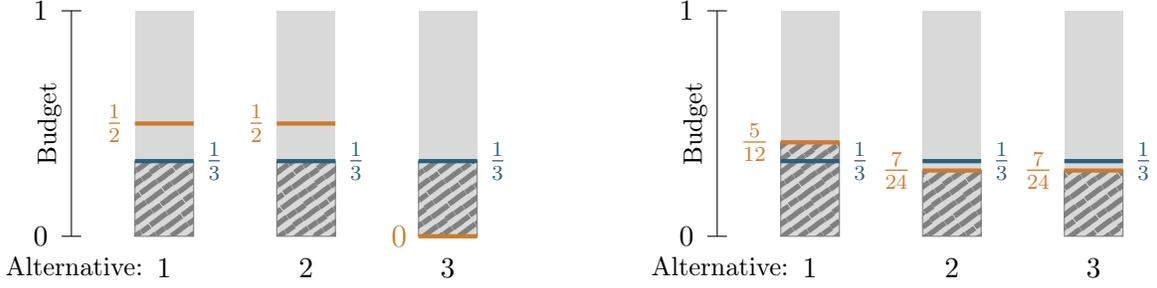
\proof
    Consider the profile $P \in \mathcal{P}_{2,3}$ with $p_1 = \big(\frac{1}{2}, \frac{1}{2}, 0\big)$ and $p_2 = \big(\frac{1}{3}, \frac{1}{3}, \frac{1}{3}\big)$ (see \Cref{fig:decomposability_truthfulness1}).
    Let $a$ be the output of a neutral, decomposable mechanism for $P$ and let $c_1$ and $c_2$ be voter contributions certifying decomposability of $a$. We claim that $a = \big(\frac{1}{3}, \frac{1}{3}, \frac{1}{3}\big)$. 
    Note that, by neutrality, we have $a_1 = a_2$. 
    Suppose first that $a_1 = a_2 > \frac{1}{3}$.
    Then, voter $2$ cannot contribute to alternatives $1$ and $2$, i.e., $c_{2,1} = c_{2,2} = 0$, which implies that this voter must put their whole contribution on alternative $3$, i.e., $c_{2,3} = \frac{1}{2}$. But then, we have $a_3 \ge \frac{1}{2} > \frac{1}{3} = p_{2,3}$, contradicting decomposability.
    On the other hand, suppose that $a_1 = a_2 < \frac{1}{3}$. Then, we have $a_3 > \frac{1}{3} = p_{2,3} > p_{1,3}$, contradicting decomposability. 

    Now, consider the profile $P' \in \mathcal{P}_{2,3}$ in which voter $1$ misreports $p_1'= \big(\frac{5}{12}, \frac{7}{24}, \frac{7}{24}\big)$ (see \Cref{fig:decomposability_truthfulness2}). 
    Let $a'$ be the output of a neutral, decomposable mechanism for $P'$, and let $c_1'$ and $c_2'$ be voter contributions certifying decomposability of $a'$. 
    We claim that $a' = \big(\frac{5}{12}, \frac{7}{24}, \frac{7}{24}\big)$. 
    By neutrality, $a_2' = a_3'$. Suppose first that $a_2' = a_3' > \frac{7}{24}$. Then, voter $1$ cannot contribute to alternative $2$ and $3$, i.e., $c'_{1,2} = c'_{1,3} = 0$, which implies that this voter must put their whole contribution on alternative $1$, i.e., $c'_{1,1} = \frac{1}{2}$. But then, we have $a_1' \ge \frac{1}{2} > \frac{5}{12} = p_{1,1}'$, contradicting decomposability.
    On the other hand, suppose that $a_2' = a_3' < \frac{7}{24}$. Then, we have $a_1' > \frac{5}{12} = p_{1,1}' > p_{2,1}$, contradicting decomposability.

    Since voter $1$ with preference $p_1$ profits from misreporting, i.e., $u_1(a) = \frac{2}{3} < \frac{17}{24} = u_1(a')$, we conclude from these profiles that any neutral and decomposable mechanism is not truthful.
\endproof

\changed{Neutrality is necessary; we show in Appendix~\ref{sec:neutrality-necessary} that the non-neutral Input-Dependent Phantom-Order Rule of~\citet{bei2026individual} is decomposable (\citeauthor{bei2026individual} show truthfulness). The rule is only defined for $n=2$ and $m=3$, but this is already sufficient to circumvent the example used in the proof of \Cref{prop:decomp-truthfulness}.}

\subsection{Proportional Mechanisms}\label{subsec:prop-mechanisms}

We now introduce two single-minded proportional mechanisms, \UtP and \GrdDec, where \UtP is additionally truthful and \GrdDec satisfies decomposability.

\paragraph{\UtP}
The \UtP mechanism is the moving-phantom mechanism induced by
\begin{equation*}
    f_k(t) = \max\bigg(0, \min\bigg(\frac{n-k}{n}, t(n+1)-k\bigg)\bigg),
\end{equation*}
which corresponds to moving each phantom with index $k$ from $0$ to $\frac{n-k}{n}$ consecutively.

\citet{freeman2021truthful} show that, if there is a time such that each phantom with index $k$ is at position $\frac{n-k}{n}$, then the corresponding mechanism is single-minded proportional. We will show in \Cref{lem:proportional_phantoms_characterization} that in fact this condition is also necessary. For now, \Cref{prop:utilprop} is an immediate implication of their result (for time $t=1$) and the fact that \UtP is a moving-phantom mechanism.

\begin{proposition}
    \UtP is truthful and single-minded proportional. \label{prop:utilprop}
\end{proposition}

In addition to truthfulness, \UtP satisfies neutrality and anonymity, since it is a moving-phantom mechanism.

\paragraph{\GrdDec}
The underlying idea of \GrdDec is to iteratively allocate the budget to alternatives 
where the maximum number of voters approve of the spending. However, since we aim for decomposability, we need to ensure that any marginal spending is covered by the budget of a voter above it. To do this, we always charge spending to the highest-voting voter with a positive remaining budget. This choice is ``safe'', as it avoids later increasing the spending on some alternative above the vote of a voter who has already contributed to that alternative, which would violate decomposability.

As a warm-up, we describe a non-neutral, non-anonymous version of this mechanism. The algorithm starts with an empty allocation $\Tilde{a} = (0, \dots, 0)$ and voter budgets $b_i = \frac{1}{n}$ for every voter $i$, and it spends these budgets over $n$ iterations. In each iteration $k$, we go through the alternatives in some arbitrary order. For each alternative, we aim to increase $\Tilde{a}_j$ to $\mu^k_j$ (the $k^{\text{th}}$ lowest vote on alternative $j$), if the (at least $n+1-k$) voters approving this marginal increase can afford it with their own budget. 
We first check if there are voters with a positive budget who vote higher than $\Tilde{a}_j$. If not, we skip this alternative. Otherwise, we find the highest-voting voter $i$ with positive budget $b_i > 0$ (subject to some tie-breaking) and increase $\Tilde{a}_j$ while charging voter $i$ for it. We do this until either voter $i$ runs out of budget, in which case we look for the next highest vote above $\Tilde{a}_j$, or until $\Tilde{a}_j$ reaches $\mu^k_j$, in which case we move on to the next alternative. 
Note that if $\Tilde{a}_j$ reaches the vote of voter $i$ (the current maximum voter with positive budget), then $\Tilde{a}_j$ also reaches $\mu^k_j = p_{i,j}$. Thus, there is no voter with positive budget voting higher than $\Tilde{a}_j$ anymore and we will not spend any more budget on alternative $j$ in further iterations, which is important for decomposability. This has the effect that in each iteration $k$ the algorithm will only consider increasing $\Tilde{a}_j$ to $\mu^k_j$, if we already spent $\Tilde{a}_j = \mu^{k-1}_j$ (where we denote $\mu^0_j = 0$) on alternative $j$.

While the procedure described above is already a decomposable mechanism, we formalize an anonymous and neutral version of it as \GrdDec in \Cref{alg:greedy_decomposable}. The conceptual differences are that, in each iteration $k$, (1) we spend the budget on all alternatives at the same time and (2) in case there are multiple voters with the maximum vote, we charge all of them equally when increasing $\Tilde{a}$.
To achieve this, in each iteration $k$, we first identify for every alternative $j$ the set $N_j^+$ of voters with the maximum vote on $j$, among those with votes higher than $\Tilde{a}_j$ that still have a positive budget. 
Given an alternative $j$ and some value $\tau \in [0,1]$, we charge each voter $i \in N_j^+$ a payment of $\pi_{i,j}(\tau) = \max\Big(0, \frac{\min(\mu^k_j, \tau) - \Tilde{a}_j}{|N_j^+|}\Big)$ and the other voters $i \in [n] \setminus N_j^+$ pay $\pi_{i,j}(\tau) = 0$. 
We can think of the value $\tau$ as a proposal to increase $\Tilde{a}$ on each alternative $j$ to $\tau$ whenever $\tau$ lies between $\Tilde{a}_j$ and $\mu^k_j$, that is, to $\max(\Tilde{a}_j, \min(\mu^k_j, \tau))$. 
Each proposed value $\tau$ then comes with payments $\pi_{i,j}$ for each voter $i \in N^+_j$ on each alternative $j$, and we need to make sure that no voter pays more than they have. Therefore, we compute the maximum value $\tau^\star \le 1$, such that $b_i \ge \sum_{j \in [m]} \pi_{i,j}(\tau^\star)$ for all voters $i \in [n]$. 
For each alternative $j \in [m]$, we then update $\Tilde{a}_j$ to $\Tilde{a}_j + \sum_{i\in [n]} \pi_{i,j}(\tau^\star)$ and $b_i$ to $b_i - \pi_{i,j}(\tau^\star)$ for $i \in N_j^+$. We repeat this iteratively, until $\tau^\star = 1$, and then continue with the next iteration for $k$. 
Note that, for all alternatives with $N^+_j \neq \emptyset$, the update of $\Tilde{a}_j$ can equivalently be written as $\Tilde{a}_j = \max(\Tilde{a}_j, \min(\mu^k_j, \tau^\star))$. We thus never increase $\Tilde{a}_j$ above $\mu^k_j$ (and therefore never above $p_{i,j}$ for any $i \in N_j^+$) in iteration $k$ and never increase $\Tilde{a}_j$ at all if there are no voters with positive budget voting above it, i.e., $N^+_j = \emptyset$.

\begin{algorithm}
    \caption{\GrdDec}\label{alg:greedy_decomposable}
    $\Tilde{a} \gets (0)_{j \in [m]}$\;
    $b_i \gets \frac{1}{n}$ for each $i \in [n]$\;
    \For(\tcp*[f]{``outer iterations''}){$k = 1, \dots, n$}{
      $\tau^\star = 0$\;
      \While(\tcp*[f]{``inner iterations''}){$\tau^\star \neq 1$}{
        $N_j^+ \gets \argmax \{p_{i,j}: i \in [n] \text{ s.t. } b_{i} > 0 \text{ and } p_{i,j} > \Tilde{a}_j\}$ for each $j \in [m]$\; 
        choose $\tau^\star \le 1$ maximal such that $b_i \ge \sum_{j \in [m]} \pi_{i,j}(\tau^\star)$ for all voters $i \in [n]$, where $\pi_{i,j}(\tau^\star) = \max\Big(0, \frac{\min(\mu^k_j, \tau^\star) - \Tilde{a}_j}{|N_j^+|}\Big)$ if $i \in N_j^+$ and $\pi_{i,j}(\tau^\star) = 0$ otherwise\;
        \For{$j \in [m], i \in [n]$}{
            $b_i \gets b_i - \pi_{i,j}(\tau^\star)$ \; \label{alg_line:update_voter_budget}
            $\Tilde{a}_j \gets \Tilde{a}_j + \pi_{i,j}(\tau^\star)$ \;
        }
      }
    }
\end{algorithm}

\begin{figure*}[t]
    \def\coorddist{0.2}
    \def\coordwidth{0.06}
    \def\shift{0.0093}
    \newcommand{\drawvoteannotated}[6]{
        \draw[ultra thick, #3] (#1*\coorddist-\coordwidth,#2) -- (#1*\coorddist+\coordwidth,#2)
        node[pos=0, inner sep=0pt, anchor=east, yshift=#6, color=black] {\small #4}
        node[pos=1, inner sep=0pt, anchor=west, yshift=#6] {\small #5};
    }
    \newcommand{\drawcontribution}[5]{
        \fill[pattern={Lines[angle=#5*0.8,distance=4pt,line width=2pt]},pattern color=#4!80!white, draw=#4!60!white] (#1*\coorddist-\coordwidth,#2) rectangle (#1*\coorddist-\coordwidth+\coordwidth*2,#3);
    }
    \centering
    \hfill
    \begin{tikzpicture}[yscale=3,xscale=6]
        \draw (0,0) -- (0,1);  
        \draw (-.02,0) -- (0.02,0); 
        \foreach \x in {1,...,6}{\draw (-.01,\x/7) -- (0.01,\x/7);}
        \draw (-.02,1) -- (0.02,1); 
        \node[anchor=east, xshift=-5px] at (0,0) {$0$}; 
        \node[anchor=east, xshift=-5px] at (0,1) {$1$}; 
        \node[anchor=south, rotate=90] at (0,0.5) {\small Budget}; 
        \filldraw[fill=alternativebarcolor,draw=none] (1*\coorddist-\coordwidth,0) rectangle (1*\coorddist+\coordwidth,1); 
        \filldraw[fill=alternativebarcolor,draw=none] (2*\coorddist-\coordwidth,0) rectangle (2*\coorddist+\coordwidth,1); 
        \filldraw[fill=alternativebarcolor,draw=none] (3*\coorddist-\coordwidth,0) rectangle (3*\coorddist+\coordwidth,1); 
        \node[anchor=north east, outer sep=4.4pt] at (1*\coorddist,0) {\small Alternative:};
        \node[anchor=north, outer sep=5pt] at (1*\coorddist,0) {$1$};
        \node[anchor=north, outer sep=5pt] at (2*\coorddist,0) {$2$};
        \node[anchor=north, outer sep=5pt] at (3*\coorddist,0) {$3$};
        %
        \drawvoteannotated{1}{1/2}{color2}{$\frac{1}{2}$}{}{0};
        \drawvoteannotated{2}{1/3}{color2}{}{}{0};
        \drawvoteannotated{3}{1/6}{color2}{}{}{0};
        \drawvoteannotated{1}{1/3}{color4}{$\frac{1}{3}$}{}{0};
        \drawvoteannotated{2}{1/6}{color4}{$\frac{1}{6}$}{}{0};
        \drawvoteannotated{3}{1/2}{color4}{}{}{0};
        \drawvoteannotated{1}{0}{color1}{}{}{0};
        \drawvoteannotated{2}{0}{color1}{}{}{0};
        \drawvoteannotated{3}{1}{color1}{$1$}{}{0};
    \end{tikzpicture}
    \hfill
    \begin{tikzpicture}[yscale=3,xscale=6]
        \draw (0,0) -- (0,1);  
        \draw (-.02,0) -- (0.02,0); 
        \foreach \x in {1,...,6}{\draw (-.01,\x/7) -- (0.01,\x/7);}
        \draw (-.02,1) -- (0.02,1); 
        \node[anchor=east, xshift=-5px] at (0,0) {$0$}; 
        \node[anchor=east, xshift=-5px] at (0,1) {$1$}; 
        \node[anchor=south, rotate=90] at (0,0.5) {\small Budget}; 
        \filldraw[fill=alternativebarcolor,draw=none] (1*\coorddist-\coordwidth,0) rectangle (1*\coorddist+\coordwidth,1); 
        \filldraw[fill=alternativebarcolor,draw=none] (2*\coorddist-\coordwidth,0) rectangle (2*\coorddist+\coordwidth,1); 
        \filldraw[fill=alternativebarcolor,draw=none] (3*\coorddist-\coordwidth,0) rectangle (3*\coorddist+\coordwidth,1); 
        \node[anchor=north east, outer sep=4.4pt] at (1*\coorddist,0) {\small Alternative:};
        \node[anchor=north, outer sep=5pt] at (1*\coorddist,0) {$1$};
        \node[anchor=north, outer sep=5pt] at (2*\coorddist,0) {$2$};
        \node[anchor=north, outer sep=5pt] at (3*\coorddist,0) {$3$};
        %
        \drawcontribution{3}{0}{1/6}{color1}{45};
        %
        \drawvoteannotated{1}{1/2}{color2}{$\frac{1}{2}$}{}{0};
        \drawvoteannotated{2}{1/3}{color2}{}{}{0};
        \drawvoteannotated{3}{1/6}{color2}{}{}{0};
        \drawvoteannotated{1}{1/3}{color4}{$\frac{1}{3}$}{}{0};
        \drawvoteannotated{2}{1/6}{color4}{$\frac{1}{6}$}{}{0};
        \drawvoteannotated{3}{1/2}{color4}{}{}{0};
        \drawvoteannotated{1}{0}{color1}{}{}{0};
        \drawvoteannotated{2}{0}{color1}{}{}{0};
        \drawvoteannotated{3}{1}{color1}{$1$}{}{0};
    \end{tikzpicture}
    \hfill
    \begin{tikzpicture}[yscale=3,xscale=6]
        \draw (0,0) -- (0,1);  
        \draw (-.02,0) -- (0.02,0); 
        \foreach \x in {1,...,6}{\draw (-.01,\x/7) -- (0.01,\x/7);}
        \draw (-.02,1) -- (0.02,1); 
        \node[anchor=east, xshift=-5px] at (0,0) {$0$}; 
        \node[anchor=east, xshift=-5px] at (0,1) {$1$}; 
        \node[anchor=south, rotate=90] at (0,0.5) {\small Budget}; 
        \filldraw[fill=alternativebarcolor,draw=none] (1*\coorddist-\coordwidth,0) rectangle (1*\coorddist+\coordwidth,1); 
        \filldraw[fill=alternativebarcolor,draw=none] (2*\coorddist-\coordwidth,0) rectangle (2*\coorddist+\coordwidth,1); 
        \filldraw[fill=alternativebarcolor,draw=none] (3*\coorddist-\coordwidth,0) rectangle (3*\coorddist+\coordwidth,1); 
        \node[anchor=north east, outer sep=4.4pt] at (1*\coorddist,0) {\small Alternative:};
        \node[anchor=north, outer sep=5pt] at (1*\coorddist,0) {$1$};
        \node[anchor=north, outer sep=5pt] at (2*\coorddist,0) {$2$};
        \node[anchor=north, outer sep=5pt] at (3*\coorddist,0) {$3$};
        %
        \drawcontribution{3}{0}{1/6}{color1}{45};
        \drawcontribution{1}{0}{1/6}{color2}{45};
        \drawcontribution{2}{0}{1/6}{color2}{45};
        \drawcontribution{1}{1/6}{1/3}{color4}{-45};
        \drawcontribution{3}{1/6}{1/3}{color1}{45};
        \drawcontribution{3}{1/3}{1/2}{color4}{-45};
        %
        \drawvoteannotated{1}{1/2}{color2}{$\frac{1}{2}$}{}{0};
        \drawvoteannotated{2}{1/3}{color2}{}{}{0};
        \drawvoteannotated{3}{1/6}{color2}{}{}{0};
        \drawvoteannotated{1}{1/3}{color4}{$\frac{1}{3}$}{}{0};
        \drawvoteannotated{2}{1/6}{color4}{$\frac{1}{6}$}{}{0};
        \drawvoteannotated{3}{1/2}{color4}{}{}{0};
        \drawvoteannotated{1}{0}{color1}{}{}{0};
        \drawvoteannotated{2}{0}{color1}{}{}{0};
        \drawvoteannotated{3}{1}{color1}{$1$}{}{0};
    \end{tikzpicture}
    \caption{Illustration of \GrdDec on the profile from \Cref{ex:greedydecomp-example}. Colored hatchings represent contributions of the voter of corresponding color. The left figure indicates the initial profile, the center figure indicates the progression of \GrdDec after the first outer iteration, and the right figure indicates the progression of \GrdDec after the second outer iteration, which is also the final allocation.}
    \label{fig:greedydecomp}
\end{figure*}

\begin{example} \label{ex:greedydecomp-example}
   \changed{Consider the profile with three alternatives, three voters, and votes $\big(\frac{1}{2},\frac{1}{3},\frac{1}{6}\big)$, $\big(\frac{1}{3},\frac{1}{6},\frac{1}{2}\big)$, and $(0,0,1)$. Starting from $\Tilde{a}=(0,0,0)$, each voter has budget $\frac{1}{3}$. The initial profile is shown in \Cref{fig:greedydecomp} (left). The key role of $\tau^*$ is that it determines the highest common level that can be reached in the current inner step while the paying voters can still afford to fill each alternative up to $\min(\mu_j^k,\tau^*)$.
    In the first outer iteration, the target levels are $\mu^1=\big(0,0,\frac{1}{6}\big)$. Thus only alternative 3 can be filled, and voter 3 can afford to raise the level all the way to the target $\mu_3^1=\frac{1}{6}$, so $\tau^*=1$ and we obtain $\Tilde{a}=\big(0,0,\frac{1}{6}\big)$. This is shown in \Cref{fig:greedydecomp} (center).}
    
    \changed{In the second outer iteration, the target levels are $\mu^2=\big(\frac{1}{3},\frac{1}{6},\frac{1}{2}\big)$. At first, voter 1 is the paying voter on alternatives 1 and 2, and voter 3 on alternative 3, so the common level can only be raised to $\tau^*=\frac{1}{6}$ before voter 1 runs out of budget, giving $\big(\frac{1}{6},\frac{1}{6},\frac{1}{6}\big)$.
    From this point on, the allocation on alternative 2 will not change anymore, because there is no voter with positive remaining budget whose vote on alternative 2 is strictly above that level.
    After that, voter 2 becomes the paying voter on alternative 1, while voter 3 remains the paying voter on alternative 3, and the next feasible value is $\tau^*=\frac{1}{3}$, which yields $\big(\frac{1}{3},\frac{1}{6},\frac{1}{3}\big)$, exhausting voter 3's budget and reaching the cap $\mu^2_1$ on alternative 1.
    Finally, because voter 2 still has budget and the cap on alternative 3 has still not been reached, the allocation on alternative 3 can be raised the rest of the way to its target $\mu_3^2=\frac{1}{2}$.
    Again this corresponds to $\tau^*=1$, because voter 2 can afford the full remaining increase.
    Thus the final allocation is $\big(\frac{1}{3},\frac{1}{6},\frac{1}{2}\big)$, shown in \Cref{fig:greedydecomp} (right).}
\end{example}

The following proposition summarizes the main properties of \GrdDec: It outputs a feasible aggregate and it satisfies decomposability.

\begin{restatable}{proposition}{propGreedyDecomp}
    \GrdDec is a decomposable budget aggregation mechanism and can be computed in polynomial time. \label{prop:greedydecomp_decomposable}
\end{restatable}

\proof
    We first show that after termination of \GrdDec, each voter $i$ has a remaining budget of $b_i = 0$. Since each voter starts with a budget of $\frac{1}{n}$ and every amount that is subtracted from the voters' budgets is added to the allocation, this then implies that \GrdDec returns a valid allocation after it terminates.
    Suppose there is some voter $i$ with $b_i > 0$ and therefore $\sum_{j \in [m]} a_j < 1$ after termination. Let $j$ be an alternative with $p_{i,j} > a_j$ (such an alternative must exist, because votes are normalized). Choose $k$ such that $\mu_j^k = p_{i,j}$. Then in the $k^\text{th}$ outer iteration and after the last inner iteration with $\tau^\star = 1$, we have $p_{i,j} > a_j \ge \Tilde{a}_j$, which means $N_j^+$ is non-empty. But then the payment for some voter $i' \in N^+_j$ and alternative $j$ is $\pi_{i',j} = \max \Big(0, \frac{\min(\mu^k_j, \tau^\star) - \Tilde{a}_j}{|N_j^+|}\Big) \ge \frac{\min(\mu^k_j, \tau^\star) - \Tilde{a}_j}{|N_j^+|} = \frac{p_{i,j} - \Tilde{a}_j}{{|N_j^+|}} > 0$ and thus $\Tilde{a}_j$ would have been increased to some value strictly above $a_j$.

    We now show that \GrdDec is decomposable. We can decompose the output of \GrdDec by tracking the amount of budget each voter pays for each alternative during the computation of \GrdDec. More formally, we can initialize $c_{i,j} = 0$ for all voters $i \in [n]$ and alternatives $j \in [m]$ and insert $c_{i,j} \gets c_{i,j} + \pi_{i,j}(\tau^\star)$ after \Cref{alg_line:update_voter_budget}. Since the budget of each voter $i$ is fully exhausted (as shown above), we know that $\sum_{j \in [m]} c_{i,j} = \frac{1}{n}$ and it only remains to show that no voter $i$ is charged for an alternative $j$ with $a_j > p_{i,j}$. 
    Let $i \in [n]$ be a voter with $c_{i,j} > 0$ for some alternative $j$. Consider an inner iteration in which $c_{i,j}$ is increased and let $k$ be the index of the current outer iteration. In that iteration voter $i$ has a positive payment $\pi_{i,j}(\tau^\star) > 0$, which implies $i \in N^+_j$. 
    When $i \in N^+_j$ in outer iteration $k$, we know that $p_{i,j} > \Tilde{a}_j$ and $\Tilde{a}_j \ge \mu^{k-1}_j$ (by the invariant of the algorithm explained in \Cref{subsec:prop-mechanisms}). By the definition of $\mu$, there is no vote between $\mu^{k-1}_j$ and $\mu^k_j$ and therefore $p_{i,j} \ge \mu^k_j$.
    By the definition of $N^+_j$, there can be no voter with positive budget voting above $p_{i,j}$ on alternative $j$, which means that $\Tilde{a}_j$ can never be increased above $p_{i,j}$.

    Finally, we show that \GrdDec always terminates and that the output can be computed in polynomial time.
    In every inner iteration of the algorithm, at least (i) one voter's budget is exhausted ($b_i = \sum_{j \in [m]} \pi_{i,j}(\tau^\star)>0$ for some voter $i$), or (ii) $\tau^\star = 1$ and the next outer iteration starts. 
    Also note that each voter can only exhaust their budget in a single iteration. Thus, we can have no more than $2n$ inner iterations in total.
    To show that \GrdDec is computable in polynomial time, it only remains to show that, in each iteration, we can compute the next value of $\tau^\star$ efficiently. To do so, we iteratively try out different values $\tau$, until a voter overpays or $\tau = 1$ is reached. 
    Say we are in the $k^\text{th}$ outer iteration with a current allocation $\Tilde{a}$, sets $N_j^+$ for each $j \in [m]$, and voter budgets $b_i$ for each $i \in [n]$. 
    We compute the payments for all voters for the values $\tau \in \{0, 1\} \cup \{\Tilde{a}_j \mid j \in [m]\} \cup \{ \mu^k_j \mid j\in[m] \}$. If we have $b_i \ge \sum_{j \in [m]} \pi_{i,j}(1)$ for all voters $i \in [n]$, then we set $\tau^\star = 1$, as no voter overpays. 
    Otherwise, we identify the two lowest consecutive values $\tau_1, \tau_2$ from the set with $\sum_{j \in [m]} \pi_{i,j}(\tau_1) < b_i \le \sum_{j \in [m]} \pi_{i,j}(\tau_2)$ for some voter $i$ with $b_i > 0$. 
    Note that, for all $j \in [m]$ and all such voters $i$, the payment $\pi_{i,j}(\tau) = \max \Big(0, \frac{\min(\mu^k_j, \tau) - \Tilde{a}_j}{|N_j^+|}\Big)$ is linear within the interval $[\tau_1, \tau_2]$, as we do not cross $\Tilde{a}_j$ or $\mu^k_j$. 
    Therefore, we can compute the value $\tau^\star_i$ for which each such voter $i$ runs out of budget as
    $$\tau^\star_i = \tau_1 + \frac{b_i - \sum_{j \in [m]} \pi_{i,j}(\tau_1)}{\sum_{j \in [m]} \pi_{i,j}(\tau_2) - \sum_{j \in [m]} \pi_{i,j}(\tau_1)} (\tau_2 - \tau_1).$$
    Setting $\tau^\star = \min_{i \in [n]} \tau^\star_i$ yields the minimum value for which some voter runs out of budget.
\endproof

We note that, in addition to decomposability, \GrdDec satisfies neutrality and anonymity, because it treats voters and alternatives symmetrically when iteratively allocating the budget.

\subsection{Welfare Approximation: Lower Bounds} \label{subsec:welfare_approximation}

Given the existence of single-minded proportional and decomposable mechanisms, we now turn to the question of how much welfare has to be sacrificed to satisfy these properties. 
\changed{
The functions $\alpha^\star, \beta^\star \colon \N \to \R$, defined as 
\[
    \alpha^\star(n) = \max_{\ell\in [n]} \frac{n\ell}{n+\ell(\ell-1)}\quad \text{for all } n\in \N, \qquad \beta^\star(m) = \frac{\sqrt{m}+1}{2}\quad \text{for all } m \in \N
\]
play a key role in what follows.
}
The following simple proposition, whose proof is deferred to Appendix~\ref{app:propAlphaStar}, provides a simple 
upper bound on the value of the function $\alpha^\star$ for each $n$.
\begin{restatable}{proposition}{propAlphaStar}\label{prop:alpha-star}
    For every $n\in \N$, we have
    \(
        \alpha^\star(n) 
        \leq \frac{n}{2\sqrt{n}-1}.
    \)
\end{restatable}

\changed{
We next show that $\alpha^\star(n)$ is a lower bound on the welfare approximation in terms of $n$ for single-minded proportional mechanisms.
}

\begin{proposition} \label{prop:proportional_social_welfare_lower_bound_n}
    For every $n\in \N$ and every single-minded proportional budget aggregation mechanism $\A$, there exists $m\in \N$ and a profile $P\in \mathcal{P}_{n,m}$ such that $\frac{w^\opt}{w(\A(P))} \ge \alpha^\star(n)$. 
\end{proposition}

\changed{Before giving a proof for \Cref{prop:proportional_social_welfare_lower_bound_n}, we remark that a closely related instance was used by \citet[Theorem 1]{MPS20a} to show that, under approval preferences, the welfare loss implied by individual fair share is at least $\frac{\sqrt{m}}{2}$. Their lower bound of $\frac{\sqrt{m}}{2}$ carries over to our setting for single-minded proportionality, but we strengthen this bound slightly to $\beta^\star(m) = \frac{\sqrt{m}+1}{2}$, and give the proof in Appendix~\ref{app:welfare_m_approximations}.}

\changed{
\begin{restatable}{proposition}{mlowerbound} \label{prop:proportional_social_welfare_lower_bound_m}
    For every $m \in \N$ with $m \geq 2$, every single-minded proportional budget aggregation mechanism $\A$, and every $\varepsilon>0$, there exists $n\in \N$ and a profile $P\in \mathcal{P}_{n,m}$ such that $\frac{w^\opt}{w(\A(P))} \ge \beta^\star(m)-\varepsilon$.
\end{restatable}
}

\begin{proof}[Proof of \Cref{prop:proportional_social_welfare_lower_bound_n}]
    Fix the values $n\in \N$ and $\ell\in [n]$ and the mechanism $\A$, and consider the following single-minded instance with $n$ voters and $m = n-\ell+1$ alternatives.
    The first $n-\ell$ voters each want to allocate their entire budget to a different alternative, i.e., for $i \in [n-\ell]$ we have $p_{i,j} = 1$ if $i = j$ and $p_{i,j} = 0$ otherwise. The remaining $\ell$ voters want to allocate their entire budget to the last alternative, i.e., for $i \in \{n-\ell+1, \dots, n\}$ we have $p_{i,j} = 1$ if $j = m$ and $p_{i,j} = 0$ otherwise. See \Cref{fig:proportional_social_welfare_lower_bound} for a visualization of this profile in the case where $\ell=\sqrt{n}\in \N$.

    Since $\A$ is single-minded proportional we have $\A(P) = \big(\frac{1}{n}, \dots, \frac{1}{n}, \frac{\ell}{n}\big)$, while an optimal outcome in terms of welfare (in fact, the unique one when $\ell>1$) is $a^\opt = (0, \dots, 0, 1)$. Thus,
    \[
        \frac{w(a^\opt)}{w(\A(P))} = \frac{\ell}{\frac{1}{n}(n-\ell) + \frac{\ell}{n} \ell} = \frac{n\ell}{n+\ell(\ell-1)}.
    \]
    Since this analysis is valid for any $\ell\in [n]$, it is valid in particular for the one that maximizes the last term, which finishes the proof.
\end{proof}
\begin{figure}
    \centering
    \def\coorddist{0.18}
    \def\coordwidth{0.06}
    \newcommand{\drawaggregatearea}[5]{
        \fill[pattern={Lines[angle=#4*0.8,distance=4pt,line width=2pt]},pattern color=#3!80!white, draw=#3!60!white] (#1*\coorddist-\coordwidth,0) rectangle (#1*\coorddist+\coordwidth,#2);
    }
    \newcommand{\drawaggregatelineannotated}[5]{
        \draw[#3] (#1*\coorddist-\coordwidth,#2) -- (#1*\coorddist+\coordwidth,#2) node[anchor=west] {\small #5};
    }
    \begin{tikzpicture}[yscale=2,xscale=5]
        \draw (0,0) -- (0,1);  
        \draw (-.02,0) -- (0.02,0); 
        \draw (-.02,1) -- (0.02,1); 
        \node[anchor=east, xshift=-5px] at (0,0) {$0$}; 
        \node[anchor=east, xshift=-5px] at (0,1) {$1$}; 
        \node[anchor=south, rotate=90] at (0,0.5) {\small Budget}; 
        \filldraw[fill=alternativebarcolor,draw=none] (1*\coorddist-\coordwidth,0) rectangle (1*\coorddist+\coordwidth,1); 
        \node at (2*\coorddist,0.5) {$\cdots$}; 
        \filldraw[fill=alternativebarcolor,draw=none] (3*\coorddist-\coordwidth,0) rectangle (3*\coorddist+\coordwidth,1); 
        \filldraw[fill=alternativebarcolor,draw=none] (5*\coorddist-\coordwidth,0) rectangle (5*\coorddist+\coordwidth,1); 
        \node[anchor=north east, outer sep=4.4pt] at (1*\coorddist,0) {\small Alternative:};
        \node[anchor=north, outer sep=5pt] at (1*\coorddist,0) {$1$};
        \node[anchor=north, outer sep=5pt] at (2*\coorddist,0) {$\cdots$\phantom{2}};
        \node[anchor=north, outer sep=5pt] at (3*\coorddist,0) {$n-\sqrt{n}$};
        \node[anchor=north, outer sep=5pt] at (5*\coorddist,0) {$n-\sqrt{n}+1$};
        \drawaggregatearea{5}{1}{lightgray}{45}{$1$}
        \drawaggregatearea{1}{0.1}{darkgray}{-45}{$\frac{1}{n}$}
        \drawaggregatearea{3}{0.1}{darkgray}{-45}{$\frac{1}{n}$}
        \drawaggregatearea{5}{0.35}{darkgray}{-45}{$\frac{1}{\sqrt{n}}$}
        \drawaggregatelineannotated{5}{1}{lightgray}{45}{$1$}
        \drawaggregatelineannotated{1}{0.1}{darkgray}{-45}{$\frac{1}{n}$}
        \drawaggregatelineannotated{3}{0.1}{darkgray}{-45}{$\frac{1}{n}$}
        \drawaggregatelineannotated{5}{0.35}{darkgray}{-45}{$\frac{1}{\sqrt{n}}$}
        \draw[ultra thick, color1] (1*\coorddist-\coordwidth,1) -- (1*\coorddist+\coordwidth,1) node[pos=0, anchor=east, inner sep=1pt] {\small $1\times$}; 
        \draw[ultra thick, color1] (3*\coorddist-\coordwidth,1) -- (3*\coorddist+\coordwidth,1) node[pos=0, anchor=east, inner sep=1pt] {\small $1\times$}; 
        \draw[ultra thick, color1] (5*\coorddist-\coordwidth,1) -- (5*\coorddist+\coordwidth,1) node[pos=0, anchor=east, inner sep=1pt] {\small $\sqrt{n}\,\times$}; 
    \end{tikzpicture}
    \caption{Example instance used for the proof of \Cref{prop:proportional_social_welfare_lower_bound_n}. Only non-zero votes are drawn. The numbers to the left of the votes (``$\sqrt{n} \, \times$'') indicate the number of voters voting that value. The aggregate of any single-minded proportional mechanism is drawn in dark gray (downwards hatching) and the welfare-optimal aggregate in light gray (upwards hatching).}
    \label{fig:proportional_social_welfare_lower_bound}
\end{figure}
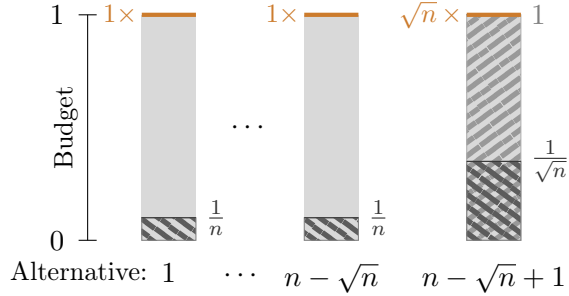

\subsection{Welfare Approximation: Upper Bounds}

In this section, we show that the above lower bounds on the welfare approximation that single-minded proportional mechanisms can guarantee are tight and are attained by both \UtP and \GrdDec.
This implies that, on top of a best worst-case welfare guarantee, we can either add truthfulness or strengthen single-minded proportionality to decomposability. 
To prove that both mechanisms meet the welfare bound in terms of $n$ from \Cref{prop:proportional_social_welfare_lower_bound_n}, we introduce two properties of budget aggregation mechanisms that, together, suffice to achieve this guarantee. We then show that \UtP and \GrdDec satisfy these properties. \changed{Afterwards, we use a similar approach to show the bound with respect to $m$.}

\subsubsection{Upper Bound with Respect to the Number of Voters}

\paragraph{Upper range respect}
\changed{The first property states that a mechanism should never spend more than the maximum vote on each alternative.
Formally, a mechanism $\A$ is \emph{upper range respecting} if, for every profile $P \in \mathcal{P}_{n,m}$, the output of $\A$ lies below the maximum vote on each alternative $j$, i.e., $\A(P)_j \le \max_{i \in [n]} p_{i,j} = \mu^n_j$.}

\paragraph{Proportional spending} The second property requires a mechanism to spend, for any $k \in [n]$, at least $\frac{n-k+1}{n}$ 
in total below the ceiling formed by the $k^\text{th}$ lowest votes on the alternatives (if possible), 
since that is the spending supported by $n-k+1$ voters.
Formally, a mechanism satisfies \emph{proportional spending} if, for every profile $P \in \mathcal{P}_{n,m}$ and $k \in [n]$, we have 
\[
    \overlap{\A(P)}{\mu^k} \ge \min\Bigg( \frac{n-k+1}{n}, \sum_{j \in [m]} \mu^k_j \Bigg).
\]

\changed{Note that proportional spending with $k=1$ implies the lower range respect property analogous to upper range respect (i.e., that a mechanism should never spend less than the minimum vote on each alternative).}

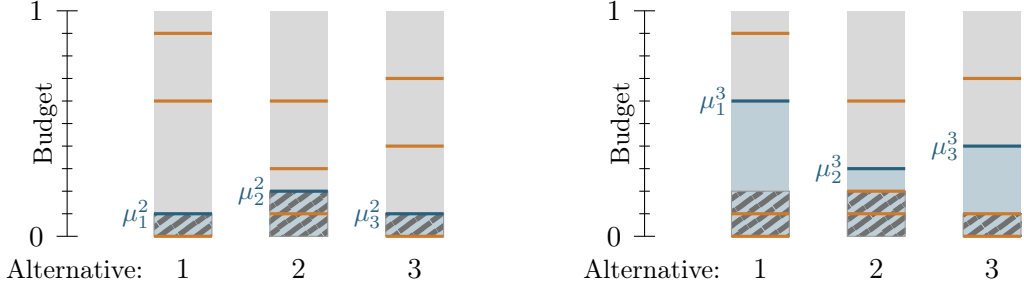
\begin{figure}
    \centering
    \def\coordwidth{0.055}
    \def\coorddist{0.22}
    \def\posone{0.22}
    \def\postwo{\posone + \coorddist}
    \def\posthr{\postwo + \coorddist}
    \def\posfou{\posthr + \coorddist}
    \newcommand{\drawvoteannotated}[5]{
        \draw[very thick, #3] (#1-\coordwidth,#2) -- (#1+\coordwidth,#2)
        node[pos=0, inner sep=1pt, anchor=east] {\small #4}
        node[pos=1, inner sep=2pt, anchor=west] {\small #5};
    }
    \newcommand{\drawVOTEannotated}[5]{
        \draw[very thick, #3] (#1-\coordwidth,#2) -- (#1+\coordwidth,#2)
        node[pos=0, inner sep=1pt, anchor=east] {\small #4}
        node[pos=1, inner sep=2pt, anchor=west] {\small #5};
    }
    \newcommand{\drawaggregate}[5]{
        \fill[pattern={Lines[angle=#5*0.8,distance=4pt,line width=2pt]},pattern color=#4!70!white, draw=#4!50!white] (#1-\coordwidth,#2) rectangle (#1-\coordwidth+\coordwidth*2,#3);
    }
    \newcommand{\drawarea}[5]{
        \fill[color=#4!30!white] (#1-\coordwidth,#2) rectangle (#1-\coordwidth+\coordwidth*2,#3);
    }
    \begin{subfigure}[t]{0.48\textwidth}
        \begin{tikzpicture}[yscale=2.5,xscale=6]
            \draw (0,0) -- (0,1);  
            \draw (-.02,0) -- (0.02,0); 
            \draw (-.02,1) -- (0.02,1); 
            \node[anchor=east, xshift=-5px] at (0,0) {$0$}; 
            \foreach \x in {1,...,9}{\draw (-.01,\x/10) -- (0.01,\x/10);}
            \node[anchor=east, xshift=-5px] at (0,1) {$1$}; 
            \node[anchor=south, rotate=90] at (0,0.5) {\small Budget}; 
            \filldraw[fill=alternativebarcolor,draw=none] (\posone-\coordwidth,0) rectangle (\posone+\coordwidth,1); 
            \filldraw[fill=alternativebarcolor,draw=none] (\postwo-\coordwidth,0) rectangle (\postwo+\coordwidth,1); 
            \filldraw[fill=alternativebarcolor,draw=none] (\posthr-\coordwidth,0) rectangle (\posthr+\coordwidth,1); 
            \node[anchor=north, outer sep=5pt] (label1) at (\posone,0) {1};
            \node[anchor=north, outer sep=5pt] at (\postwo,0) {2};
            \node[anchor=north, outer sep=5pt] at (\posthr,0) {3};
            \node[anchor=east, outer sep=0pt] at (label1.west) {\small Alternative:};
            \drawarea{\posone}{0}{1/10}{color2}{45};
            \drawarea{\postwo}{0}{2/10}{color2}{45};
            \drawarea{\posthr}{0}{1/10}{color2}{45};
            \drawaggregate{\posone}{0}{1/10}{darkgray}{45};
            \drawaggregate{\postwo}{0}{2/10}{darkgray}{45};
            \drawaggregate{\posthr}{0}{1/10}{darkgray}{45};
            \drawvoteannotated{\posone}{9/10}{color1}{}{};
            \drawvoteannotated{\posone}{0/10}{color1}{}{};
            \drawVOTEannotated{\posone}{1/10}{color2}{$\mu^2_1$}{};
            \drawvoteannotated{\posone}{6/10}{color1}{}{};
            \drawvoteannotated{\postwo}{1/10}{color1}{}{};
            \drawvoteannotated{\postwo}{6/10}{color1}{}{};
            \drawVOTEannotated{\postwo}{2/10}{color2}{$\mu^2_2$}{};
            \drawvoteannotated{\postwo}{3/10}{color1}{}{};
            \drawvoteannotated{\posthr}{0/10}{color1}{}{};
            \drawvoteannotated{\posthr}{4/10}{color1}{}{};
            \drawvoteannotated{\posthr}{7/10}{color1}{}{};
            \drawVOTEannotated{\posthr}{1/10}{color2}{$\mu^2_3$}{};
        \end{tikzpicture}
    \end{subfigure}\hfill
    \begin{subfigure}[t]{0.48\textwidth}
        \begin{tikzpicture}[yscale=2.5,xscale=6]
            \draw (0,0) -- (0,1);  
            \draw (-.02,0) -- (0.02,0); 
            \draw (-.02,1) -- (0.02,1); 
            \node[anchor=east, xshift=-5px] at (0,0) {$0$}; 
            \foreach \x in {1,...,9}{\draw (-.01,\x/10) -- (0.01,\x/10);}
            \node[anchor=east, xshift=-5px] at (0,1) {$1$}; 
            \node[anchor=south, rotate=90] at (0,0.5) {\small Budget}; 
            \filldraw[fill=alternativebarcolor,draw=none] (\posone-\coordwidth,0) rectangle (\posone+\coordwidth,1); 
            \filldraw[fill=alternativebarcolor,draw=none] (\postwo-\coordwidth,0) rectangle (\postwo+\coordwidth,1); 
            \filldraw[fill=alternativebarcolor,draw=none] (\posthr-\coordwidth,0) rectangle (\posthr+\coordwidth,1); 
            \node[anchor=north, outer sep=5pt] (label1) at (\posone,0) {1};
            \node[anchor=north, outer sep=5pt] at (\postwo,0) {2};
            \node[anchor=north, outer sep=5pt] at (\posthr,0) {3};
            \node[anchor=east, outer sep=0pt] at (label1.west) {\small Alternative:};
            \drawarea{\posone}{0}{6/10}{color2}{45};
            \drawarea{\postwo}{0}{3/10}{color2}{45};
            \drawarea{\posthr}{0}{4/10}{color2}{45};
            \drawaggregate{\posone}{0}{2/10}{darkgray}{45};
            \drawaggregate{\postwo}{0}{2/10}{darkgray}{45};
            \drawaggregate{\posthr}{0}{1/10}{darkgray}{45};
            \drawvoteannotated{\posone}{9/10}{color1}{}{};
            \drawvoteannotated{\posone}{0/10}{color1}{}{};
            \drawvoteannotated{\posone}{1/10}{color1}{}{};
            \drawVOTEannotated{\posone}{6/10}{color2}{$\mu^3_1$}{};
            \drawvoteannotated{\postwo}{1/10}{color1}{}{};
            \drawvoteannotated{\postwo}{6/10}{color1}{}{};
            \drawvoteannotated{\postwo}{2/10}{color1}{}{};
            \drawVOTEannotated{\postwo}{3/10}{color2}{$\mu^3_2$}{};
            \drawvoteannotated{\posthr}{0/10}{color1}{}{};
            \drawVOTEannotated{\posthr}{4/10}{color2}{$\mu^3_3$}{};
            \drawvoteannotated{\posthr}{7/10}{color1}{}{};
            \drawvoteannotated{\posthr}{1/10}{color1}{}{};
        \end{tikzpicture}
    \end{subfigure}
    \caption{Illustration of the proportional spending property on an instance with $n = 4$ voters and $m = 3$ alternatives. For $k = 2$ (left figure), we have $\frac{4}{10} = \sum_{j \in [m]} \mu^k_j < \frac{n-k+1}{n} = \frac{3}{4}$, thus proportional spending requires that $a_j \ge \mu^k_j$ for each alternative $j$. For $k = 3$ (right figure), we have $\frac{13}{10} = \sum_{j \in [m]} \mu^k_j > \frac{n-k+1}{n} = \frac{1}{2}$, thus proportional spending requires that an allocation spends a total of at least $\frac{1}{2}$ below $\mu^3$.}
    \label{fig:proportionalSpending}
\end{figure}

These two properties together guarantee the best-possible welfare approximation with respect to $n$ constrained to single-minded proportionality, as stated in the following theorem.

\begin{theorem}\label{thm:optimal_social_welfare_approximation}
    Any \changed{upper-}range-respecting mechanism satisfying proportional spending attains a welfare approximation of $\alpha^\star(n)$.
\end{theorem}

To prove this theorem, we introduce one more lemma about the welfare-maximizing aggregate returned by \Ut. The lemma states that \Ut satisfies a stronger version of proportional spending: It always allocates budget to intervals approved by a high number of voters, as long as the total size of these intervals is at most 1.

\begin{lemma} \label{lem:Util_lower_bound}
    For any given instance, \Ut returns an aggregate $a^\Uts$ with the following property for any $k \in [n]$:
    $$ \overlap{a^\Uts}{\mu^k} = \min\Bigg( 1, \sum_{j \in [m]} \mu^k_j \Bigg).$$
\end{lemma}

\proof
    Let $k \in [n]$. It is easy to see that, when 
    phantom $f_{k-1}$
    reaches $1$, the median on each alternative $j$ is exactly $\mu^k_j$. 
    Thus, if $\sum_{j \in [m]} \mu^k_j < 1$, then normalization has not been reached and $a^\Uts_j \ge \mu^k_j$ for all $j \in [m]$. On the other hand, if $\sum_{j \in [m]} \mu^k_j \ge 1$, then normalization is attained before or when 
    $f_{k-1}$
    reaches $1$ and thus $a^\Uts_j \le \mu^k_j$ for all $j \in [m]$, implying $\overlap{a^\Uts}{\mu^k} = \sum_{j\in [m]} a^\Uts_j = 1$.
\endproof

Using this lemma, we can now show \Cref{thm:optimal_social_welfare_approximation}.

\begin{proof}[Proof of \Cref{thm:optimal_social_welfare_approximation}]
    Let $m,n\in \N$ and consider a profile $P\in \mathcal{P}_{n,m}$ and an \changed{upper-}range-respecting mechanism $\A$ satisfying proportional spending.
    We denote its aggregate by~$a=\A(P)$.
    We prove the theorem by first showing a lower bound on the welfare achieved by such a mechanism and then giving an upper bound on the optimal welfare achieved by \Ut. 
    
    For aggregate $a$, we can compute the welfare as $$w(a) = \sum_{i \in [n]} \sum_{j \in [m]} \min(a_j, p_{i,j}) = \sum_{k \in [n]} \sum_{j \in [m]} \min(a_j, \mu^k_j).$$
    Let $k^\star\in [n]$ be the minimum number such that $\sum_{j \in [m]} \mu^{k^\star}_j \ge \frac{n-k^\star+1}{n}$, and let $\ell=n-k^\star+1$. This is well-defined, since by the normalization of the votes we have $\sum_{j \in [m]} \mu^n_j \ge 1 \ge \frac{1}{n}$. By the proportional spending property, we know that for each $k < k^\star$, it holds that $a_j \ge \mu^k_j$ for all $j \in [m]$. For $k \ge k^\star$, we get that $\sum_{j\in [m]} \min(a_j, \mu^k_j) \ge \sum_{j\in [m]} \min(a_j, \mu^{k^\star}_j) \ge \frac{\ell}{n}$. We can thus bound the welfare of $a$ in the following way:

    \begin{align*}
        w(a) &= \sum_{k \in [n]} \sum_{j \in [m]} \min(a_j, \mu^k_j) \\
        &= \sum_{k = 1}^{k^\star-1} \sum_{j \in [m]} \min(a_j, \mu^k_j) + \sum_{k = k^\star}^{n-1} \sum_{j \in [m]} \min(a_j, \mu^k_j) + \sum_{j \in [m]} \min(a_j, \mu^n_j) \\
        &\ge \sum_{k = 1}^{k^\star-1} \sum_{j \in [m]} \mu^k_j + \sum_{k = k^\star}^{n-1} \frac{\ell}{n} + \sum_{j \in [m]} a_j  = \left(\sum_{k = 1}^{k^\star-1} \sum_{j \in [m]} \mu^k_j \right) + (n-k^\star) \frac{\ell}{n} + 1,
    \end{align*}
    where the inequality follows from the arguments above and the \changed{upper }range respect of $\A$.
    
    By partitioning the welfare in a similar way for $\Ut$, we obtain
    \begin{align*}
        w(a^\Uts) &= \sum_{k \in [n]} \sum_{j \in [m]} \min(a^\Uts_j, \mu^k_j) = \sum_{k = 1}^{k^\star-1} \sum_{j \in [m]} \min(a^\Uts_j, \mu^k_j) + \sum_{k = k^\star}^{n} \sum_{j \in [m]} \min(a^\Uts_j, \mu^k_j) \\
        & \le \left( \sum_{k = 1}^{k^\star-1} \sum_{j \in [m]} \mu^k_j \right) + n - k^\star + 1.
    \end{align*}
     Here, the inequality follows by applying \Cref{lem:Util_lower_bound} for $k < k^\star$ and the normalization of $a^\Uts$.
     
     Combining the two bounds above, we can bound the welfare approximation as
     \begin{align*}
        \frac{w(a^\Uts)}{w(a)} 
        & \le \frac{\left(\sum_{k = 1}^{k^\star-1} \sum_{j \in [m]} \mu^k_j \right) + n - k^\star + 1}{\left( \sum_{k = 1}^{k^\star-1} \sum_{j \in [m]} \mu^k_j \right) + (n-k^\star) \frac{\ell}{n} + 1}
        \le \frac{n - k^\star + 1}{(n-k^\star) \frac{\ell}{n} + 1}\\
        & = \frac{n (n - k^\star + 1)}{(n-k^\star) \ell + n}
        \leq \max_{\ell\in [n]} \frac{n\ell}{n+\ell(\ell-1)},
     \end{align*}
    where the last inequality holds because $k^\star\in [n]$, and thus $\ell\in [n]$.
    Since \Ut returns the welfare-maximizing allocation~\citep{freeman2021truthful} and the last expression is exactly $\alpha^\star(n)$, this finishes the proof.
\end{proof}

We now show that \UtP satisfies proportional spending.

\begin{lemma} \label{lem:UtilProp_lower_bound}
    \UtP satisfies proportional spending.
\end{lemma}

\proof
    Let $P \in \mathcal{P}_{n,m}$ be a profile. Let $t$ be a time of normalization of \UtP and let $k^\star \in [n]_0$ be the highest index such that $f_{k^\star}(t) > 0$. We observe for each $1 \le k \le k^\star$ that 
    the median on each alternative $j$ is at least $\min\big(\frac{n-k+1}{n}, \mu^k_j\big)$. This is because we have $f_{k-1}(t) = \frac{n-k+1}{n}$ and therefore there are at most $n-k+1$ phantoms and at most $k-1$ voters strictly below $\min\big(\frac{n-k+1}{n}, \mu^k_j\big)$.
    Thus, for any $k \le k^\star$ we have $a_j \ge \min\big(\frac{n-k+1}{n}, \mu^k_j\big)$ for all $j \in [m]$, implying
    $$
        \overlap{a}{\mu^k} \ge \sum_{j\in [m]} \min \left(\frac{n-k+1}{n}, \mu^k_j \right) \ge  \min\Bigg( \frac{n-k+1}{n}, \sum_{j \in [m]} \mu^k_j \Bigg)
    .$$
    For the final inequality, note that if $\mu^k_j>\frac{n-k+1}{n}$ for some $j \in [m]$, then $\sum_{j\in [m]} \min \big(\frac{n-k+1}{n}, \mu^k_j \big) \ge \frac{n-k+1}{n}$. Otherwise, $\sum_{j\in [m]} \min \big(\frac{n-k+1}{n}, \mu^k_j \big) = \sum_{j \in [m]} \mu^k_j$.

    For $k > k^\star$, we know that $f_k(t) = 0$ and therefore the median on each alternative can be at most $\mu^k_j$ (there are at most $k$ phantoms and at most $n-k$ voters strictly above $\mu^k_j$).
    Thus, $\overlap{a}{\mu^k} = \sum_{j\in [m]} \min(a_j, \mu^k_j) = \sum_{j\in [m]} a_j = 1$ for each $k > k^\star$.
\endproof

\citet{freeman2021truthful} showed that any moving-phantom mechanism that first moves the first phantom to position $1$ while keeping all others at $0$, and keeps the last phantom at $0$ at all times, is \changed{upper and lower} range respecting. Since \UtP satisfies this property, \UtP is \changed{upper} range respecting.
We obtain the following corollary by combining this observation, \Cref{thm:optimal_social_welfare_approximation,lem:UtilProp_lower_bound}.

\begin{corollary} \label{cor:util_prop_welfare_approximation_n}
    \UtP attains a welfare approximation of $\alpha^\star(n)$.
\end{corollary}

Finally, we also show that \GrdDec satisfies proportional spending.

\begin{lemma} \label{lem:GrdDec_lower_bound}
    \GrdDec satisfies proportional spending.
\end{lemma}

\proof
    We first observe that, after the $k^\text{th}$ (outer) iteration of \GrdDec, it must either be the case that $N_j^+ = \emptyset$ or that $N_j^+ \neq \emptyset$ and $\Tilde{a}_j = \mu^k_j$ for each alternative $j$, since the only reason \GrdDec does not increase $\Tilde{a}_j$ to $\mu^k_j$ is if all voters voting above $\Tilde{a}_j$ have run out of budget. Also, in either case we have $\Tilde{a}_j \le \mu^k_j$.

    Now suppose the lemma is false and let $P \in \mathcal{P}_{n,m}$ be a profile for which there exists a $k \in [n]$ with $\overlap{a}{\mu^k} < \min\big( \frac{n-k+1}{n}, \sum_{j \in [m]} \mu^k_j \big)$. Then there must be some alternative $j$ with $a_j < \mu^k_j$. Together with the above observation it follows that $N_j^+$ is empty after the $k^\text{th}$ iteration of \GrdDec.
    By definition, there are at least $n-k+1$ voters voting above or equal to $\mu^k_j$ (and thus strictly above $\Tilde{a}_j \le a_j$). If $N_j^+$ is empty, then these $n-k+1$ voters must have zero budget, which implies that they have spent their total budget of $\frac{n-k+1}{n}$ in previous iterations, i.e., $$\overlap{a}{\mu^k} \ge \overlap{\Tilde{a}}{\mu^k} = \sum_{j\in [m]} \min(\Tilde{a}_j, \mu^k_j) = \sum_{j \in [m]} \Tilde{a}_j \ge \frac{n-k+1}{n},$$ contradicting the initial assumption.
\endproof

Note that decomposability implies that a mechanism can never spend more than the highest vote on an alternative. 
Thus, \GrdDec is also \changed{upper-}range-respecting, and we conclude the following corollary from \Cref{thm:optimal_social_welfare_approximation,lem:GrdDec_lower_bound}.

\begin{corollary} \label{cor:greedy_decomp_welfare_approximation_n}
    \GrdDec attains a welfare approximation of $\alpha^\star(n)$. 
\end{corollary}

Together with the lower bound from~\Cref{prop:proportional_social_welfare_lower_bound_n}, the matching upper bounds in \Cref{cor:util_prop_welfare_approximation_n} and \Cref{cor:greedy_decomp_welfare_approximation_n} show that the price of single-minded proportionality \changed{with respect to $n$} is exactly $\alpha^\star(n)$. 
Notably, this price does not increase even under the additional requirement of truthfulness (\Cref{cor:util_prop_welfare_approximation_n}) and under the stronger fairness notion of decomposability (\Cref{cor:greedy_decomp_welfare_approximation_n}).

\subsubsection{Upper Bound with Respect to the Number of Alternatives}\label{subsec:ub-m}

\changed{To show that the price of single-minded proportionality in terms of $m$ is exactly $\beta^\star$, we use a similar approach and identify a strengthening of upper range respect that, together with proportional spending, implies a $\beta^\star(m)$-approximation of the optimal welfare. Intuitively, it requires that a mechanism never spends more than $\frac{k}{n}$ on an alternative, where $k$ is the number of voters supporting that spending.
Formally, a mechanism $\A$ satisfies \emph{proportional upper bounds} if, for every profile $P \in \mathcal{P}_{n,m}$ and alternative $j \in [m]$, we have 
\[
    \A(P)_j \le \frac{\lvert \{i \in [n] \mid p_{i,j} \ge \A(P)_j \} \rvert}{n}.
\]
}

\changed{
We can then show that both \UtP and \GrdDec satisfy proportional upper bounds (\Cref{prop:prop_phantom_prop_upper_bounds,prop:decomposability_prop_upper_bounds}) and conclude the optimal worst-case approximation from the following theorem. The proofs are deferred to Appendix~\ref{app:welfare_m_approximations}.
\begin{restatable}{theorem}{optmapprox}\label{thm:optimal_social_welfare_approximation_m}
    Any mechanism satisfying proportional spending and proportional upper bounds attains a welfare approximation of $\beta^\star(m)$.
\end{restatable}
\begin{corollary} \label{cor:util_prop_greedy_decomp_welfare_approximation_m}
    \UtP and \GrdDec attain welfare approximations of $\beta^\star(m)$. 
\end{corollary}
}

\changed{
The lower bound from~\Cref{prop:proportional_social_welfare_lower_bound_m}, together with the matching upper bound in \Cref{cor:util_prop_greedy_decomp_welfare_approximation_m} show that the price of single-minded proportionality with respect to $m$ is exactly $\beta^\star(m)$ and does not increase when additionally requiring truthfulness or decomposability.
}

\section{Social Welfare of Truthful and Single-Minded Proportional Mechanisms} \label{sec:phantom_domination}

In this section, we define a welfare dominance relation among mechanisms and characterize the relationship of eight moving-phantom mechanisms. In particular, we show that \UtP is instance-wise welfare-optimal among all single-minded proportional moving-phantom mechanisms, and thereby among all known single-minded proportional and truthful mechanisms.

The definition of our relation $\trianglerighteq$ is simple: For two mechanisms $\A$ and $\mathcal{B}$, $\A \trianglerighteq \mathcal{B}$ indicates that, for every input profile, $\A$ attains at least as much social welfare as $\mathcal{B}$. Moreover, we refer to a mechanism $\A$ that has $\smash{\A\trianglerighteq \mathcal{B}}$ (resp.\ $\smash{\mathcal{B}\trianglerighteq \A}$) for every $\mathcal{B}$ within a certain class of mechanisms as a \textit{welfare-maximizing} (resp.\ \textit{welfare-minimizing}) mechanism within this class.
Although not all moving-phantom mechanisms are comparable under $\trianglerighteq$, \Cref{lem:phantom_domination} enables us to order almost all moving-phantom mechanisms proposed in the literature, with only a single pair being incomparable.
\Cref{thm:phantom-domination} states these comparisons; \Cref{fig:phantom_domination} provides a visual aid.

Before stating the theorem, we define two additional moving-phantom mechanisms that we show to be welfare-minimizing within certain classes. First, \Constant is the trivial moving-phantom mechanism whose output is the uniform aggregate ($\frac{1}{m}$ for every alternative) regardless of the profile. Second, \Fan is the moving-phantom mechanism induced by the phantom system that moves all phantoms together from $0$ to $1$, but each phantom $f_k$ stops moving once it reaches its final position $\frac{n-k}{n}$. Again, we refer to Appendix~\ref{app:mp-mechs} for formal definitions and illustrations. 

\bigskip 

\begin{restatable}{theorem}{thmPhantomDomination}\label{thm:phantom-domination}
    The following relations hold:
    \begin{enumerate}[label=(\roman*)]
        \item \Ut is welfare-maximizing and \Constant is welfare-minimizing among moving-phantom mechanisms.\label{item:welf-all}
        \item \UtP is welfare-maximizing and \Fan is welfare-minimizing among single-minded proportional moving-phantom mechanisms.\label{item:welf-smp}
        \item $\PiecewiseUniform\trianglerighteq \IM$ and $\Ladder\trianglerighteq \IM$, but \PiecewiseUniform and \Ladder are incomparable according to $\trianglerighteq$.\label{item:welf-pwu-lad}
        \item $\Fan \trianglerighteq \GreedyMax$.\label{item:welf-fan}
    \end{enumerate}
\end{restatable}

\begin{figure}[h]
    \centering
    \small
    \begin{tikzpicture}[
        mech/.style={draw=black,circle,fill,inner sep=2pt},
        every label/.append style={font=\footnotesize, fill=none, text opacity=1, inner sep=2pt},
        every node/.style={inner sep=0},
        dom/.style={-Stealth, shorten >=3pt, shorten <=3pt},
    ]
        \def\xdist{0.13}
        \def\ydist{0.05}
    
        \node (left) at (0,0) {};
        \node (top) at (50*\xdist, 42*\ydist) {};
        \node (bottom) at (50*\xdist, -40*\ydist) {};
        \node (right) at (100*\xdist, 0) {};
        \fill [fill=color1!30!white] (50*\xdist, 0) ellipse ({50*\xdist} and {50*\ydist});
    
        \node (leftinner) at (15*\xdist,0) {};
        \node (topinner) at (42.5*\xdist, 27*\ydist) {};
        \node (bottominner) at (42.5*\xdist, -25*\ydist) {};
        \node (rightinner) at (70*\xdist, 0) {};
        \fill [fill=color2!30!white] (42.5*\xdist, 0) ellipse ({27.5*\xdist} and {35*\ydist});
    
        \node[color=color1!30!black] (phantoms) at (top) {moving-phantom mechanisms};
        \node[color=color2!60!black,  align=center] (proportional) at (topinner) {single-minded proportional\\[-4pt]moving-phantom mechanisms};
        
        \node[mech, color=color1!90!black ,label={[anchor=south]above:\Ut}] (util) at (left) {};
        \node[mech, color=color2!99!black ,label={[anchor=north]below:\UtP}] (utilprop) at (leftinner) {};
        \node[mech, color=color2!99!black ,label={[anchor=south]above:\PiecewiseUniform}] (piecewise) at (35*\xdist, 8*\ydist) {};
        \node[mech, color=color2!99!black ,label={[anchor=north]below:\Ladder}] (ladder) at (35*\xdist, -8*\ydist) {};
        \node[mech, color=color2!99!black ,label={[inner sep=4pt, anchor=north]below:\IM}] (im) at (55*\xdist, 0) {};
        \node[mech, color=color2!99!black ,label={[anchor=south]above:\Fan}] (greedymaxprop) at (rightinner) {};
        \node[mech, color=color1!90!black ,label={[anchor=north]below:\GreedyMax}] (greedymax) at (85*\xdist, 0) {};
        \node[mech, color=color1!90!black ,label={[anchor=south]above:\Constant}] (constant) at (right) {};
        
        \draw[dom] (util) to (utilprop);
        \draw[dom] (utilprop) to (piecewise);
        \draw[dom] (utilprop) to (ladder);
        \draw[dom] (piecewise) to (im);
        \draw[dom] (ladder) to (im);
        \draw[dom] (im) to (greedymaxprop);
        \draw[dom] (greedymaxprop) to (greedymax);
        \draw[dom] (greedymax) to (constant);
        
    \end{tikzpicture}
    \caption{Welfare domination relationships among known moving-phantom mechanisms. The mechanisms \UtP, \Fan, and the trivial \Constant mechanism are defined in this paper.}
    \label{fig:phantom_domination}
\end{figure}
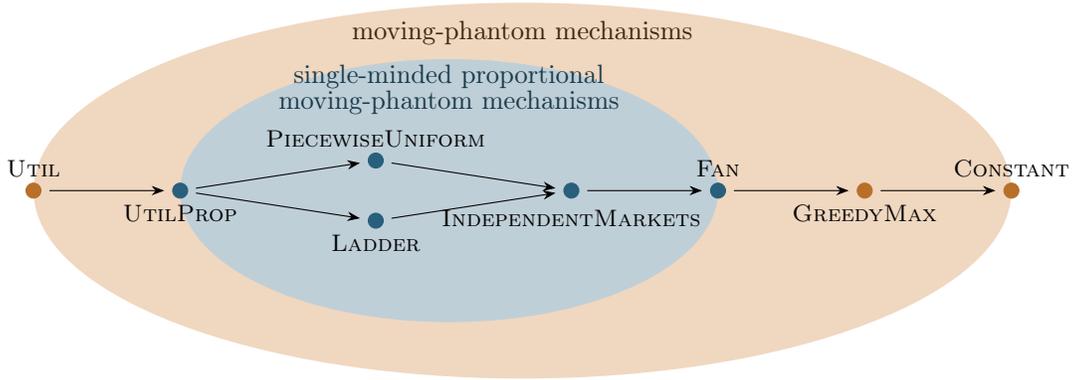

The rest of this section is devoted to proving \Cref{thm:phantom-domination}. The main technical ingredient to compare the welfare that two moving-phantom mechanisms provide for a particular instance is given by the following structural lemma: 
If, at the time of normalization, the phantom positions of one system can be obtained from those of another by shifting all phantoms above some threshold upward and all phantoms below the threshold downward, then the former system yields weakly higher social welfare.

\begin{restatable}{lemma}{lemPhantomDomination} \label{lem:phantom_domination}
    Let $n \in \N$ be fixed, $P = (p_1, \dots, p_n)$ be a profile and $f_0, \dots, f_n \in [0,1]$ and $g_0, \dots, g_n \in [0,1]$ the positions of the phantoms of two phantom systems $\F_n$ and $\G_n$ at a time of normalization (respectively) for $P$, i.e., $$\sum_{j \in [m]} \med(f_0, \dots, f_n, \p_{1,j}, \dots, \p_{n,j}) = \sum_{j \in [m]} \med(g_0, \dots, g_n, \p_{1,j}, \dots, \p_{n,j}) = 1.$$
    Let $a^{\F}$ and $a^{\G}$ be the output allocations of the corresponding moving-phantom mechanisms.
    If there exists an index $k^\star \in \{-1,0,\ldots,n\}$ such that $f_i \ge g_i$ for all $i \in \{0, \dots, k^\star\}$ and $f_i \le g_i$ for all $i \in \{k^\star+1, \dots, n\}$, then $w(a^{\F}) \ge w(a^{\G})$.
\end{restatable}

The key intuition for the proof of this lemma is the following: The phantom placement of $\mathcal{G}$ can be transformed into the phantom placement of $\mathcal{F}$ by moving up those phantoms with index $i \leq k^\star$ and then moving down those phantoms with index $i>k^\star$. The first movement increases the budget spent on intervals approved by at least $n-k^\star$ voters, the second movement decreases the budget spent on intervals approved by at most $n-k^\star-1$ voters.

\begin{proof}[Proof of \Cref{lem:phantom_domination}]
    We show that the increase in welfare due to the alternatives $j$ receiving more from $\F$ than from $\G$ (i.e., $a^{\F}_j - a^{\G}_j > 0$) is larger than the decrease from alternatives receiving less. If $k^\star = - 1$, then all phantoms of $\F_n$ are less than or equal to the phantoms from $\G_n$, implying that each median must be lower or equal as well. Due to normalization, we then get $a^{\F} = a^{\G}$. The same holds for $k^\star = n$, thus we assume $0 \le k^\star \le n-1$ from this point.
    
    As both $a^{\F}$ and $a^{\G}$ are normalized, we have $\sum_{j \in [m]} (a^{\F}_j - a^{\G}_j) = 0$. Let $\smash{M^+ = \{j \in [m] \mid a^{\F}_j > a^{\G}_j\}}$ and $\smash{M^- = \{j \in [m] \mid a^{\F}_j < a^{\G}_j\}}$ denote the alternatives with a positive and a negative difference, respectively.
    \changed{If both of these sets are empty, the result holds trivially, so we assume in what follows that this is not the case.}
    We claim that, on any alternative $j \in M^+$, there are at least $n - k^\star$ voters with a vote larger than or equal to $\smash{a^\F_j}$.

    \textit{Proof of claim:}
    Suppose there are $k < n - k^\star$ such voters. Then, by the definition of the median, there must be at least $n+1-k > k^\star + 1$ phantoms $f_i$ at positions larger than or equal to $\smash{a^\F_j}$, thus $\smash{f_{n-k} \ge a^\F_j}$. But since $n-k > k^\star$, we also know that $f_{n-k} \le g_{n-k}$ and thus $\smash{g_{n-k} \ge a^\F_j}$. Therefore, there are $k$ voters and at least $n-k+1$ phantoms $g_i$ at positions larger than or equal to $\smash{a^\F_j}$, contradicting the assumption $\smash{a^\F_j > a^\G_j}$. \hfill $\blacksquare$
    
    \medskip 

    Similarly, we claim that on each alternative $j \in M^-$, there are at most $n - k^\star -1$ voters with a vote strictly larger than $\smash{a^\F_j}$.
    
    \textit{Proof of claim:} Suppose there are $k \geq n - k^\star$ such voters. Then by the definition of the median, there can be only $n-k \leq k^\star $ phantoms $f_i$ at positions strictly larger than $\smash{a^\F_j}$, thus $\smash{f_{n-k} \le a^\F_j}$. But since $n-k \leq k^\star$, we also know that $f_{n-k} \ge g_{n-k}$ and thus $\smash{g_{n-k} \le a^\F_j}$. Therefore, there are $n-k$ voters and at least $k+1$ phantoms $g_i$ at positions lower than or equal to $\smash{a^\F_j}$, contradicting the assumption $\smash{a^\F_j < a^\G_j}$. \hfill $\blacksquare$

    Combining the previous claims, we obtain
    \begin{align*}
        w(a^{\F}) & - w(a^{\G})\\
        = {} & \sum_{i \in [n]} \sum_{j \in [m]} \min(p_{i,j},a^\F_j) - \sum_{i \in [n]} \sum_{j \in [m]} \min(p_{i,j},a^\G_j)\\
        = {} & \sum_{j \in M^+} \underbrace{\sum_{i \in [n]} (\min(p_{i,j},a^\F_j)-\min(p_{i,j},a^\G_j))}_{\ge (n - k^\star)(a^\F_j-a^\G_j)} - \sum_{j \in M^-} \underbrace{\sum_{i \in [n]} (\min(p_{i,j},a^\G_j)-\min(p_{i,j},a^\F_j))}_{< (n - k^\star)(a^\G_j-a^\F_j)}\\
        > {} & (n - k^\star) \underbrace{\left(\sum_{j \in M^+} (a^\F_j-a^\G_j) - \sum_{j \in M^-} (a^\G_j-a^\F_j)\right)}_{=0} = 0,
    \end{align*}
    which concludes the proof. 
\end{proof}

\smallskip

For the relations stated in \Cref{thm:phantom-domination}, it suffices to show that we can apply \Cref{lem:phantom_domination} by showing that its condition is satisfied for any two points in time for the two corresponding phantom systems. 
We defer most of these proofs to Appendix~\ref{app:thmPhantomDomination}, where we also show the incomparability of \PiecewiseUniform and \Ladder by presenting two instances such that each mechanism strictly outperforms the other in terms of welfare in one of them.

The fact that \UtP is welfare-maximizing among single-minded proportional moving-phantom mechanisms, stated in \cref{item:welf-smp}, is particularly interesting and requires a deeper understanding of these mechanisms. 
To this end, we show that a condition that \citet{freeman2021truthful} showed to be sufficient for a moving-phantom mechanism to be single-minded proportional is also necessary:
A moving-phantom mechanism is single-minded proportional if and only if it is induced by a phantom system such that each phantom $k$ is at position $\frac{n-k}{n}$ at time $t = 1$. 
The proof, deferred to Appendix~\ref{app:lemProportionalPhantoms}, works by modifying the phantom system of an arbitrary single-minded proportional moving-phantom mechanism, in such a way that no phantom $k$ moves beyond position $\frac{n-k}{n}$. We show that this modification does not change the output for any instance. 

\begin{restatable}{lemma}{lemProportionalPhantoms}\label{lem:proportional_phantoms_characterization}
    A moving-phantom mechanism $\A$ is single-minded proportional if and only if there is a family of phantom systems, $\F_n = \{f_k \mid k \in [n]_0\}$ for each $n \in \N$, that produces the same outcome as $\A$ for every instance, with the property that for all $n \in \N$ we have $f_k(1) = \frac{n-k}{n}$ for each $k \in [n]_0$.
\end{restatable}

Using \Cref{lem:phantom_domination,lem:proportional_phantoms_characterization}, we can now prove that \UtP is welfare-maximizing among single-minded proportional moving-phantom mechanisms.

\begin{proof}[Proof of \Cref{thm:phantom-domination}.\ref{item:welf-smp} (first part)]

Let $n$ be fixed and $\F_n$ be the phantom system of \UtP. By \Cref{lem:proportional_phantoms_characterization}, for any single-minded proportional moving-phantom mechanism $\A$, we can find a phantom system $\G_n$ inducing it such that for any phantom $k$ of $\G_n$, we have $g_k(1) = 1 - \frac{k}{n}$.
    We show that for any two times $t, t' \in [0,1]$, we can find a threshold index $k^\star\in \{-1,\ldots,n\}$, such that all phantoms with index $k \le k^\star$ from $\F_n$ lie above the ones of $\G_n$ at times $t$ and $t'$, respectively, and all the phantoms with index $k > k^\star$ lie below. 
    
    Let $k'$ be the index of the lowest non-zero phantom of $\F_n$ at time $t$. Note that, by the definition of \UtP, we have $f_{k}(t) = 1 - \frac{k}{n} \ge g_{k}(t')$ for all $k < k'$ and $f_{k}(t) = 0 \le g_{k}(t')$ for all $k > k'$. Now, we choose $k^\star = k'$ if $f_{k'}(t) \ge g_{k'}(t')$ and $k^\star = k'-1$ if $f_{k'}(t) < g_{k'}(t')$. In both cases, $k^\star$ satisfies the conditions from \Cref{lem:phantom_domination}. Applying the lemma yields that \UtP always achieves at least the welfare that $\A$ achieves, completing the proof.
\end{proof}

\section{Social Welfare of Decomposable Mechanisms} \label{sec:optimal_decomposable}

In the previous sections, we analyzed the social welfare of single-minded proportional mechanisms and, in particular, showed that \UtP both guarantees the best-possible approximation ratio in the worst case and maximizes the welfare among single-minded proportional moving-phantom mechanisms for every instance.
Here, we show that the situation is different for decomposable mechanisms: While \GrdDec gives the best-possible welfare approximation in the worst case, it does \textit{not} maximize welfare among decomposable mechanisms on every instance and, in fact, no computationally efficient mechanism can do so, unless $\textsf{P}=\textsf{NP}$. 

To make this statement formal, let \UtDec be a mechanism that outputs, on every instance, some decomposable allocation maximizing social welfare, using an arbitrary but fixed tie-breaking rule. To gain intuition, we begin with a simple example showing that \GrdDec 
may return an allocation with strictly lower welfare than \UtDec.\footnote{
We further show in Appendix~\ref{app:paretoGrdDecUtDec} that there are instances on which \GrdDec and \UtDec output the same allocation, and that allocation is Pareto-dominated by a non-decomposable one.}

\begin{example}\label{ex:greedydecomp_social_welfare_small}
    Consider the instance with $n = 4$ voters and $m=5$ alternatives and profile
    \[
        P = \begin{pmatrix}
            \sfrac{3}{4} & 0 & \sfrac{1}{4} & 0 & 0 \\
            0 & \sfrac{3}{4} & 0 & \sfrac{1}{4} & 0 \\
            0 & 0 & \sfrac{1}{3} & \sfrac{1}{3} & \sfrac{1}{3} \\
            0 & 0 & \sfrac{1}{3} & \sfrac{1}{3} & \sfrac{1}{3} \\
        \end{pmatrix},
    \]
    which we illustrate in \Cref{fig:greedydecomp_social_welfare_small}.

    \GrdDec first distributes the budget from voters $3$ and $4$ over the alternatives $3$ and $4$ and then has to spend the budget of voters $1$ and $2$  on alternatives $1$ and $2$  (compare \Cref{fig:greedydecomp_social_welfare_small_grdydecomp}), leading to the allocation $a = \big(\frac{1}{4}, \frac{1}{4}, \frac{1}{4}, \frac{1}{4}, 0\big)$ with welfare $w(a) = 2 \cdot \frac{1}{4} + 2 \cdot 3 \cdot \frac{1}{4} = 2$.

    \changed{On the other hand, $a' = \big(\frac{1}{12}, \frac{1}{12}, \frac{1}{4}, \frac{1}{4}, \frac{1}{3}\big)$ is a decomposable allocation, as witnessed by the following voter contributions (compare \Cref{fig:greedydecomp_social_welfare_small_opt}):}
    \[
        c = \begin{pmatrix}
            \sfrac{1}{12} & 0 & \sfrac{1}{6} & 0 & 0 \\
            0 & \sfrac{1}{12} & 0 & \sfrac{1}{6} & 0 \\
            0 & 0 & \sfrac{1}{12} & 0 & \sfrac{1}{6} \\
            0 & 0 & 0 & \sfrac{1}{12} & \sfrac{1}{6} \\
        \end{pmatrix}.
    \]
    \changed{Since $w(a') = 2 \cdot \frac{1}{12} + 2 \cdot 3 \cdot \frac{1}{4} + 2 \cdot \frac{1}{3} = \frac{7}{3}$, \UtDec achieves welfare at least $\frac{7}{3} > 2$.}
\end{example}

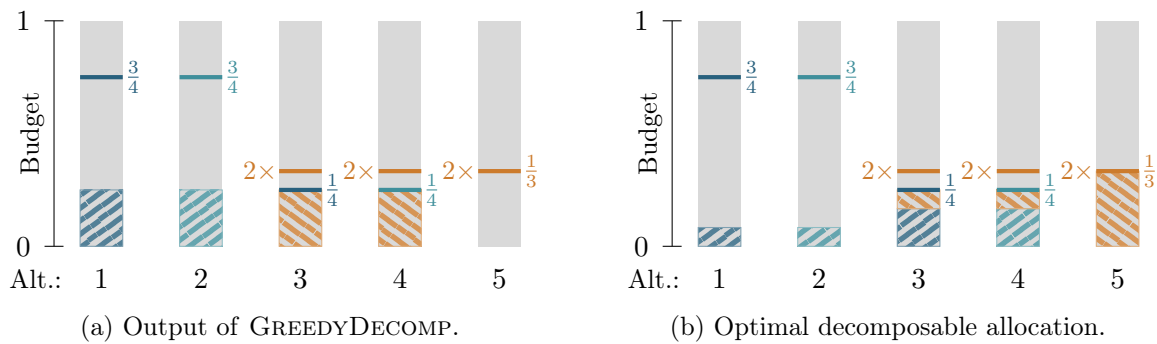
\begin{figure}
    \centering
    \def\coordwidth{0.045}
    \def\coorddist{0.20}
    \def\posone{0.1}
    \def\postwo{\posone + \coorddist}
    \def\posthr{\postwo + \coorddist}
    \def\posfou{\posthr + \coorddist}
    \def\posfiv{\posfou + \coorddist}
    \newcommand{\drawvoteannotated}[5]{
        \draw[ultra thick, #3] (#1-\coordwidth,#2) -- (#1+\coordwidth,#2)
        node[pos=0, inner sep=0pt, anchor=east] {\small #4}
        node[pos=1, inner sep=0pt, anchor=west] {\small #5};
    }
    \newcommand{\drawcontribution}[7]{
        \fill[pattern={Lines[angle=#5*0.8,distance=4pt,line width=2pt]},pattern color=#4!80!white, draw=#4!60!white] (#1-\coordwidth+#6*\coordwidth*2,#2) rectangle (#1-\coordwidth+#7*\coordwidth*2,#3);
    }
    \begin{subfigure}[t]{.49\textwidth}
        \begin{tikzpicture}[yscale=3,xscale=6]
            \draw (0,0) -- (0,1);  
            \draw (-.02,0) -- (0.02,0); 
            \draw (-.02,1) -- (0.02,1); 
            \node[anchor=east, xshift=-5px] at (0,0) {$0$}; 
            \node[anchor=east, xshift=-5px] at (0,1) {$1$}; 
            \node[anchor=south, rotate=90] at (0,0.5) {\small Budget}; 
            \filldraw[fill=alternativebarcolor,draw=none] (\posone-\coordwidth,0) rectangle (\posone+\coordwidth,1); 
            \filldraw[fill=alternativebarcolor,draw=none] (\postwo-\coordwidth,0) rectangle (\postwo+\coordwidth,1); 
            \filldraw[fill=alternativebarcolor,draw=none] (\posthr-\coordwidth,0) rectangle (\posthr+\coordwidth,1); 
            \filldraw[fill=alternativebarcolor,draw=none] (\posfou-\coordwidth,0) rectangle (\posfou+\coordwidth,1); 
            \filldraw[fill=alternativebarcolor,draw=none] (\posfiv-\coordwidth,0) rectangle (\posfiv+\coordwidth,1); 
            \node[anchor=north, outer sep=5pt] (label1) at (\posone,0) {$1$};
            \node[anchor=north, outer sep=5pt] at (\postwo,0) {$2$};
            \node[anchor=north, outer sep=5pt] at (\posthr,0) {$3$};
            \node[anchor=north, outer sep=5pt] at (\posfou,0) {$4$};
            \node[anchor=north, outer sep=5pt] at (\posfiv,0) {$5$};
            \node[anchor=east, outer sep=0pt] at (label1.west) {\small Alt.:};
            \drawcontribution{\posone}{0}{1/4}{color2}{45}{0}{1};
            \drawcontribution{\postwo}{0}{1/4}{color3}{45}{0}{1};
            \drawcontribution{\posthr}{0}{1/4}{color1}{-45}{0}{1};
            \drawcontribution{\posfou}{0}{1/4}{color1}{-45}{0}{1};
            \drawvoteannotated{\posone}{3/4}{color2}{}{$\frac{3}{4}$};
            \drawvoteannotated{\postwo}{3/4}{color3}{}{$\frac{3}{4}$};
            \drawvoteannotated{\posthr}{1/3}{color1}{$2\times$}{};
            \drawvoteannotated{\posthr}{1/4}{color2}{}{$\frac{1}{4}$};
            \drawvoteannotated{\posfou}{1/3}{color1}{$2\times$}{};
            \drawvoteannotated{\posfou}{1/4}{color3}{}{$\frac{1}{4}$};
            \drawvoteannotated{\posfiv}{1/3}{color1}{$2\times$}{$\frac{1}{3}$};
        \end{tikzpicture}
        \caption{Output of \GrdDec.} \label{fig:greedydecomp_social_welfare_small_grdydecomp}
    \end{subfigure}
    \begin{subfigure}[t]{.49\textwidth}
        \begin{tikzpicture}[yscale=3,xscale=6]
            \draw (0,0) -- (0,1);  
            \draw (-.02,0) -- (0.02,0); 
            \draw (-.02,1) -- (0.02,1); 
            \node[anchor=east, xshift=-5px] at (0,0) {$0$}; 
            \node[anchor=east, xshift=-5px] at (0,1) {$1$}; 
            \node[anchor=south, rotate=90] at (0,0.5) {\small Budget}; 
            \filldraw[fill=alternativebarcolor,draw=none] (\posone-\coordwidth,0) rectangle (\posone+\coordwidth,1); 
            \filldraw[fill=alternativebarcolor,draw=none] (\postwo-\coordwidth,0) rectangle (\postwo+\coordwidth,1); 
            \filldraw[fill=alternativebarcolor,draw=none] (\posthr-\coordwidth,0) rectangle (\posthr+\coordwidth,1); 
            \filldraw[fill=alternativebarcolor,draw=none] (\posfou-\coordwidth,0) rectangle (\posfou+\coordwidth,1); 
            \filldraw[fill=alternativebarcolor,draw=none] (\posfiv-\coordwidth,0) rectangle (\posfiv+\coordwidth,1); 
            \node[anchor=north, outer sep=5pt] (label1) at (\posone,0) {$1$};
            \node[anchor=north, outer sep=5pt] at (\postwo,0) {$2$};
            \node[anchor=north, outer sep=5pt] at (\posthr,0) {$3$};
            \node[anchor=north, outer sep=5pt] at (\posfou,0) {$4$};
            \node[anchor=north, outer sep=5pt] at (\posfiv,0) {$5$};
            \node[anchor=east, outer sep=0pt] at (label1.west) {\small Alt.:};
            \drawcontribution{\posone}{0}{1/12}{color2}{45}{0}{1};
            \drawcontribution{\postwo}{0}{1/12}{color3}{45}{0}{1};
            \drawcontribution{\posthr}{0}{1/6}{color2}{45}{0}{1};
            \drawcontribution{\posthr}{1/6}{1/4}{color1}{-45}{0}{1};
            \drawcontribution{\posfou}{0}{1/6}{color3}{45}{0}{1};
            \drawcontribution{\posfou}{1/6}{1/4}{color1}{-45}{0}{1};
            \drawcontribution{\posfiv}{0}{1/3}{color1}{-45}{0}{1};
            \drawvoteannotated{\posone}{3/4}{color2}{}{$\frac{3}{4}$};
            \drawvoteannotated{\postwo}{3/4}{color3}{}{$\frac{3}{4}$};
            \drawvoteannotated{\posthr}{1/3}{color1}{$2\times$}{};
            \drawvoteannotated{\posthr}{1/4}{color2}{}{$\frac{1}{4}$};
            \drawvoteannotated{\posfou}{1/3}{color1}{$2\times$}{};
            \drawvoteannotated{\posfou}{1/4}{color3}{}{$\frac{1}{4}$};
            \drawvoteannotated{\posfiv}{1/3}{color1}{$2\times$}{$\frac{1}{3}$};
        \end{tikzpicture}
        \caption{Optimal decomposable allocation.} \label{fig:greedydecomp_social_welfare_small_opt}
    \end{subfigure}
    \caption{Visualization of the profile and aggregates from \Cref{ex:greedydecomp_social_welfare_small}. Upwards hatched areas (blue) represent contributions of voters $1$ and $2$; downwards hatched areas (orange) are contributions of voters $3$ and $4$.}
    \label{fig:greedydecomp_social_welfare_small}
\end{figure}

{ 
In the following, we prove that computing a welfare-optimal decomposable allocation (and hence the outcome of \UtDec) is not possible in polynomial time unless $\mathrm{P}=\mathrm{NP}$. Nevertheless, a welfare-optimal decomposable allocation can be found via integer linear programming; see Appendix~\ref{app:IPDecomp}. To formalize our hardness result, we introduce the following decision problem.}

\begin{definition}[\textsc{DecomposableWelfareThreshold (DWT)}]
An instance consists of a profile $P$ and a rational number $C \in \mathbb{Q}$. The decision problem asks whether there exists a decomposable allocation $a$ and voter contributions $(c_{i,j})_{i \in [n], j \in [m]}$ such that $a$ has social welfare $w(a) \ge C$.
\end{definition}

Clearly, \textsc{DWT} is contained in NP. We show that \textsc{DWT} is NP-complete, which implies that computing the outcome of \UtDec is NP-hard.
While the proof of this result is deferred to Appendix~\ref{app:thmNPcomplete}, we sketch here the main ideas involved.

\begin{restatable}{theorem}{thmNPcomplete}
\label{thm:welfare-optimal-decomposable-hardness}
    \textsc{DecomposableWelfareThreshold} is NP-complete.
    In particular, computing the outcome of \UtDec is NP-hard.
\end{restatable}

The proof reduces from \textsc{ExactCoverBy3Sets (X3C)}, which asks whether a set of elements $U$ of size $3q$ admits a set cover of size $q$ from among a family of 3-element subsets $\mathcal{S} = \{ S_1, \ldots, S_k \}$ of $U$. 
For any given instance of X3C, we include one \emph{subset alternative} for each $S_j$ and $3q$ \emph{element voters} whose preferences encode the incidence structure: for each element $e \in U$, the corresponding voter assigns a small positive vote\footnote{For technical reasons that are consistent with the intuition provided here, the full proof actually uses $\varepsilon+\frac{1}{n}$, not $\varepsilon$.} $\varepsilon$ to those $S_j$ that contain $e$ and zero to all other subset alternatives. A collection of auxiliary \emph{toggle voters} each vote highly for a single subset alternative as well as for a commonly liked alternative $Y$. Due to decomposability constraints, $Y$ can absorb the spending of exactly $q$ toggle voters. This allows exactly $q$ subset alternatives to be ``turned on'', while the remaining $k-q$ subset alternatives remain ``turned off'' because toggle voter spending crowds out potential element voter contributions. 
Hardness follows from showing that, subject to decomposability, social welfare is strictly higher on X3C instances that admit an exact cover than those that do not.

Given the hardness of computing the optimal decomposable allocation for a given instance, approximating it is, in a sense, the best we can hope for.
We finish with a positive result for \GrdDec in this regard: For every instance, it returns an allocation whose welfare is at least half of the welfare of any other decomposable allocation.

\begin{restatable}{theorem}{thmGreedyDecompTwo}\label{thm:greedy-decomp-two}
    For any profile with $n \ge 2$ voters, \UtDec achieves at most $2-\frac{1}{n-1}$ times the welfare that \GrdDec achieves. 
\end{restatable}

To prove the theorem, we compare the allocation $a^\GrdDecs$ of \GrdDec to a welfare-optimal decomposable allocation $a^\optd$. We consider voter contributions certifying the decomposability of $a^\optd$, imagine that voters spend their contributions in some arbitrary fixed order, and show that the welfare advantage of $a^\optd$ over $a^\GrdDecs$ on alternatives where $a^\optd > a^\GrdDecs$ is bounded by the welfare that $a^\GrdDecs$ achieves. The slight improvement of $-\frac{1}{n-1}$ in the approximation ratio is achieved by exploiting the upper-range-respect and proportional spending of \GrdDec.

For this proof, we introduce some additional notation. We define an \textit{allocation region}
$A\subseteq [0,1)^m$ to be a vector of subsets, where each subset $A_j$ is a disjoint union of intervals. An allocation $a$ then corresponds to the allocation region $A$ with $A_j = [0, a_j)$. We define the intersection and exclusion of two allocation regions $A$ and $A'$ as $A \cap A' = (A_j \cap A_j')_{j\in [m]}$ and $A \setminus A' = (A_j \setminus A_j')_{j\in [m]}$. Even though allocation regions that are not of the form $[0, x)$ for each alternative do not have an intuitive interpretation, they can be a useful tool for computing the welfare of an allocation. To this end, for a disjoint union of $x$ intervals $X = \bigcupdot_{k = 1}^x X_k$, we let $|X|$ denote the Lebesgue measure of $X$, which corresponds to the sum of the lengths of the sub-intervals, i.e., $|X| = \sum_{k = 1}^x |X_k|$. We then define the welfare of an allocation region $A$ as $w(A) = \sum_{j \in [m]} \sum_{k \in [n]} (n-k+1) \cdot |[\mu^{k-1}_j, \mu^k_j) \cap A_j|$, where we set $\mu^0_j = 0$ for all $j \in [m]$. Note that for a given allocation $a$ and its corresponding allocation region $A$, we indeed have
\begin{align*}
    w(A) &= \sum_{j \in [m]} \sum_{k \in [n]} (n-k+1) \cdot |[\mu^{k-1}_j, \mu^k_j) \cap A_j|
    = \sum_{j \in [m]} \sum_{k \in [n]} |[0, \mu^k_j) \cap A_j| \\
    & = \sum_{j \in [m]} \sum_{k \in [n]} \min(\mu^k_j, a_j) 
    = \sum_{i \in [n]} \sum_{j \in [m]} \min(p_{i,j}, a_j)
    = \sum_{i \in [n]} u_i(a) = w(a).
\end{align*}

\begin{figure}
    \centering
    \def\coorddist{0.5}
    \def\coordwidth{0.06}
    \newcommand{\drawaggregatearea}[5]{
        \fill[pattern={Lines[angle=#4*0.8,distance=4pt,line width=2pt]},pattern color=#3!70!white, draw=#3!60!white] (#1*\coorddist-\coordwidth,0) rectangle (#1*\coorddist+\coordwidth,#2);
    }
    \newcommand{\drawaggregatelineannotated}[5]{
        \draw[#3] (#1*\coorddist-\coordwidth,#2) -- (#1*\coorddist+\coordwidth,#2) node[anchor=west] {\small #5};
    }
    \begin{tikzpicture}[yscale=3,xscale=6]
        \draw (0,0) -- (0,1);  
        \draw (-.02,0) -- (0.02,0); 
        \draw (-.02,1) -- (0.02,1); 
        \node[anchor=east, xshift=-5px] at (0,0) {$0$}; 
        \node[anchor=east, xshift=-5px] at (0,1) {$1$}; 
        \node[anchor=south, rotate=90] at (0,0.5) {\small Budget}; 
        \filldraw[fill=alternativebarcolor,draw=none] (1*\coorddist-\coordwidth,0) rectangle (1*\coorddist+\coordwidth,1); 
        \filldraw[fill=alternativebarcolor,draw=none] (2*\coorddist-\coordwidth,0) rectangle (2*\coorddist+\coordwidth,1); 
        \filldraw[fill=alternativebarcolor,draw=none] (3*\coorddist-\coordwidth,0) rectangle (3*\coorddist+\coordwidth,1); 
        \node[anchor=north east, outer sep=4.4pt] at (1*\coorddist,0) {\small Alternative:};
        \node[anchor=north, outer sep=5pt] at (1*\coorddist,0) {$1$};
        \node[anchor=north, outer sep=5pt] at (2*\coorddist,0) {$2$};
        \node[anchor=north, outer sep=5pt] at (3*\coorddist,0) {$3$};
        \drawaggregatearea{1}{0.3}{darkgray}{-45}{$a_1$}
        \drawaggregatearea{2}{0.1}{darkgray}{-45}{$a_2$}
        \drawaggregatearea{3}{0.6}{darkgray}{-45}{$a_3$}
        \drawaggregatearea{1}{0.5}{lightgray}{45}{$a_1'$}
        \drawaggregatearea{2}{0.2}{lightgray}{45}{$a_2'$}
        \drawaggregatearea{3}{0.3}{lightgray}{45}{$a_3'$}
        \drawaggregatelineannotated{1}{0.3}{darkgray}{-45}{$a_1$}
        \drawaggregatelineannotated{2}{0.1}{darkgray}{-45}{$a_2$}
        \drawaggregatelineannotated{3}{0.6}{darkgray}{-45}{$a_3$}
        \drawaggregatelineannotated{1}{0.5}{lightgray}{45}{$a_1'$}
        \drawaggregatelineannotated{2}{0.2}{lightgray}{45}{$a_2'$}
        \drawaggregatelineannotated{3}{0.3}{lightgray}{45}{$a_3'$}
        \draw [black, very thick] (1*\coorddist-\coordwidth-0.02,0.3) -- (1*\coorddist-\coordwidth-0.02,0.5) node[midway, anchor=east] {\small $(A' \setminus A)_1$};
        \draw [black, very thick] (2*\coorddist-\coordwidth-0.02,0.1) -- (2*\coorddist-\coordwidth-0.02,0.2) node[midway, anchor=east] {\small $(A' \setminus A)_2$};
    \end{tikzpicture}
    \caption{Illustration of example allocation regions. The exclusion of allocation regions $A$ (dark gray, downwards hatching) and $A'$ (light gray, upwards hatching) corresponding to the allocations $a$ and $a'$, respectively, is given as $(A' \setminus A) = ([a_1, a_1'), [a_2, a_2'), \emptyset)$.}
    \label{fig:allocation_intervals}
\end{figure}

\begin{proof}[Proof of \Cref{thm:greedy-decomp-two}]
    Let $n,m$ be fixed, $P \in \mathcal{P}_{n,m}$ be a profile, $a^\optd$ be the output of \UtDec, and $a^\GrdDecs$ the output of \GrdDec. 
    Let $c^\optd_i$ for $i \in [n]$ be the voter contributions certifying the decomposability of $a^\optd$; i.e., $\sum_{i \in [n]} c^\optd_{ij} = a^\optd_j$ for each $j \in [m]$ and $a^\optd_j \le p_{i,j}$ whenever $c^\optd_{ij} > 0$. 

    Now, let $A^\optd$ and $A^\GrdDecs$ be the allocation regions corresponding to $a^\optd$ and  $a^\GrdDecs$, respectively.
    We 
    assign each voter $i\in [n]$ the allocation region $C^\optd_i = \big[\sum_{i' \in [i-1]} c^\optd_{ij}, \sum_{i' \in [i]} c^\optd_{ij}\big)_{j\in[m]}$. Intuitively, this corresponds to letting the voters add their contribution to the allocation $a^\optd$ in order. We also construct the allocation regions $C^\GrdDecs_i$ for $i \in [n]$ by tracking the order in which the voters spend their budget during the execution of $\GrdDec$. We initialize all voters' regions as $C^\GrdDecs_i = (\emptyset)_{j \in [m]}$. Whenever a voter's payment gets deducted from their budget in \Cref{alg_line:update_voter_budget} of \Cref{alg:greedy_decomposable}, we add the interval $[\Tilde{a}_j, \Tilde{a}_j+\pi_{i,j}(\tau^\star))$ to $C^\GrdDecs_{ij}$.\footnote{If $\pi_{i,j}(\tau^\star) = 0$ we do not add anything.}

    We decompose $A^\optd$ into two disjoint allocation regions $A^\optd \setminus A^\GrdDecs$ and $A^\optd \cap A^\GrdDecs$ and upper bound the additional welfare $w(A^\optd \setminus A^\GrdDecs)$ achieved by \UtDec over \GrdDec. To this end, we further decompose $A^\optd \setminus A^\GrdDecs$ by considering each voter's contribution to it, computing its welfare as $w(A^\optd \setminus A^\GrdDecs) = \sum_{i \in [n]} w((A^\optd \setminus A^\GrdDecs) \cap C^\optd_i)$. It is easy to see that for each alternative $j$ the sets $C^\GrdDecs_{ij}$ of all voters form a disjoint partition of $A^\GrdDecs_j = [0, a^\GrdDecs_j)$ and that the contributions $c^\GrdDecs_{ij} = |C^\GrdDecs_{ij}|$ certify the decomposability of $a^\GrdDecs$ (see proof of \Cref{prop:greedydecomp_decomposable}).
    
    We will argue that, for each voter, their contribution to the welfare of $A^\optd \setminus A^\GrdDecs$ is at most as big as their contribution $w(C^\GrdDecs_i)$ to the welfare of $A^\GrdDecs$. 
    To see why, we first fix a voter $i$. If $(A^\optd \setminus A^\GrdDecs) \cap C^\optd_i = (\emptyset)_{j \in [m]}$, then voter $i$ does not contribute to $A^\optd \setminus A^\GrdDecs$ and thus $w((A^\optd \setminus A^\GrdDecs) \cap C^\optd_i) = 0 \le w(C^\GrdDecs_i)$ trivially holds. Otherwise, let $k \in [n]$ be the minimum value such that $[\mu^{k-1}_j, \mu^k_j) \cap (A^\optd_j \setminus A^\GrdDecs_j) \cap C^\optd_{ij}$ is non-empty for some alternative $j$. Then, the welfare contribution of that voter to $A^\optd_j \setminus A^\GrdDecs_j$ is upper-bounded by 
    \begin{align*}
        w\big((A^\optd \setminus A^\GrdDecs) \cap C^\optd_i\big) 
        &= \sum_{j \in [m]} \sum_{k' = k}^{n} (n-k'+1) \cdot \big|\big[\mu^{k'-1}_j, \mu^{k'}_j\big) \cap \big(A^\optd_j \setminus A^\GrdDecs_j\big) \cap C^\optd_{ij}\big| \\
        &\le \sum_{j \in [m]} \sum_{k' = k}^{n} (n-k'+1) \cdot \big|\big[\mu^{k'-1}_j, \mu^{k'}_j\big) \cap C^\optd_{ij}\big| \\
        &\le \sum_{j \in [m]} (n-k+1) \cdot \big|\big[\mu^{k-1}_j, 1\big) \cap C^\optd_{ij}\big| \\
        &\le \sum_{j \in [m]} (n-k+1) \cdot |C^\optd_{ij}| = \frac{n-k+1}{n}.    
    \end{align*}

    Since $[\mu^{k-1}_j, \mu^k_j) \cap (A^\optd_j \setminus A^\GrdDecs_j)$ is non-empty, we know that \GrdDec spends less than $\mu^k_j$ on alternative $j$.  
    Now, consider the $k^\text{th}$ outer iteration of \GrdDec after the last inner iteration with $\tau^\star=1$. At that time, we have $\Tilde{a}_j \le a^\GrdDecs_j < \mu^k_j$ and $\Tilde{a}_j \le a^\GrdDecs_j < a^\optd_j \le p_{i,j}$. The only reason \GrdDec did not buy more of alternative $j$ is that all voters above $\Tilde{a}_j$ have run out of budget. In particular, we thus know that voter $i$ spent their whole budget on intervals below the $k^\text{th}$ lowest vote, implying $w(C^\GrdDecs_i) \ge \frac{n-k+1}{n} \ge w((A^\optd\setminus A^\GrdDecs) \cap C^\optd_i)$.

    \changed{Before we conclude, we state two more observations. Upper range respect of \GrdDec implies that the welfare from $A^\GrdDecs \setminus A^\optd$ is at least $|A^\GrdDecs \setminus A^\optd|$. Moreover, proportional spending implies $a^\GrdDecs_j \ge \mu^1_j$ for all alternatives $j$, and hence} the welfare of $A^\optd \setminus A^\GrdDecs$ is at most $(n-1) \cdot |A^\optd \setminus A^\GrdDecs| = (n-1) \cdot |A^\GrdDecs \setminus A^\optd|$. 
    Thus, $w(A^\GrdDecs \setminus A^\optd) \ge \frac{1}{n-1} w(A^\optd \setminus A^\GrdDecs)$ and therefore
    
    \begin{align*}
        w(a^\optd) 
        & = w(A^\optd \cap A^\GrdDecs) + w(A^\optd \setminus A^\GrdDecs)  \\
        & = w(A^\GrdDecs) - w(A^\GrdDecs \setminus A^\optd) +  w(A^\optd \setminus A^\GrdDecs)  \\
        & \le w(A^\GrdDecs) + \frac{n-2}{n-1} w(A^\optd \setminus A^\GrdDecs)  \\
        & = w(A^\GrdDecs) + \frac{n-2}{n-1} \sum_{i \in [n]} w((A^\optd \setminus A^\GrdDecs) \cap C^\optd_i) \\
        & \le w(A^\GrdDecs) + \frac{n-2}{n-1} \sum_{i \in [n]} w(C^\GrdDecs_i) \\
        & \le \bigg(2 - \frac{1}{n-1}\bigg) \cdot w(a^\GrdDecs).\qedhere
    \end{align*}
\end{proof}

In Appendix~\ref{app:lbGrdDec}, we generalize the construction from \Cref{ex:greedydecomp_social_welfare_small} and give instances on which \UtDec has $2-\Theta(1/n)$ times the welfare that \GrdDec achieves, showing that the analysis above is asymptotically tight. Since the profile used in the proof has $m = n+1$, the instance also yields a lower bound of $2-\Theta(1/m)$ for the welfare approximation.
Beyond this specific mechanism, the question of the best approximation of the welfare provided by \UtDec that can be guaranteed by a decomposable, polynomial-time mechanism remains open.

\section{Discussion}\label{sec:discussion}

We have studied budget aggregation under $\ell_1$ utilities from the dual perspectives of proportionality and social welfare. By deriving a tight bound on the price of single-minded proportionality, our results quantify the welfare cost of enforcing fairness in this setting. Moreover, this cost (in terms of $n$) does not increase even when single-minded proportionality is combined with truthfulness or strengthened to decomposability. Beyond worst-case guarantees, by introducing an instance-wise welfare dominance relation, we identify \UtP as welfare-optimal among single-minded proportional moving-phantom mechanisms. Finally, we have shown that while the welfare-optimal decomposable outcome is NP-hard to compute, a simple greedy mechanism \GrdDec achieves a 2-approximation. We discuss extensions and open questions below.

\paragraph{Weighted Setting} A natural generalization of the budget aggregation setting is \textit{donor coordination} \citep{brandl2021distribution,brandt2023balanced}, where voters contribute different amounts to a pool of money that is to be distributed among several charities. We can extend our model to capture such entitlements by assigning each voter $i$ an integer weight $\omega_i$. In Appendix~\ref{app:weights} we discuss how to generalize the mechanisms and axioms studied in this paper to the weighted setting and demonstrate that all results of this paper extend. In particular, we show that the class of weighted moving-phantom mechanisms is truthful, that all our welfare-approximation bounds still hold when replacing the parameter $n$ by the total weight of all voters $b$, and that all instance-wise welfare domination results remain true for the weighted versions of the moving-phantom mechanisms.

\paragraph{Welfare Approximation for Other Proportional Mechanisms} 
In \Cref{sec:worst_case_price}, we showed that \UtP attains the optimal welfare approximation among single-minded proportional mechanisms, namely $\Theta(\sqrt{n})$.
This guarantee does not hold for \PiecewiseUniform, \IM, or \Fan.
Indeed, in Appendix~\ref{app:lb-im-pu} we establish a lower bound, linear in $n$, on the welfare approximation of \PiecewiseUniform; the same bound applies to \IM and \Fan by \Cref{thm:phantom-domination}.
\changed{On the other hand, while \Ladder does not satisfy proportional spending, it does attain a welfare approximation of $\Theta(\sqrt{n})$ (see Appendix~\ref{app:lb-im-pu}).}

\paragraph{Phantom Lattice} 
Consider the welfare-dominance relation defined in \Cref{sec:phantom_domination}. Interestingly, for both the class of moving-phantom mechanisms and the class of single-minded moving-phantom mechanisms, we can identify a welfare-maximizing and a welfare-minimizing moving-phantom mechanism. Since our dominance relation is, by construction, a preorder, the existence of such extremal elements was not a priori guaranteed and may therefore be viewed as surprising. Their existence naturally leads to the following open question: Does the dominance relation between moving-phantom mechanisms form a \emph{lattice} within the set of all moving-phantom mechanisms? While we do not settle this question in full generality, we show in Appendix~\ref{app:latticeTwoAlternatives} that the statement holds in the special case of two alternatives (up to adjusting definitions for technicalities). 

\paragraph{$\ell_1$-Cost versus $\ell_1$-Utility}
\changed{We have focused on $\ell_1$-utilities as our objective, which coincide with the natural interpretation of overlap utilities. However, one could imagine focusing on $\ell_1$-cost, where a voter's cost for an outcome is the $\ell_1$ distance from their preference. Approximation guarantees for the total utility do not directly carry over to the total cost. Since the optimal cost can get arbitrarily close to zero, any small absolute deviation from the optimum can lead to a large approximation factor. To illustrate this, consider the following instance with $n$ voters and two alternatives parametrized by some small $0 < \varepsilon< \frac{1}{n}$, where $n-1$ voters report $(1,0)$, while one voter reports $(1-\varepsilon, \varepsilon)$. The cost-optimal allocation is $(1,0)$, with a total $\ell_1$-cost of $2\varepsilon$, whereas \UtP outputs $(1-\varepsilon, \varepsilon)$, with a total $\ell_1$-cost of $2(n-1)\varepsilon$. Although the output of \UtP is very close to the optimal allocation, the approximation factor is large $(n-1)$, which suggests that $\ell_1$-cost may be less suitable as a benchmark for multiplicative approximation guarantees. Note that this outcome is not specific to \UtP; in fact, all single-minded proportional truthful mechanisms, as well as all decomposable mechanisms, return $(1-\varepsilon, \varepsilon)$ for this profile.}

\paragraph{Equilibrium Analysis for Non-Truthful Mechanisms}
\changed{Our results measure social welfare according to the reported preferences of the voters. For truthful mechanisms this is a very reasonable assumption, but for non-truthful mechanisms it might be the case that true preferences differ from reported preferences. An alternative direction for future work would be to study welfare under strategic play. Such equilibrium welfare guarantees could in principle be either better or worse than those established here and would require additional assumptions about voter behavior and equilibrium selection.}

\appendix

\pagebreak

\section{Omitted Proofs} 

\subsection{Proof of \Cref{prop:alpha-star}}\label{app:propAlphaStar}

\propAlphaStar*

\proof
We fix $n\in \N$ and define the function $h\colon [0,n]\to \R_+$ by 
    \[
        h(\ell) = \frac{n\ell}{n+\ell(\ell-1)}\quad \text{for every } \ell\in [0,n].
    \]
    Its derivative is given by 
    \[
         h'(\ell) = \frac{n(n+\ell(\ell-1)) - n\ell(2\ell-1)}{(n+\ell(\ell-1))^2} = \frac{n(n-\ell^2)}{(n+\ell(\ell-1))^2},
    \]
    so $h$ is increasing for $\ell\in [0,\sqrt{n})$, decreasing for $\ell\in (\sqrt{n},n]$, and attains its maximum at $\ell^\star=\sqrt{n}$, with 
    \[
        h(\sqrt{n}) = \frac{n\sqrt{n}}{n+\sqrt{n}(\sqrt{n}-1)} = \frac{n}{2\sqrt{n}-1}.\qedhere
    \]
\endproof

\subsection{Welfare Approximation Bounds with Respect to the Number of Alternatives} \label{app:welfare_m_approximations}

We first give the proof for the lower bound on the welfare approximation in terms of $m$.

\mlowerbound*

\proof
    Fix $m\in \N$ with $m\geq 2$. We consider a single-minded profile $P\in \mathcal{P}_{n,m}$, such that a fraction $t\in \big(\frac{1}{m},1\big)$ of the voters support alternative $1$ and every other alternative is supported by a fraction $\frac{1-t}{m-1}$ of the voters. More specifically, for $i\in [tn]$ we have $p_{i,1}=1$ and $p_{i,j}=0$ for every $j\in \{2,\ldots,m\}$, and for every $k\in [m-1]$ and $i\in \big\{\big(t+(k-1)\frac{1-t}{m-1}\big)n+1,\ldots, \big(t+k\frac{1-t}{m-1}\big)n\big\}$ we have $p_{i,j}=1$ if $j=1+k$ and $p_{i,j}=0$ otherwise, where $n$ and $t$ are such that all these values are integers.

    Any single-minded proportional mechanism $\A$ returns $\A(P)=\big(t,\frac{1-t}{m-1},\ldots,\frac{1-t}{m-1}\big)$, while an optimal outcome in terms of welfare is $a^{\opt}=(1,0,\ldots,0)$.
    Thus,
    \[
        \frac{w(a^\opt)}{w(\A(P))} = \frac{t}{t^2+\frac{(1-t)^2}{m-1}} = \frac{(m-1)t}{mt^2-2t+1}.
    \]
    For every $\varepsilon>0$, taking a rational $t$ sufficiently close to $\frac{1}{\sqrt{m}}$, and $n$ sufficiently large and divisible by the denominator of $t$ and by $m-1$ yields 
    \[
        \frac{w(a^\opt)}{w(\A(P))} > \frac{m-1}{\sqrt{m}\big(1-\frac{2}{\sqrt{m}}+1\big)} -\varepsilon = \frac{\sqrt{m}+1}{2} - \varepsilon,
    \]
    which finishes the proof.
\endproof

We next show the matching upper bound by proving that every mechanism satisfying proportional upper bounds and proportional spending achieves a worst-case optimal welfare approximation of $\beta^\star(m)$.

\optmapprox*
\proof
    Let $m,n\in \N$ with $m \ge 2$ and consider a profile $P\in \mathcal{P}_{n,m}$ and a mechanism $\A$ satisfying proportional spending and proportional upper bounds.
    We denote its aggregate by~$a=\A(P)$ and the aggregate of \Ut by~$a^\Uts$. 
    We prove that $\frac{w(a^\Uts)}{w(a)} \le \frac{\sqrt{m}+1}{2}$ by showing a lower bound on the welfare that $a$ achieves and combining it with an upper bound on the additional welfare that \Ut achieves over $a$ on alternatives $j$ with $a^\Uts_j > a_j$. For a subset $J \subseteq [m]$ of alternatives, we write $w_{J}(a) = \sum_{j \in J} \sum_{i \in [n]} \min(p_{i,j}, a_j)$ for the social welfare derived by $a$ from these alternatives. For the remainder of the proof, let $J = \{j \in [m] \mid a_j > a^\Uts_j \}$ be the alternatives that receive strictly more in $a$ than in $a^\Uts$. 
    If $\lvert J \rvert = 0$, then $a = a^\Uts$ and the claim is immediate.

    As in the proof of \Cref{thm:optimal_social_welfare_approximation}, let $k^\star\in [n]$ be the minimum number such that $\sum_{j \in [m]} \mu^{k^\star}_j \ge \frac{n-k^\star+1}{n}$. This is well-defined, since by the normalization of the votes we have $\sum_{j \in [m]} \mu^n_j \ge 1 \ge \frac{1}{n}$. By proportional spending, we know that $a_j \ge \mu^{k^\star-1}_j$ for all $j \in [m]$ (where we denote $\mu^0_j = 0$ for all $j \in [m]$). 
    
    Let $\ell = n-k^\star+1$ be the (lower bound on the) number of voters with vote at least $\mu^{k^\star}_j$ on each alternative $j$.
    For each $j \in J$, let $q_j=|\{i\in[n] \mid p_{i,j}\ge a_j\}|$ be the number of voters who support the entire allocation of $\A$ to $j$. Proportional upper bounds imply $a_j \le \frac{q_j}{n}$.

    \textbf{Lower bound on the welfare obtained by $\A$.}
    For $j \in J$, we consider the additional funding $a_j - a^\Uts_j$ that $j$ receives and split it into two parts. Let $x_j$ be the part of the additional funding on alternative $j$ that lies above $\mu^{k^\star}_j$. Explicitly,
    $$
        x_j =
        \begin{cases}
            a_j - \max(a^\Uts_j, \mu^{k^\star}_j) & \text{if }a_j > \mu^{k^\star}_j, \\
            0 &\text{if } a_j \le \mu^{k^\star}_j.
        \end{cases}
    $$
    By the definition of $\mu^{k^\star}_j$, there are at least $\ell$ voters supporting the spending up to $\min(a_j, \mu_j^{k^\star})$, and there are $q_j$ voters supporting the extra spending $x_j$ above $\mu^{k^\star}_j$ on alternatives $j \in J$. 
    Therefore 
    \begin{equation}
        \label{eq:difference-on-J}
    w_J(a) - w_J(a^\Uts) \ge \sum_{j \in J} \ell (a_j - a^\Uts_j - x_j) + \sum_{j \in J} q_j x_j.
    \end{equation}
    
    Using the proportional spending property of $\A$, we get for $k^\star$ that \[\sum_{j\in [m]} \min(a_j, \mu^{k^\star}_j) \ge \min\left(\frac{\ell}{n}, \sum_{j \in [m]} \mu^{k^\star}_j\right) = \frac{\ell}{n}\] and can lower bound the welfare of $\A$ by
    \begin{equation}
        w(a) \ge \ell \sum_{j \in [m]} \min(a_j, \mu^{k^\star}_j) + \sum_{j\in J}q_j x_j \ge \frac{\ell^2}{n} + \sum_{j \in J} q_j x_j.
        \label{eq:ps-pup-welfare_lower_bound}
    \end{equation}

    \textbf{Upper bound on the extra welfare obtained by \Ut.}
    For an alternative $j \notin J$, i.e., one with $a^\Uts_j \geq a_j$, only voters reporting strictly more than $a_j$ can obtain additional utility on alternative $j$ under $a^{\Uts}$ relative to $a$. As $a_j \ge \mu^{k^\star-1}_j$ on all alternatives $j$, there are at most $\ell$ such voters and each gains at most $a^\Uts_j-a_j$. Consequently, the total extra welfare on these alternatives is at most
    \begin{equation}
        w_{[m] \setminus J}(a^\Uts) - w_{[m] \setminus J}(a) \le \ell \sum_{j\in [m] \setminus J} (a^\Uts_j - a_j) = \ell \sum_{j\in J} (a_j - a^\Uts_j),
        \label{eq:ps-pup-extra-opt}
    \end{equation}
    where the equality follows from normalization.

    \textbf{Combining the bounds.}
    Combining the bounds from Equations~\eqref{eq:difference-on-J},~\eqref{eq:ps-pup-welfare_lower_bound}, and~\eqref{eq:ps-pup-extra-opt} above, we obtain
    
    \begin{align*}
        \frac{w(a^\Uts)}{w(a)}
        = 1 + \frac{w(a^\Uts) - w(a)}{w(a)} & = 1 + \frac{w_{[m] \setminus J}(a^\Uts) - w_{[m] \setminus J}(a) - (w_J(a) - w_J(a^\Uts))}{w(a)} \\
        & \le 1 + \frac{\ell \sum_{j\in J} x_j - \sum_{j\in J} q_j x_j}{\frac{\ell^2}{n} + \sum_{j \in J} q_j x_j} =\frac{\frac{\ell^2}{n} + \ell\sum_{j \in J} x_j}{\frac{\ell^2}{n} + \sum_{j \in J} q_j x_j}.
    \end{align*}
    
    Since $x_j\le a_j\le q_j/n$ for each $j \in J$, the denominator sum is bounded by $\sum_{j \in J} q_j x_j \ge n \sum_{j \in J} x_j^2$.
    Moreover, normalization implies that at least one alternative receives strictly more under $a^\Uts$, so $1 \le |J| \le m-1$. By Cauchy-Schwarz (applied to $(x_j)_{j \in J}$ and $(1)_{j \in J}$), we then get $n \sum_{j\in J} x_j^2 \ge \frac{n}{|J|} (\sum_{j\in J} x_j )^2 \ge \frac{n}{m-1} (\sum_{j \in J} x_j)^2$. Substituting this lower bound into the denominator gives
    $$
        \frac{w(a^\Uts)}{w(a)} \le \frac{\frac{\ell^2}{n} + \ell\sum_{j\in J}x_j}{\frac{\ell^2}{n}+\frac{n}{m-1} (\sum_{j \in J} x_j)^2}.
    $$
    
    To prove that this term is always at most $\frac{\sqrt{m}+1}{2}$, we divide the numerator and denominator by $\frac{\ell^2}{n}$ and then substitute $t = \frac{n}{\ell} \sum_{j\in J} x_j \ge 0$:
    $$
        \frac{w(a^\Uts)}{w(a)} \le f(t) := \frac{1+t}{1+\frac{t^2}{m-1}}.
    $$
    Differentiating, we obtain
    \begin{equation*}
        f'(t)
        =\frac{1+\frac{t^2}{m-1}-\frac{2t(1+t)}{m-1}}{\left(1+\frac{t^2}{m-1}\right)^2}
        =\frac{m-(t+1)^2}{(m-1)\left(1+\frac{t^2}{m-1}\right)^2}.
    \end{equation*}
    The denominator is positive, so the sign of the derivative is determined by $m-(t+1)^2$. Thus $f$ is increasing for $0 \le t < \sqrt{m}-1$ and decreasing for $t > \sqrt{m} - 1$. Its maximum is therefore attained at $t = \sqrt{m}-1$, yielding
    $$
        \frac{w(a^\Uts)}{w(a)} \le f(\sqrt{m}-1) = \frac{\sqrt{m}}{1+\frac{(\sqrt{m}-1)^2}{m-1}} = \frac{\sqrt{m}}{1+\frac{\sqrt{m}-1}{\sqrt{m}+1}} = \frac{\sqrt{m}}{\frac{2\sqrt{m}}{\sqrt{m}+1}} = \frac{\sqrt{m}+1}{2} = \beta^\star(m).
    $$
    
    This proves the claimed approximation bound.
    \endproof

It remains to show that \UtP and \GrdDec satisfy proportional upper bounds. We do this by proving two more general statements.

\begin{proposition} \label{prop:prop_phantom_prop_upper_bounds}
    Every single-minded proportional moving-phantom mechanism satisfies proportional upper bounds.
\end{proposition}

\proof
    Suppose that, for some profile $P \in \mathcal{P}_{n,m}$ and alternative $j \in [m]$, the allocation $a$ returned by some single-minded proportional moving-phantom mechanism is such that $a_j > \frac{\lvert \{i \in [n] \mid p_{i,j} \ge a_j \} \rvert}{n}$.  From \Cref{lem:proportional_phantoms_characterization} we know that the mechanism is induced by a phantom system with $f_k(1) = \frac{n-k}{n}$. Let $q = \lvert \{i \in [n] \mid p_{i,j} \ge a_j \} \rvert$ be the number of voters voting at least $a_j$ on $j$. Since $a_j$ is the median on alternative $j$ at the time of normalization $t^\star$, there must be $n-q+1$ phantoms above $a_j$ at the time of normalization. In particular, we have $f_{n-q}(t^\star) \geq a_j > \frac{q}{n} = f_{n-q}(1)$, contradicting the monotonicity of the phantom functions.
\endproof

\begin{proposition}\label{prop:decomposability_prop_upper_bounds}
    Every decomposable mechanism satisfies proportional upper bounds.
\end{proposition}

\proof
     Suppose that, for some profile $P \in \mathcal{P}_{n,m}$ and alternative $j \in [m]$, a decomposable allocation $a$ is such that $a_j > \frac{\lvert \{i \in [n] \mid p_{i,j} \ge a_j \} \rvert}{n}$. Let $(c_i)_{i \in [n]}$ be voter contributions witnessing decomposability. Let $S = \{i \in [n] \mid p_{i,j} \ge a_j \}$. The total contribution to alternative $j$ from voters in $S$ is at most $\sum_{i \in S} c_{i,j} \le \sum_{i\in S}\sum_{j' \in [m]} c_{i,j'} = \frac{|S|}{n} < a_j$. Thus, some voter $i \in [n] \setminus S$ must have $c_{i,j} > 0$. As $i$ is not in $S$, we have $p_{i,j} < a_j$, contradicting decomposability.
\endproof

\subsection{Proof of \Cref{thm:phantom-domination}}\label{app:thmPhantomDomination}

\thmPhantomDomination*

\proof

 \textit{Item \ref{item:welf-all}.} It is well known that \Ut is welfare-maximizing among moving-phantom mechanisms \citep{freeman2021truthful}. 
    To see that \Constant is welfare-minimizing, we observe that for any profile, \Constant attains normalization when all phantoms are at position $\frac{1}{m}$. 
    Let $n, m$ be fixed, $P \in \mathcal{P}_{n,m}$ be a profile, and $\F_n$ a phantom system attaining normalization at time $t$. 
    Choose $k^\star$ to be the index of the lowest phantom at a position $\frac{1}{m}$ or larger, or $k^\star = -1$ if no such phantom exists. 
    Then, we have $f_i(t) \ge \frac{1}{m}$ for all $i \in \{0, \dots, k^\star\}$ and $f_i \le \frac{1}{m}$ for all $i \in \{k^\star+1, \dots, n\}$. 
    By \Cref{lem:phantom_domination}, we conclude that $\A^{\F_n} \trianglerighteq \Constant$.

    \textit{Item \ref{item:welf-smp} (second part).} Fix $n\in \N$ and let $\G_n$ be the phantom system of \Fan. By \Cref{lem:proportional_phantoms_characterization}, for any single-minded proportional moving-phantom mechanism $\A$, we can find a phantom system $\F_n=\{f_k\mid k\in [n]_0\}$ inducing it such that $f_k(1) = 1 - \frac{k}{n}$ for each $k\in [n]_0$.
    We show that for any two times $t, t' \in [0,1]$, we can find a threshold index $k^\star$ such that all phantoms with index $k \le k^\star$ from $\F_n$ lie above the ones from $\G_n$ at times $t$ and $t'$, respectively, and all the phantoms with index $k > k^\star$ lie below. 
    Indeed, let $k^\star$ be the highest index such that $f_{k^\star}(t) \ge t'$, or choose $k^\star = -1$ if such an index does not exist. Note that, by the definition of \Fan, we then have $f_{k}(t) \ge t' \ge g_{k}(t')$ for all $k \leq k^\star$. For the lower phantoms of $\F_n$ with $k > k^\star$, we know by monotonicity that they are lower than (or equal to) both $t'$ and their end positions, i.e., $f_{k}(t) \le \min\big(t', 1 - \frac{k}{n}\big) = g_{k}(t')$ for all $k > k^\star$. Applying \Cref{lem:phantom_domination}, with $t$ and $t'$ equal to the respective times of normalization, yields $\A \trianglerighteq \Fan$.

\textit{Item \ref{item:welf-pwu-lad}.}
    We first prove that $\PiecewiseUniform \trianglerighteq \IM$. 
    Let $\F = \{\F_n \mid n \in \N\}$ and $\G = \{\G_n \mid n \in \N\}$ be the families of phantom systems of \PiecewiseUniform and \IM, respectively. Fix some $n \in \N$. We show that for any two times $t, t' \in [0,1]$, we can find a threshold index $k^\star$ such that $f_i(t) \ge g_i(t')$ for all $i \in \{0, \dots, k^\star\}$ and $f_i(t) \le g_i(t')$ for all $i \in \{k^\star+1, \dots, n\}$.
    Applying \Cref{lem:phantom_domination} then completes the proof.

    Note that at time $t'$, the phantoms of \IM are spaced equally with a distance of $d^\IMs = \frac{t'}{n}$ between them. Similarly, the lower phantoms $f_n$ to $f_{\lfloor \frac{n}{2} \rfloor + 1}$ of \PiecewiseUniform  are equally spaced with a distance of $d^\PiecewiseUniforms_\ell = \max
    \big(0, \frac{2t-1}{n}\big)$ between them. The upper phantoms $f_{\lfloor \frac{n}{2} \rfloor}$ to $f_0$ are also equally spaced, with a distance $d^\PiecewiseUniforms_u$ of $\frac{4t}{n}$ if $t < \frac{1}{2}$ and $\frac{3-2t}{n}$ if $t \ge \frac{1}{2}$.

    We scan the phantoms from \PiecewiseUniform from bottom to top and compare their positions with the phantoms from \IM. Both lowest phantoms are at position $0$, i.e. $f_n(t) = 0  = g_n(t')$. If $d^\PiecewiseUniforms_\ell = \max(0, \frac{2t-1}{n}) > \frac{t'}{n} = d^\IMs$, then we have $t > \frac{1}{2}$ and thus $d^\PiecewiseUniforms_u = \frac{3-2t}{n} \ge \frac{1}{n} \ge d^\IMs$. Therefore, all phantoms of \PiecewiseUniform are above or equal to the phantoms from \IM, and we can choose $k^\star = n$.
    Otherwise ($d^\PiecewiseUniforms_\ell \le d^\IMs$), we know that all lower phantoms from \PiecewiseUniform are at most at the phantom positions from \IM, i.e., $f_k(t) \le g_k(t')$ for all $k \in \{\lfloor \frac{n}{2} \rfloor + 1, \dots, n\}$. 
    If $d^\PiecewiseUniforms_u \le d^\IMs$, then the same holds for the remaining phantoms, and we can choose $k^\star = -1$. Otherwise, choose $k^\star \le \lfloor \frac{n}{2} \rfloor$ to be the index of the lowest phantom with $f_{k^\star}(t) > g_{k^\star}(t')$ or $k^\star = -1$ if no such index exists. Then, for all phantoms $k < k^\star$ from \PiecewiseUniform above $f_{k^\star}$, we have 
    $$
        f_{k}(t) = f_{k^\star}(t) + (k^\star-k) d^\PiecewiseUniforms_u > g_{k^\star}(t') + (k^\star-k) d^\IMs = g_{k}(t').
    $$

    We next prove that $\Ladder \trianglerighteq \IM$.
    Let $\F = \{\F_n \mid n \in \N\}$ and $\G = \{\G_n \mid n \in \N\}$ be the families of phantom systems of \Ladder and \IM, respectively. Fix some $n \in \N$. We show that for any two times $t, t' \in [0,1]$, we can find a threshold index $k^\star$ such that $f_i(t) \ge g_i(t')$ for all $i \in \{0, \dots, k^\star\}$ and $f_i(t) \le g_i(t')$ for all $i \in \{k^\star+1, \dots, n\}$. Applying \Cref{lem:phantom_domination} then completes the proof.

    We choose $k^\star$ to be the largest index such that $f_{k^\star}(t) > g_{k^\star}(t')$. If this does not exist, then we can choose $k^\star = -1$ and conclude. Note that, in particular, we have $f_{k^\star}(t) > 0$, since $g_k(t') \ge 0$ for any $k$. By the definition of \Ladder and \IM, we have for each $k < k^\star$
    \begin{align*}
        f_{k}(t) &= t - \frac{k}{n} = t - \frac{k^\star}{n} + \frac{k^\star-k}{n} = f_{k^\star}(t) + \frac{k^\star-k}{n} \\
        &> g_{k^\star}(t') + t' \cdot \frac{k^\star-k}{n} = t' \cdot \frac{n-k^\star}{n} + t' \cdot \frac{k^\star-k}{n} = t' \cdot \frac{n-k}{n} = g_{k}(t').
    \end{align*}
    Similarly, for each $k \ge k^\star+1$ with $f_k(t) > 0$ we know that 
    \begin{align*}
        f_{k}(t) &= t-\frac{k}{n} = t - \frac{k^\star+1}{n} - \frac{k-k^\star-1}{n} = f_{k^\star+1}(t) - \frac{k-k^\star-1}{n}  \\
        &\le g_{k^\star+1}(t') - t' \cdot \frac{k-k^\star-1}{n} = t' \cdot \frac{n-k^\star-1}{n} - t' \cdot \frac{k-k^\star-1}{n} = t'\cdot \frac{n-k}{n} = g_{k}(t').
    \end{align*}
    Finally, for all $k > k^\star$ with $f_k(t) = 0$ it is trivial that $f_{k}(t) = 0 \le g_{k}(t')$.

    For the last part of the claim in the statement, we construct a profile $P$ where \Ladder achieves higher welfare than \PiecewiseUniform, and a profile $P'$ for which the opposite relation holds.

    Let $n = 4$ and $m = 3$.
    Profiles $P, P' \in \mathcal{P}_{4,3}$ are defined by the votes
    \begin{align*} 
        p_1 &= (1, 0, 0) & p_2 &= (0, 1, 0) \qquad\qquad& p'_1 &= (\textstyle\frac{1}{2}, \textstyle\frac{1}{2}, 0) & p'_2 &= (\textstyle\frac{1}{2}, \textstyle\frac{1}{2}, 0) \\
        p_3 &= (0, 0, 1) &  p_4 &= (\textstyle\frac{1}{2}, \textstyle\frac{1}{2}, 0) & p'_4 &= (0, \textstyle\frac{1}{2}, \textstyle\frac{1}{2}) & p'_3 &= (\textstyle\frac{1}{2}, 0, \textstyle\frac{1}{2})
    \end{align*}
    
    In profile $P$, \Ladder reaches normalization at time $t = \frac{11}{12}$ when the phantoms positions are $\big( \frac{11}{12}, \frac{2}{3}, \frac{5}{12}, \frac{1}{6}, 0\big)$, outputting the allocation $\big(\frac{5}{12}, \frac{5}{12}, \frac{1}{6}\big)$ with welfare $\frac{11}{6}$. \PiecewiseUniform reaches normalization at time $t = \frac{9}{10}$ when the phantom positions are $\big( 1, \frac{7}{10}, \frac{2}{5}, \frac{1}{5}, 0 \big)$, outputting the allocation $\big(\frac{2}{5}, \frac{2}{5}, \frac{1}{5}\big)$ with welfare $\frac{9}{5} < \frac{11}{6}$.

    In profile $P'$, \Ladder reaches normalization at time $t = \frac{2}{3}$ when the phantom positions are $\big( \frac{2}{3}, \frac{5}{12}, \frac{1}{6}, 0, 0\big)$, outputting the allocation $\big(\frac{5}{12}, \frac{5}{12}, \frac{1}{6}\big)$ with welfare $\frac{17}{6}$. \PiecewiseUniform reaches normalization at time $t = \frac{1}{2}$ when the phantom positions are $\big( 1, \frac{1}{2}, 0, 0, 0 \big)$, outputting the allocation $\big(\frac{1}{2}, \frac{1}{2}, 0\big)$ with welfare $3 > \frac{17}{6}$.

\textit{Item \ref{item:welf-fan}.}
    Let $\F = \{\F_n \mid n \in \N\}$ and  $\G = \{\G_n \mid n \in \N\}$ be the families of phantom systems of \Fan and \GreedyMax, respectively. Fix some $n \in \N$. We show that for any two times $t, t' \in [0,1]$, we can find a threshold index $k^\star$, such that $f_i(t) \ge g_i(t')$ for all $i \in \{0, \dots, k^\star\}$ and $f_i(t) \le g_i(t')$ for all $i \in \{k^\star+1, \dots, n\}$. Applying \Cref{lem:phantom_domination} then completes the proof.

    We choose $k^\star$ to be the highest index such that $f_{k^\star}(t) \ge g_{0}(t')$, or $k^\star = -1$ if no such index exists. If $k^\star = n$, then $t' = 0$, since $f_n(t) = 0$ and thus $f_i(t) \ge g_i(t') = 0$ for all $i = 0, \dots, k^\star = n$ follows trivially. Otherwise, we have $f_k(t) \ge g_k(t') = g_{0}(t')$ for all $k = 0, \dots, k^\star$ and $f_k(t) \le g_k(t') = g_{0}(t')$ for all $k \in \{k^\star+1, \dots, n-1\}$. Since $f_n(t) = 0 \le 0 = g_n(t')$ also holds, we conclude.
\endproof

\subsection{Proof of \Cref{lem:proportional_phantoms_characterization}}\label{app:lemProportionalPhantoms}

\lemProportionalPhantoms*

\proof
    The backwards direction of the implication was shown by \citet{freeman2021truthful}; we include it here for completeness. Let $n$ be fixed, $P \in \mathcal{P}_{n,m}$ be a single-minded profile, and $n_j$ denote the number of voters voting $1$ on alternative $j\in [m]$. Then, at time $t = 1$, the median on each alternative $j$ is the $(n-n_j)^\text{th}$ phantom, which is at position $f_{n-n_j}(1) = 1 - \frac{n-n_j}{n} = \frac{n_j}{n} = \frac{1}{n} \sum_{i \in [n]} p_{i,j}$.
    Thus, the mechanism is single-minded proportional.

    For the forward direction, let $\A^\G$ be a single-minded proportional moving-phantom mechanism induced by a phantom system $\G_n = \{g_k \mid k \in [n]_0\}$ for each $n \in \N$. Recall that, by the definition of moving-phantom mechanisms, $g_k(1) \ge \frac{n-k}{n}$ for all $k \in [n]_0$.
    
    We consider the phantom system $\F_n = \big\{f_k(t) = \min\big(g_k(t), \frac{n-k}{n}\big) \mid k \in [n]_0\big\}$, and claim that for any profile $P$, any time $t$ of normalization for $\G_n$ is also a time of normalization for $\F_n$. Since all phantoms and thus medians of $\F_n$ are at most as high as the ones of $\G_n$, normalization then implies that $\A^\G(P) = \A^\F(P)$, concluding the proof of the theorem.

    We fix $n\in \N$ in what follows. Before showing the claim, we make the following observation: For $k\in [n]_0$ and $t\in [0,1]$,
    \begin{equation}
        g_k(t) > 1 - \frac{k}{n} \quad \Longrightarrow \quad g_{k'}(t) \ge 1 - \frac{k'}{n} \text{ for all } k'\in [n]_0 \text{ with } k'\geq n-k.\label{imp:prop-phantoms}
    \end{equation}
    Indeed, suppose that $g_k(t) > 1 - \frac{k}{n}$ for some $k\in [n]_0$ and $t\in [0,1]$ and that there is some $k' \ge n-k$ with $g_{k'}(t) < 1 - \frac{k'}{n}$.
    Consider the following profile in $\mathcal{P}_{n,3}$, where $k \times p$ denotes $k$ voters voting $p$: $$\big( (n-k) \times (1,0,0), (n-k') \times (0,1,0), (k+k'-n) \times (0,0,1) \big).$$
    For $\A^\G$ to be single-minded proportional, there must be a time $t'$ at which $g_k(t') = 1 - \frac{k}{n}$ and $g_{k'}(t') = 1 - \frac{k'}{n}$. But then $g_k(t') < g_k(t)$ and $g_{k'}(t') > g_{k'}(t)$, contradicting the monotonicity of the phantom functions.

    We now prove the claim from above. Let $P\in \mathcal{P}_{n,m}$ be a profile and $t\in [0,1]$ a time of normalization for the phantom system $\G_n$. 
    We will iteratively move the phantoms from $\G_n$ continuously down to the positions of the phantoms from $\F_n$, starting with the lowest phantom $g_n$. 
    If none of the medians change during this process, then $t$ is a time of normalization for $\F$ and we conclude. 
    Otherwise, we let $k^\star$ be the index of the first (lowest) phantom whose movement changes the median, and $j$ be an alternative where the median changes. 
    Since the median is continuous and always equal to one of its input values, when moving the phantom from $g_{k^\star}(t)$ to $f_{k^\star}(t)$ there must be a position $\tau \in (f_{k^\star}(t), g_{k^\star}(t)) \setminus \{p_{i,j} \mid i \in [n]\}$ at which the median is exactly $\tau$. We denote the adjusted positions of the phantoms as $h_k \in [f_k(t), g_k(t)]$ for each $k \in [n]_0$, i.e., $h_k=f_k(t)$ for $k>k^\star$, $h_{k^\star}=\tau$, and $h_k=g_k(t)$ for $k<k^\star$.

    Now, let $n_j = n-k^\star$ be the number of voters strictly above $h_{k^\star} = \tau$ on alternative $j$. We manipulate the profile, without changing the median on $j$, by making these $n_j$ voters single-minded on $j$. Note that both moving the phantoms down and making the voters single-minded can only decrease the medians on all alternatives. Since no voter votes exactly $\tau$ on $j$, each of the remaining $n-n_j$ voters votes strictly below $\tau$ on $j$ and has some alternative $j' \neq j$ with a vote strictly higher than the current median by the normalization of the votes. 
    We make these voters single-minded on these alternatives to obtain a single-minded profile, which again cannot increase any medians. Moreover, the median on alternative $j$ is still $h_{k^\star} = \tau > f_{k^\star}(t) = \frac{n_j}{n}$. By Observation~\eqref{imp:prop-phantoms}, from $g_{k^\star}(t) = g_{n-n_j}(t) > \frac{n_j}{n}$ it follows that $g_k(t) \ge 1 - \frac{k}{n}$ for $k \ge n_j$. Let $n_{j'}$ denote the number of voters now voting $1$ on each alternative $j' \neq j$. Then, on each alternative $j' \neq j$, we have $n-n_{j'} \ge n_j$ and thus the median on $j'$ is exactly $$h_{n-n_{j'}} \ge f_{n-n_{j'}}(t) = \min\bigg(g_{n-n_{j'}}(t), \frac{n_{j'}}{n}\bigg) = \frac{n_{j'}}{n}.$$ But then, the current sum of medians is $$h_{k^\star} + \sum_{j' \in [m]: j' \neq j}  h_{n-n_{j'}} > \frac{n_j}{n} + \sum_{j' \in [m]:j' \neq j} \frac{n_{j'}}{n} = 1,$$ contradicting the normalization of $\mathcal{G}_n$ at time $t$.
\endproof

\subsection{Proof of \Cref{thm:welfare-optimal-decomposable-hardness}}\label{app:thmNPcomplete}

\thmNPcomplete*

We make use of the following problem for the reduction.

\begin{definition}[\textsc{ExactCoverBy3Sets} (X3C)]
    An instance consists of a finite ground set \(U\) with \(|U|=3q\) for some integer \(q\ge 1\), and a family \(\mathcal{S}=\{S_1,\dots,S_k\}\subseteq \binom{U}{3}\) of 3-element subsets of \(U\). The decision problem asks whether there exists a subfamily
    \(\mathcal{C}\subseteq \mathcal{S}\) of size \(q\) such that the sets in \(\mathcal{C}\) are pairwise disjoint and \(\bigcup_{S\in\mathcal{C}} S = U\) (i.e., \(\mathcal{C}\) is an exact cover of \(U\)).
\end{definition}

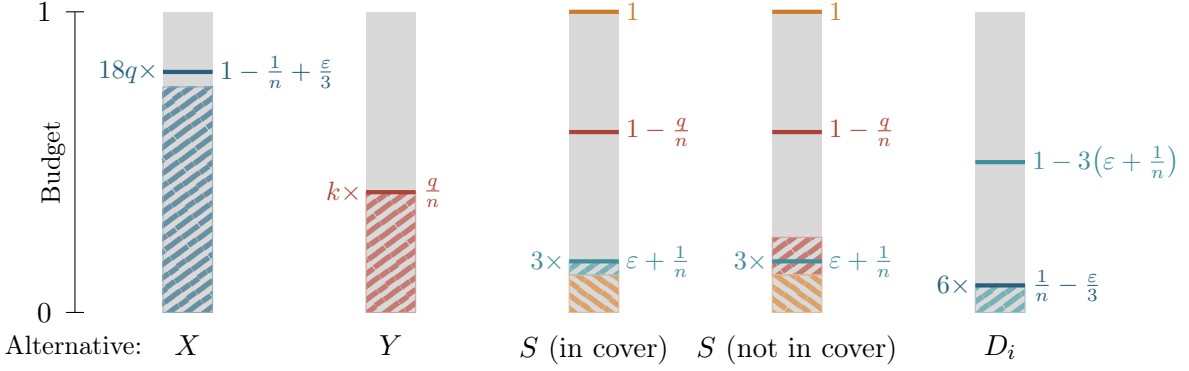
\begin{figure}
    \centering
    \def\coordwidth{0.055}
    \def\coorddist{0.47}
    \def\posone{0.25}
    \def\epsmark{0.12}
    \def\postwo{\posone + \coorddist}
    \def\posthr{\postwo + \coorddist}
    \def\posfou{\posthr + \coorddist}
    \def\posfiv{\posfou + \coorddist}
    \newcommand{\drawvoteannotated}[5]{
        \draw[ultra thick, #3] (#1-\coordwidth,#2) -- (#1+\coordwidth,#2)
        node[pos=0, inner sep=1pt, anchor=east] {\small #4}
        node[pos=1, inner sep=2pt, anchor=west] {\small #5};
    }
    \newcommand{\drawcontribution}[5]{
        \fill[pattern={Lines[angle=#5*0.8,distance=4pt,line width=2pt]},pattern color=#4!70!white, draw=#4!50!white] (#1-\coordwidth,#2) rectangle (#1-\coordwidth+\coordwidth*2,#3);
    }
    \begin{tikzpicture}[yscale=4,xscale=6]
        \draw (0,0) -- (0,1);  
        \draw (-.02,0) -- (0.02,0); 
        \draw (-.02,1) -- (0.02,1); 
        \node[anchor=east, xshift=-5px] at (0,0) {$0$}; 
        \node[anchor=east, xshift=-5px] at (0,1) {$1$}; 
        \node[anchor=south, rotate=90] at (0,0.5) {\small Budget}; 
        \filldraw[fill=alternativebarcolor,draw=none] (\posone-\coordwidth,0) rectangle (\posone+\coordwidth,1); 
        \filldraw[fill=alternativebarcolor,draw=none] (\postwo-\coordwidth,0) rectangle (\postwo+\coordwidth,1); 
        \filldraw[fill=alternativebarcolor,draw=none] (\posthr-\coordwidth,0) rectangle (\posthr+\coordwidth,1); 
        \filldraw[fill=alternativebarcolor,draw=none] (\posfou-\coordwidth,0) rectangle (\posfou+\coordwidth,1); 
        \filldraw[fill=alternativebarcolor,draw=none] (\posfiv-\coordwidth,0) rectangle (\posfiv+\coordwidth,1); 
        \node[anchor=north, outer sep=5pt] (label1) at (\posone,0) {$X$};
        \node[anchor=north, outer sep=5pt] at (\postwo,0) {$Y$};
        \node[anchor=north, outer sep=5pt] at (\posthr,0) {$S$ (in cover)};
        \node[anchor=north, outer sep=5pt] at (\posfou,0) {$S$ (not in cover)};
        \node[anchor=north, outer sep=5pt] at (\posfiv,0) {$D_i$};
        \node[anchor=east, outer sep=0pt] at (label1.west) {\small Alternative:};
        \drawcontribution{\posone}{0}{0.75}{color2}{45};
        \drawcontribution{\postwo}{0}{0.4}{color4}{45};
        \drawcontribution{\posfou}{1/8}{1/4}{color4}{45};
        \drawcontribution{\posthr}{0}{1/8}{color1}{-45};
        \drawcontribution{\posfou}{0}{1/8}{color1}{-45};
        \drawcontribution{\posthr}{1/8}{0.05+\epsmark}{color3}{45};
        \drawcontribution{\posfiv}{0}{0.05+\epsmark/3}{color3}{45};
        \drawvoteannotated{\posone}{0.8}{color2}{$18q \times$}{$1-\frac{1}{n}+\frac{\varepsilon}{3}$};
        \drawvoteannotated{\posfiv}{0.05+\epsmark/3}{color2}{$6 \times$}{$\frac{1}{n}-\frac{\varepsilon}{3}$};
        \drawvoteannotated{\postwo}{0.4}{color4}{$k\times$}{$\frac{q}{n}$};
        \drawvoteannotated{\posthr}{0.6}{color4}{}{$1-\frac{q}{n}$};
        \drawvoteannotated{\posfou}{0.6}{color4}{}{$1-\frac{q}{n}$};
        \drawvoteannotated{\posthr}{1}{color1}{}{$1$};
        \drawvoteannotated{\posfou}{1}{color1}{}{$1$};
        \drawvoteannotated{\posthr}{0.05+\epsmark}{color3}{$3\times$}{$\varepsilon+\frac{1}{n}$};
        \drawvoteannotated{\posfou}{0.05+\epsmark}{color3}{$3\times$}{$\varepsilon+\frac{1}{n}$};
        \drawvoteannotated{\posfiv}{0.5}{color3}{}{$1-3\!\left(\varepsilon+\tfrac{1}{n}\right)$};
    \end{tikzpicture}
    \caption{Illustration of the construction used in the proof of Theorem~\ref{thm:welfare-optimal-decomposable-hardness}. Only non-zero votes are drawn and only representative examples of subset and dummy alternatives are shown. Blue line segments represent helper voters, orange line segments represent single-minded voters, red line segments represent toggle voters, and purple line segments represent element voters. Quantities to the left of the line segments indicate multiple voters at that position. Quantities to the right label the vote values. Colored hatching represents contributions of each voter type in the welfare-optimal decomposable allocation assuming that an exact cover exists. If an exact cover does not exist then the amount allocated to $D_i$ increases while the amount allocated to subset alternatives decreases in the welfare-optimal decomposable allocation.}
    \label{fig:reduction}
\end{figure}

\begin{proof}[Proof of \Cref{thm:welfare-optimal-decomposable-hardness}]
We reduce from \textsc{ExactCoverBy3Sets} (X3C).
Note that the restricted version of X3C in which every element of $U$ is contained in exactly three subsets, denoted RX3C, remains NP-complete~\citep{gonzalez1985clustering}.  Given an instance $(U, \mathcal{S})$ of RX3C we define the following budget aggregation instance. There are $m=3q+k+2$ alternatives: $k$ \emph{subset} alternatives $S_1, \ldots, S_k$, $3q$ \emph{dummy} alternatives $D_1, \ldots, D_{3q}$, one for each element, and two additional alternatives $X$ and $Y$. Note that in the RX3C problem it is necessarily the case that $k=3q$; however, we will continue to treat $k$ and $q$ as distinct to make the mapping to elements and subsets more clear.

There are $n=21q+2k \ge 21q \geq 21$ voters. 
We describe in the following the construction of a profile $P\in \mathcal{P}_{n,m}$. 
Let $0<\varepsilon<\frac{1}{n(q+1)}$ be a small positive constant. For every element $i \in U$, we consider an \emph{element} voter $e_i$ with $p_{e_i,S_j}=\varepsilon+\frac{1}{n}$ for each $j$ with $i \in S_j$ and $p_{e_i,D_i}=1-3(\varepsilon+\frac{1}{n})$. For every subset $S_i$, $i \in [k]$, we consider a \emph{toggle} voter $t_i$ with $p_{t_i,Y}=\frac{q}{n}$ and $p_{t_i,S_i}=1-\frac{q}{n}$. Additionally, for every element $i \in U$ we consider six additional \emph{helper} voters $h_{i,1},\ldots, h_{i,6}$ with $p_{h_{i,\ell},D_i}=\frac{1}{n}-\frac{\varepsilon}{3}$ and $p_{h_{i,\ell},X}=1-\frac{1}{n}+\frac{\varepsilon}{3}$. Finally, for every subset $S_i$, we consider a single-minded voter $s_i$ with $p_{s_i,S_i}=1$. All other votes are zero.

We show that there exists a decomposable solution to the budget aggregation instance defined above with welfare at least $\frac{324q^{2} + kq + 28q + 4k}{n} \;+\; 4q\,\varepsilon$ if and only if the RX3C instance has an exact cover.

Suppose that there is an exact cover $\mathcal{C}$ for the RX3C instance. We claim that the following decomposable allocation $a^\star$ achieves the required welfare: allocate $a^\star_X = \frac{18q}{n}$ to alternative $X$, allocate $a^\star_Y = \frac{q}{n}$ to alternative $Y$, allocate $a^\star_{D_i} = \frac{1}{n}-\frac{\varepsilon}{3}$ to each alternative $D_i$, allocate $a^\star_{S_i}=\frac{1}{n}+\varepsilon$ to every $S_i \in \mathcal{C}$, and allocate $a^\star_{S_i}=\frac{2}{n}$ to every $S_i \not \in \mathcal{C}$. 

To see that $a^\star$ is decomposable, we can define the following voter contributions. All helper voters $h_{i,\ell}$, $i \in [3q], \ell \in [6]$, contribute all their $1/n$ mass to alternative $X$, i.e., $c_{h_{i,\ell},X}=1/n$. All single-minded voters $s_i$ contribute all of their $1/n$ mass to alternative $S_i$, i.e., $c_{s_i,S_i}=1/n$. For all $i$ with $S_i \in \mathcal{C}$, toggle voters $t_i$ contribute all $1/n$ of their mass to alternative $Y$, i.e., $c_{t_i,Y}=1/n$ (note that there are $q$ such voters). For all $i$ with $S_i \not \in \mathcal{C}$, toggle voters $t_i$ contribute all $1/n$ of their mass to alternative $S_i$, i.e., $c_{t_i, S_i}=1/n$ (note that there are $k-q$ such voters). The element voters $e_i$ contribute $\varepsilon/3$ of their mass to the alternative $S_j$ corresponding to the set that contains them in the exact cover and the remaining $1/n-\varepsilon/3$ of their mass to alternative $D_i$, i.e., $c_{e_i,S_j}=\varepsilon/3$ for $i \in S_j \in \mathcal{C}$ and $c_{e_i,D_i}=1/n-\varepsilon/3$. Note that, by definition of an exact cover, there exists exactly one such $S_j$ with $i \in S_j \in \mathcal{C}$. Notice that every voter contributes $1/n$ in total and that the sum of contributions on each alternative is equal to that alternative's allocation. It remains to check that $a^\star_j \le p_{i,j}$ for all voters $i$ and all alternatives $j$ with $c_{i,j}>0$:
\begin{itemize}
    \item Helper voters $h_{i,\ell}$ contribute positively only to alternative $X$, and we have $a^\star_X = 18q/n \le 18/21 < 1-\frac{1}{n}+\frac{\varepsilon}{3}=p_{h_{i,\ell},X}$.
    \item Single-minded voters $s_i$ contribute positively only to $S_i$, for which they have $p_{s_i,S_i}=1$.
    \item For all $i$ with $S_i \in \mathcal{C}$, toggle voters $t_i$ contribute positively only to alternative $Y$, and we have $a^\star_Y = q/n =p_{t_i,Y}$.
    \item For all $i$ with $S_i \not \in \mathcal{C}$, toggle voters $t_i$ contribute positively only to alternative $S_i$, and we have $a^\star_{S_i} = 2/n < 20/21 \le 1-q/n = p_{t_i,S_i}$.
    \item Element voters $e_i$ contribute positively to an alternative $S_j \in \mathcal{C}$, for which $a^\star_{S_j}=\varepsilon+1/n=p_{e_i,S_j}$, and to alternative $D_i$, for which $a^\star_{D_i} = 1/n-\varepsilon/3 <15/21 \le 1-3(\varepsilon+1/n) = p_{e_i,D_i}$.
\end{itemize}

 The utility of every helper voter is $18q/n+1/n-\varepsilon/3$, of every element voter is $3(\varepsilon+1/n)+1/n-\varepsilon/3$, of every toggle voter $t_i$ is $q/n+\varepsilon+1/n$ if $S_i \in \mathcal{C}$ and $q/n+2/n$ if $S_i \not \in \mathcal{C}$, and of every single-minded voter $s_i$ is $\varepsilon+1/n$ if $S_i \in \mathcal{C}$ and $2/n$ if $S_i \not \in \mathcal{C}$. Thus, the welfare of $a^\star$ is 
 \begin{align*}
     w(a^\star) ={} & 18q\bigg(\frac{18q}{n}+\frac{1}{n}-\frac{\varepsilon}{3}\bigg) + 3q\bigg(3\bigg(\varepsilon+\frac{1}{n}\bigg)+\frac{1}{n}-\frac{\varepsilon}{3} \bigg) +q \bigg(\frac{q}{n}+\varepsilon+\frac{1}{n}\bigg) + (k-q)\bigg(\frac{q}{n} + \frac{2}{n}\bigg) \\
     & + q\bigg(\varepsilon+\frac{1}{n}\bigg) + (k-q)\cdot \frac{2}{n} \\ 
     ={} &\frac{324q^{2} + kq + 28q + 4k}{n} \;+\; 4q\,\varepsilon
 \end{align*}

We now show that on any RX3C instance for which there does not exist an exact cover, all decomposable budget allocations must achieve strictly lower welfare than $\frac{324q^{2} + kq + 28q + 4k}{n} \;+\; 4q\,\varepsilon$. Note that, because decomposability implies that only $q/n$ mass can be allocated to alternative $Y$, at least $(k-q)/n$ mass from the toggle voters $t_1, \ldots, t_k$ must be allocated among subset alternatives $S_1, \ldots, S_k$. Furthermore, at most $1/n$ of the mass can be placed on each subset alternative (since no two toggle voters like the same one). In particular, this implies that at least $k-q$ subset alternatives receive more than $\varepsilon$ mass from toggle voters. To see this, suppose otherwise: that at most $k-q-1$ subset alternatives receive more than $\varepsilon$ mass from toggle voters. Then the total amount of mass allocated among subset alternatives from toggle voters is at most 
\[
\frac{1}{n}(k-q-1) + \varepsilon(q+1) < \frac{k-q}{n}-\frac{1}{n}+\frac{1}{n}=\frac{k-q}{n},
\]
where the inequality follows from $\varepsilon<\frac{1}{n(q+1)}$.

Note that every subset alternative $S_j$ receives $1/n$ mass from the corresponding single-minded voter $s_j$ in any decomposable allocation. Thus, if alternative $S_j$ additionally receives more than $\varepsilon$ mass from toggle voter $t_j$, decomposability prohibits it from receiving positive mass from any element voter $e_i$ such that $i\in S_j$, each of whom has $p_{e_i,S_j}=1/n+\varepsilon$. Therefore, there are at most $q$ subset alternatives that can receive mass from element voters. If there is no exact cover, then it is by definition the case that among any set of at most $q$ subsets, there exists at least one element $i$ not contained within any of the subsets. Therefore, the corresponding element voter $e_i$ must allocate all of its $1/n$ mass to alternative $D_i$. 

Subject to allocating at most $18q/n$ to alternative $X$ (as required by decomposability) and exactly $1/n$ to alternative $D_i$ (as argued above), the welfare-optimal decomposable allocation $a$ allocates $18q/n$ to alternative $X$, $q/n$ to alternative $Y$, $1/n-\varepsilon/3$ to all alternatives $D_j$ with $j \neq i$, and some amount $x$ with $1/n+\varepsilon \le x \le 2/n$ to each subset alternative $S_j$.

To justify optimality, observe that welfare can be maximized greedily by allocating budget to segments with the highest marginal contribution. Every unit of spending on $X$ contributes to the utility of $18q$ voters; every unit of spending up to $q/n$ on $Y$ contributes to the utility of $k$ voters; every unit of spending up to $1/n-\varepsilon/3$ on a dummy alternative contributes to the utility of $7$ voters; and every unit of spending up to $1/n+\varepsilon$ on a subset alternative contributes to the utility of $5$ voters. The allocation therefore saturates all such regions subject to decomposability. Any remaining budget is optimally spent on subset alternatives in the range $(1/n+\varepsilon,\,1-q/n]$, where each unit contributes to the utility of only $2$ voters. All other admissible marginal spending contributes to the utility of at most one voter and is therefore minimized, with the unavoidable exception of the mass allocated to $D_i$.

Computing the welfare of this allocation yields
\[
w(a)=\frac{324q^{2} + kq + 28q + 4k}{n} \;+\; \Bigl(4q-\frac{1}{3}\Bigr)\varepsilon
\;<\;
\frac{324q^{2} + kq + 28q + 4k}{n} \;+\; 4q\,\varepsilon.
\]
In particular, $w(a)=w(a^\star)-\varepsilon/3$, reflecting the fact that an amount $\varepsilon/3$ of budget is shifted from a segment of a subset alternative approved by two voters to a dummy alternative approved by only one voter.
\end{proof}

\subsection{Lower Bound on Welfare Ratio Between \GrdDec and \UtDec}\label{app:lbGrdDec}

\begin{proposition}\label{prop:lbGrdDec}
    For every even $n \ge 4$ and every $\delta > 0$, there exists a profile with $n$ voters such that \UtDec achieves at least $2-\frac{6}{n+4} - \delta$ times the welfare that \GrdDec achieves.
    In particular, this ratio can be made arbitrarily close to $2$ as $n$ grows.
\end{proposition}

\begin{figure}
    \centering
    \def\posone{0.2}
    \def\postwo{\posone + 0.15}
    \def\posthr{\postwo + 0.15}
    \def\posfou{\posthr + 0.33}
    \def\posfiv{\posfou + 0.15}
    \def\possix{\posfiv + 0.15}
    \def\possev{\possix + 0.33}
    \def\coordwidth{0.05}
    \newcommand{\drawvoteannotated}[5]{
        \draw[ultra thick, #3] (#1-\coordwidth,#2) -- (#1+\coordwidth,#2)
        node[pos=0, inner sep=0pt, anchor=east] {\small #4}
        node[pos=1, inner sep=0pt, anchor=west] {\small #5};
    }
    \newcommand{\drawcontribution}[7]{
        \fill[pattern={Lines[angle=#5*0.8,distance=4pt,line width=2pt]},pattern color=#4!80!white, draw=#4!60!white] (#1-\coordwidth+#6*\coordwidth*2,#2) rectangle (#1-\coordwidth+#7*\coordwidth*2,#3);
    }
    \begin{subfigure}[t]{.8\textwidth}
        \begin{tikzpicture}[yscale=3,xscale=6]
            \draw (0,0) -- (0,1);  
            \draw (-.02,0) -- (0.02,0); 
            \draw (-.02,1) -- (0.02,1); 
            \node[anchor=east, xshift=-5px] at (0,0) {$0$}; 
            \node[anchor=east, xshift=-5px] at (0,1) {$1$}; 
            \node[anchor=south, rotate=90] at (0,0.5) {\small Budget}; 
            \filldraw[fill=alternativebarcolor,draw=none] (\posone-\coordwidth,0) rectangle (\posone+\coordwidth,1); 
            \node at (\postwo,0.5) {$\cdots$}; 
            \filldraw[fill=alternativebarcolor,draw=none] (\posthr-\coordwidth,0) rectangle (\posthr+\coordwidth,1); 
            \filldraw[fill=alternativebarcolor,draw=none] (\posfou-\coordwidth,0) rectangle (\posfou+\coordwidth,1); 
            \filldraw[fill=alternativebarcolor,draw=none] (\possix-\coordwidth,0) rectangle (\possix+\coordwidth,1); 
            \node at (\posfiv,0.5) {$\cdots$}; 
            \filldraw[fill=alternativebarcolor,draw=none] (\possev-\coordwidth,0) rectangle (\possev+\coordwidth,1); 
            \node[anchor=north, outer sep=5pt] (label1) at (\posone,0) {$1$};
            \node[anchor=north, outer sep=5pt] at (\postwo,0) {\phantom{2}$\cdots$\phantom{2}};
            \node[anchor=north, outer sep=5pt] at (\posthr,0) {$\frac{n}{2}$};
            \node[anchor=north, outer sep=5pt] at (\posfou,0) {$\frac{n}{2}+1$};
            \node[anchor=north, outer sep=5pt] at (\posfiv,0) {\phantom{5}$\cdots$\phantom{5}};
            \node[anchor=north, outer sep=5pt] at (\possix,0) {$n$};
            \node[anchor=north, outer sep=5pt] at (\possev,0) {$n+1$};
            \node[anchor=east] at (label1.west) {\small Alternative:};
            \drawcontribution{\posone}{0}{0.22}{color2}{45}{0}{1};
            \drawcontribution{\posthr}{0}{0.22}{color3}{45}{0}{1};
            \drawcontribution{\posfou}{0}{0.22}{color1}{-45}{0}{1};
            \drawcontribution{\possix}{0}{0.22}{color1}{-45}{0}{1};
            \drawvoteannotated{\posone}{0.78}{color2}{}{$\frac{n-1}{n}$};
            \drawvoteannotated{\posthr}{0.78}{color3}{}{$\frac{n-1}{n}$};
            \drawvoteannotated{\posfou}{0.26}{color1}{$\frac{n}{2}\times$}{$\frac{1+\varepsilon}{n}$};
            \drawvoteannotated{\posfou}{0.22}{color2}{}{};
            \drawvoteannotated{\possix}{0.26}{color1}{$\frac{n}{2}\times$}{};
            \drawvoteannotated{\possix}{0.22}{color3}{}{$\frac{1}{n}$};
            \drawvoteannotated{\possev}{0.5}{color1}{$\frac{n}{2}\times$}{$\frac{1-\varepsilon}{2}$};
        \end{tikzpicture}
        \caption{Output of \GrdDec.} \label{fig:greedydecomp_social_welfare_grdydecomp}
    \end{subfigure}
    \begin{subfigure}[t]{.8\textwidth}
        \begin{tikzpicture}[yscale=3,xscale=6]
            \draw (0,0) -- (0,1);  
            \draw (-.02,0) -- (0.02,0); 
            \draw (-.02,1) -- (0.02,1); 
            \node[anchor=east, xshift=-5px] at (0,0) {$0$}; 
            \node[anchor=east, xshift=-5px] at (0,1) {$1$}; 
            \node[anchor=south, rotate=90] at (0,0.5) {\small Budget}; 
            \filldraw[fill=alternativebarcolor,draw=none] (\posone-\coordwidth,0) rectangle (\posone+\coordwidth,1); 
            \node at (\postwo,0.5) {$\cdots$}; 
            \filldraw[fill=alternativebarcolor,draw=none] (\posthr-\coordwidth,0) rectangle (\posthr+\coordwidth,1); 
            \filldraw[fill=alternativebarcolor,draw=none] (\posfou-\coordwidth,0) rectangle (\posfou+\coordwidth,1); 
            \filldraw[fill=alternativebarcolor,draw=none] (\possix-\coordwidth,0) rectangle (\possix+\coordwidth,1); 
            \node at (\posfiv,0.5) {$\cdots$}; 
            \filldraw[fill=alternativebarcolor,draw=none] (\possev-\coordwidth,0) rectangle (\possev+\coordwidth,1); 
            \node[anchor=north, outer sep=5pt] (label1) at (\posone,0) {$1$};
            \node[anchor=north, outer sep=5pt] at (\postwo,0) {\phantom{2}$\cdots$\phantom{2}};
            \node[anchor=north, outer sep=5pt] at (\posthr,0) {$\frac{n}{2}$};
            \node[anchor=north, outer sep=5pt] at (\posfou,0) {$\frac{n}{2}+1$};
            \node[anchor=north, outer sep=5pt] at (\posfiv,0) {\phantom{5}$\cdots$\phantom{5}};
            \node[anchor=north, outer sep=5pt] at (\possix,0) {$n$};
            \node[anchor=north, outer sep=5pt] at (\possev,0) {$n+1$};
            \node[anchor=east] at (label1.west) {\small Alternative:};
            \drawcontribution{\posone}{0}{0.04}{color2}{45}{0}{1};
            \drawcontribution{\posthr}{0}{0.04}{color3}{45}{0}{1};
            \drawcontribution{\posfou}{0.00}{0.04}{color1}{-45}{0}{1};
            \drawcontribution{\posfou}{0.04}{0.22}{color2}{45}{0}{1};
            \drawcontribution{\possix}{0.00}{0.04}{color1}{-45}{0}{1};
            \drawcontribution{\possix}{0.04}{0.22}{color3}{45}{0}{1};
            \drawcontribution{\possev}{0}{0.5}{color1}{-45}{0}{1};
            \drawvoteannotated{\posone}{0.78}{color2}{}{$\frac{n-1}{n}$};
            \drawvoteannotated{\posthr}{0.78}{color3}{}{$\frac{n-1}{n}$};
            \drawvoteannotated{\posfou}{0.26}{color1}{$\frac{n}{2}\times$}{$\frac{1+\varepsilon}{n}$};
            \drawvoteannotated{\posfou}{0.22}{color2}{}{};
            \drawvoteannotated{\possix}{0.26}{color1}{$\frac{n}{2}\times$}{};
            \drawvoteannotated{\possix}{0.22}{color3}{}{$\frac{1}{n}$};
            \drawvoteannotated{\possev}{0.5}{color1}{$\frac{n}{2}\times$}{$\frac{1-\varepsilon}{2}$};
        \end{tikzpicture}
        \caption{Welfare-optimal decomposable allocation.} \label{fig:greedydecomp_social_welfare_opt}
    \end{subfigure}
    \caption{Visualization of the profile and aggregates from the proof of \Cref{prop:lbGrdDec}.}
\label{fig:greedydecomp_social_welfare}
\end{figure}

\proof
    Let $n\geq 4$ be an even number and $\varepsilon > 0$.
    Consider a profile with $n$ voters and $m=n+1$ alternatives. Voters $i$ from $1$ to $\frac{n}{2}$ vote $\frac{n-1}{n}$ on alternative $i$ and $\frac{1}{n}$ on alternative $i+\frac{n}{2}$. All remaining voters vote $\frac{1+\varepsilon}{n}$ on alternatives $\frac{n}{2}+1$ to $n$ and $\frac{1-\varepsilon}{2}$ on alternative $n+1$. \Cref{fig:greedydecomp_social_welfare} illustrates this profile.

    \GrdDec first distributes the entire budget from the second half of the voters over alternatives $\frac{n}{2}+1$ to $n$ and then has to spend the budget of each remaining voter $i$ fully on alternative $i$ (compare \Cref{fig:greedydecomp_social_welfare_grdydecomp}), leading to the allocation $a = \big(\frac{1}{n}, \dots, \frac{1}{n}, 0\big)$ with welfare $w(a) = \frac{n}{2} \cdot \frac{1}{n} + \frac{n}{2} \cdot \big(\frac{n}{2} + 1\big) \cdot \frac{1}{n} = \frac{n}{4} + 1$.

    A better decomposable allocation is $a' = \big(\frac{\varepsilon}{n}, \dots, \frac{\varepsilon}{n}, \frac{1}{n}, \dots, \frac{1}{n}, \frac{1-\varepsilon}{2}\big)$ with welfare $w(a') = \frac{n}{2} \cdot \frac{\varepsilon}{n} + \frac{n}{2} \cdot \big(\frac{n}{2} + 1\big) \cdot \frac{1}{n} + \frac{n}{2} \cdot \frac{1-\varepsilon}{2} = \frac{n}{2} + \frac{1}{2} - \varepsilon \cdot \big(\frac{n}{4} - \frac{1}{2}\big)$. 
    To see that $a'$ is decomposable, consider the voter contributions (compare \Cref{fig:greedydecomp_social_welfare_opt}) $c\in [0,1]^{n\times m}$ given by 
    \begin{align*}
        c_{i,j} & = \begin{cases} \frac{\varepsilon}{n} & \text{if } j=i,\\
        \frac{1-\varepsilon}{n} & \text{if } j=i+\frac{n}{2},\\
        0 & \text{otherwise,}\end{cases} \quad \text{for } i \in \textstyle [\frac{n}{2}], j\in [m],\\ 
        c_{i,j} & =  \begin{cases} \frac{\varepsilon}{n} & \text{if } j=i,\\
        \frac{1-\varepsilon}{n} & \text{if } j=n+1,\\
        0 & \text{otherwise,}\end{cases} \quad \text{for } i\in \textstyle \{\frac{n}{2}+1,\ldots,n\}, j\in [m].
    \end{align*}
    
    Thus, \GrdDec achieves a welfare approximation with respect to \UtDec no better than 
    \[
        \frac{\frac{n}{2} + \frac{1}{2} - \varepsilon \cdot \big(\frac{n}{4} - \frac{1}{2}\big)}{\frac{n}{4} + 1} \stackrel{\varepsilon \rightarrow 0}{=} \frac{\frac{n}{2} + \frac{1}{2}}{\frac{n}{4} + 1} = 2 - \frac{6}{n+4}.\qedhere
    \]
\endproof

\newpage

\section{Supplementary Material}

\subsection{Overview of the Moving-Phantom Mechanisms from the Literature}
\label{app:mp-mechs}

\newcommand{\plotPhantom}[2]{
    \draw[phantom] (0,#1) to (100,#1) node[phlabel]{#2};
}
\newcommand{\plotPhantoms}[2]{
    \draw[phantoms] (0,#1) to (100,#1) node[phlabel]{#2};
}
\newcommand{\plotMovingPhantom}[3]{
    \plotPhantom{#1}{#2}
    \draw[phantom, -{Latex[length=3,width=3]}, opacity=#3] (50,#1) -- (50,#1+2+5*#3); 
}
\newcommand{\plotMovingPhantoms}[3]{
    \plotPhantoms{#1}{#2}
    \draw[phantom, -{Latex[length=3,width=3]}, opacity=#3] (50,#1) -- (50,#1+2+5*#3); 
}

\newcommand{\plotPhantomSystemChart}[1]{
    \adjustbox{valign=c}{\begin{tikzpicture}[xscale=0.005, yscale=0.035]
        \draw (-5,0) -- (-5,100);  
        \draw (-15,0) -- (5,0); 
        \draw (-15,100) -- (5,100); 
        \node[anchor=east, xshift=0px] at (0,0) {$0$}; 
        \node[anchor=east, xshift=0px] at (0,100) {$1$}; 
        \phantom{\plotPhantoms{100}{$f_0$--$f_2$};}
        #1
    \end{tikzpicture}}
}
\newcommand{\plotPhantomSystemChartNoLabels}[1]{
    \adjustbox{valign=c}{\begin{tikzpicture}[xscale=0.005, yscale=0.035]
        \draw (-5,0) -- (-5,100);  
        \draw (-15,0) -- (5,0); 
        \draw (-15,100) -- (5,100); 
        \phantom{\plotPhantoms{100}{$f_0$--$f_2$};}
        #1
    \end{tikzpicture}}
}

\begin{longtable}{c@{\hspace{5pt}}c@{\hspace{8pt}}m{4.7cm}}
    \caption{Overview of the moving-phantom mechanisms from the literature.} \\
    \toprule
    \begin{tabular}{@{}c@{}}\textbf{Name} \\ \textbf{Definition}\end{tabular} & \textbf{Visualization} & \textbf{Description} \\
    \midrule
    \endhead
    
    \begin{tabular}{@{}c@{}}\textbf{\Constant} \\[10pt]
        $\displaystyle f_k(t) = t$
    \end{tabular} & 
    \plotPhantomSystemChart{
        \plotMovingPhantoms{100/6}{$f_0$--$f_6$} {1};
    }\hspace*{-3pt}
    \plotPhantomSystemChartNoLabels{
        \plotMovingPhantoms{45}{$f_0$--$f_6$}{1};
    }&
    All phantoms move simultaneously. Normalization is always reached when all phantoms reach position $\frac{1}{m}.$\\

    \begin{tabular}{@{}c@{}}\textbf{\GreedyMax} \\[0pt]
        \citep{berg2024truthful} \\[10pt]
        $\displaystyle f_k(t) = t \cdot \min(1,n-k)$
    \end{tabular} & 
    \plotPhantomSystemChart{
        \plotPhantom{0}{$f_6$}
        \plotMovingPhantoms{100/6}{$f_0$--$f_5$}{1};
    }\hspace*{-3pt}
    \plotPhantomSystemChartNoLabels{
        \plotPhantom{0}{$f_6$}
        \plotMovingPhantoms{45}{$f_0$--$f_5$}{1};
    }&
    The top $n$ phantoms move simultaneously, while the last one stays at zero. \\

    \begin{tabular}{@{}c@{}}\textbf{\Fan} \\[10pt]
        $\displaystyle f_k(t) = \min\Big(\frac{n-k}{n}, t\Big)$
    \end{tabular} & 
    \plotPhantomSystemChart{
        \plotPhantom{0}{$f_6$}
        \plotPhantom{1*100/6}{$f_5$};
        \plotMovingPhantoms{1*100/6+10}{$f_0$--$f_4$}{1};
    }\hspace*{-3pt}
    \plotPhantomSystemChartNoLabels{
        \plotPhantom{0*100/6}{$f_6$};
        \plotPhantom{1*100/6}{$f_5$};
        \plotPhantom{2*100/6}{$f_4$};
        \plotPhantom{3*100/6}{$f_3$};
        \plotMovingPhantoms{3*100/6+10}{$f_0$--$f_2$}{1};
    }&
    All phantoms move simultaneously with each phantom stopping at position $\frac{n-k}{n}$. \\
    
    \begin{tabular}{@{}c@{}}\textbf{\IM} \\[0pt]
        \citep{freeman2021truthful} \\[10pt]
        $\displaystyle f_k(t) = t \cdot \frac{n-k}{n}$
    \end{tabular} & 
    \plotPhantomSystemChart{
        \def\posone{7}
        \plotPhantom{0}{$f_6$}
        \plotMovingPhantom{\posone*1}{$f_5$}{0.5};
        \plotMovingPhantom{\posone*2}{$f_4$}{0.6};
        \plotMovingPhantom{\posone*3}{$f_3$}{0.7};
        \plotMovingPhantom{\posone*4}{$f_2$}{0.8};
        \plotMovingPhantom{\posone*5}{$f_1$}{0.9};
        \plotMovingPhantom{\posone*6}{$f_0$}{1.0};
    }\hspace*{-3pt}
    \plotPhantomSystemChartNoLabels{
        \def\postwo{14}
        \plotPhantom{0}{$f_6$}
        \plotMovingPhantom{\postwo*1}{$f_5$}{0.5};
        \plotMovingPhantom{\postwo*2}{$f_4$}{0.6};
        \plotMovingPhantom{\postwo*3}{$f_3$}{0.7};
        \plotMovingPhantom{\postwo*4}{$f_2$}{0.8};
        \plotMovingPhantom{\postwo*5}{$f_1$}{0.9};
        \plotMovingPhantom{\postwo*6}{$f_0$}{1.0};
    }&
    All phantoms move simultaneously but at different speeds, until each phantom $k$ reaches its final position $\frac{n-k}{n}$ at time $t=1$. \\

    \begin{tabular}{@{}c@{}}\textbf{\Ladder} \\[0pt]
        \citep{freeman2024project} \\[10pt]
        $\displaystyle f_k(t) = \max\Big(t - \frac{k}{n},0\Big)$
    \end{tabular} & 
    \plotPhantomSystemChart{
        \def\dist{100/6}
        \plotPhantoms{0}{$f_2$--$f_6$};
        \plotMovingPhantom{8}{$f_1$}{1};
        \plotMovingPhantom{8+1*\dist}{$f_0$}{1};
    }\hspace*{-3pt}
    \plotPhantomSystemChartNoLabels{
        \def\dist{100/6}
        \plotPhantoms{0}{$f_5$--$f_6$};
        \plotMovingPhantom{10+0*\dist}{$f_4$}{1};
        \plotMovingPhantom{10+1*\dist}{$f_3$}{1};
        \plotMovingPhantom{10+2*\dist}{$f_2$}{1};
        \plotMovingPhantom{10+3*\dist}{$f_1$}{1};
        \plotMovingPhantom{10+4*\dist}{$f_0$}{1};
    }&
    All phantoms move at the same speed, but start at different times, such that each phantom $k$ reaches its final position $\frac{n-k}{n}$ at time $t=1$. \\

    \begin{tabular}{@{}c@{}}\textbf{\PiecewiseUniform} \\[0pt]
        \citep{caragiannis2022truthful} \\[10pt]
        \multicolumn{1}{l}{For $t < \frac{1}{2}$}\\[10pt]
        $
            f_k(t) = 
            \begin{cases}
                \frac{4t(n-k)}{n}-2t  & \text{ if } \frac{k}{n} \le \frac{1}{2}, \\
                0 & \text{ if } \frac{k}{n} > \frac{1}{2}, 
            \end{cases}
        $\\[20pt]
        \multicolumn{1}{l}{and for $t \ge \frac{1}{2}$} \\[10pt]
        $
            f_k(t) = 
            \begin{cases}
                \frac{(n-k) (3-2t)}{n} -2 + 2t \!\!\!\!  & \text{ if } \frac{k}{n} \le \frac{1}{2}, \\
                \frac{(n-k) (2t-1)}{n} & \text{ if } \frac{k}{n} > \frac{1}{2}.
            \end{cases}
        $ 
    \end{tabular} & 
    \plotPhantomSystemChart{
        \def\distone{30}
        \plotPhantoms{0}{$f_3$--$f_6$};
        \plotMovingPhantom{1*\distone}{$f_2$}{0.6};
        \plotMovingPhantom{2*\distone}{$f_1$}{0.8};
        \plotMovingPhantom{3*\distone}{$f_0$}{1};
    }\hspace*{-3pt}
    \plotPhantomSystemChartNoLabels{
        \def\distone{25.333}
        \def\distwo{8}
        \plotPhantom{0}{$f_6$};
        \plotMovingPhantom{1*\distwo}{$f_5$}{0.6};
        \plotMovingPhantom{2*\distwo}{$f_4$}{0.8};
        \plotMovingPhantom{3*\distwo}{$f_3$}{1};
        \plotMovingPhantom{3*\distwo+1*\distone}{$f_2$}{0.8};
        \plotMovingPhantom{3*\distwo+2*\distone}{$f_1$}{0.6};
        \plotPhantom{3*\distwo+3*\distone}{$f_0$};
    }&
    In the first phase the uppermost half of the phantoms are spread uniformly across $[0,1]$. In the second phase the lower half of the phantoms are spread across $[0, \frac{1}{2}]$ while concentrating the upper phantoms into $[\frac{1}{2}, 1]$. \\

    \begin{tabular}{@{}c@{}}\textbf{\Ut} \\[0pt]
        \citep{freeman2021truthful} \\[10pt]
        $
            f_k(t) = \max(0, \min(1, t(n+1)-k))
        $
    \end{tabular} & 
    \plotPhantomSystemChart{
        \plotPhantoms{0}{$f_1$--$f_6$};
        \plotMovingPhantom{70}{$f_0$}{1};
    }\hspace*{-3pt}
    \plotPhantomSystemChartNoLabels{
        \plotPhantoms{0}{$f_4$--$f_6$};
        \plotMovingPhantom{45}{$f_3$}{1};
        \plotPhantoms{100}{$f_0$--$f_2$};
    }&
    The phantoms move consecutively from zero to one. \\

    \begin{tabular}{@{}c@{}}\textbf{\UtP} \\[10pt]
        $
            f_k(t) = \max\big(0, \min\big(\frac{n-k}{n}, t(n+1)-k\big)\big)
        $
    \end{tabular} & 
    \plotPhantomSystemChart{
        \plotPhantoms{0}{$f_1$--$f_6$};
        \plotMovingPhantom{70}{$f_0$}{1};
    }\hspace*{-3pt}
    \plotPhantomSystemChartNoLabels{
        \plotPhantoms{0}{$f_4$--$f_6$};
        \plotMovingPhantom{28}{$f_3$}{1};
        \plotPhantom{4*100/6}{$f_2$};
        \plotPhantom{5*100/6}{$f_1$};
        \plotPhantom{100}{$f_0$};
    }&
    The phantoms move consecutively from zero to $\frac{n-k}{n}$. \\
    
    \bottomrule
    \label{tab:phantom_overview}
\end{longtable}

\subsection{Decomposability of the Input-Dependent Phantom-Order Rule}
\label{sec:neutrality-necessary}

\changed{We consider the Input-Dependent Phantom-Order rule of~\citet{bei2026individual}, defined for $n=2$ and $m=3$. The rule is defined as follows.\footnote{We present the rule via its closed-form expression~\citep[Lemma B.2.]{bei2026individual}.} 
Let $s$ be the alternative on which the votes differ by the most, i.e., $s \in \argmax_{j \in [3]} |p_{1,j}-p_{2,j}|$, with ties broken by a fixed order. 
Write $r=s-1$ and $t=s+1$, with indices interpreted cyclically. For each coordinate $j \in \{ r, s, t \}$, let $L_j= \min\{  p_{1,j}, p_{2,j} \}$ and $U_j = \med \{ p_{1,j}, p_{2,j}, \frac{1}{2} \}$. The output of the rule $a^{\text{IP}}$ is given by $a^{\text{IP}}_r = U_r$ and 
\[
(a^{\text{IP}}_s,a^{\text{IP}}_t) = 
\begin{cases}
    (1-U_r-L_t,L_t) & \text{if } U_r+U_s+L_t \ge 1 \\
   (U_s,1-U_r-U_s) & \text{otherwise}.
\end{cases}
\]
}

\changed{Suppose without loss of generality that $p_{1,s} \le p_{2,s}$. Since $s$ maximizes the absolute distance between the reports, normalization implies that $p_{1,r} \ge p_{2,r}$ and $p_{1,t} \ge p_{2,t}$. Note that $L_t = p_{2,t}$. Furthermore,~\citet[Theorem 4.7]{bei2026individual} show that the rule is Pareto efficient and therefore upper- and lower-range-respecting.
} 

\changed{First, suppose that $U_r+U_s+p_{2,t} \ge 1$. Then the outcome is $a^{\text{IP}} = (U_r, 1-U_r-p_{2,t}, p_{2,t})$. Note that both agents approve the spending on alternative $t$. 
}

\changed{If $U_r \le \frac{1}{2}$ and $1-U_r-p_{2,t} \le \frac{1}{2}$, then $a^{\text{IP}}$ is decomposable: agent 1 (respectively agent 2) can fully fund the allocation of alternative $r$ (respectively alternative $s$) since that agent casts the highest vote on that alternative, and the remaining budgets of each agent can be used to fund alternative $t$. 
}

\changed{If $U_r > \frac{1}{2}$, then, by definition of $U_r$, it must be the case that both agents approve the spending on alternative $r$. Agent 2 can be used to fully fund $a^{\text{IP}}_s=1-U_r-L_t<\frac{1}{2}$ and the remaining budgets of both agents can be used to fund alternatives $r$ and $t$.
}

\changed{If $1-U_r-p_{2,t} > \frac{1}{2}$ then we have $U_s \ge a^{\text{IP}}_s = 1-U_r-p_{2,t} > \frac{1}{2}$ (the first inequality follows from the fact that the rule is in the first case of its definition). Similar to the previous paragraph, we know that both agents approve of the spending on alternative $s$. Agent 1 can therefore be used to fully fund $a^{\text{IP}}_r = U_r < \frac{1}{2}$, and the remaining budgets of both agents can be used to fund alternatives $s$ and $t$.
}

\changed{Second, suppose that $U_r+U_s+p_{2,t} < 1$. Then the outcome is $a^{\text{IP}} = (U_r, U_s, 1-U_r-U_s)$. In this case, $U_s \ge \frac{1}{2}$. To see this, suppose otherwise. Then $U_s = \med \{ p_{1,s}, p_{2,s}, \frac{1}{2} \} < \frac{1}{2}$, which implies that $U_s=p_{2,s}$ by the assumption that $p_{1,s} \le p_{2,s}$. But, because $p_{1,r} \ge p_{2,r}$, we have that $U_r \ge p_{2,r}$. Therefore, $U_r+U_s+p_{2,t} \ge p_{2,r} + p_{2,s}+p_{2,t}=1$, contradicting the assumption that $U_r+U_s+p_{2,t} < 1$. It follows that $a^{\text{IP}}_r+a^{\text{IP}}_t = 1-U_s \le \frac{1}{2}$. Therefore, agent 1 can fully fund the allocations on alternatives $r$ and $t$. For alternative $s$, if $U_s=\frac{1}{2}$, then it can be fully funded by agent 2. If $U_s > \frac{1}{2}$ then it must be the case that $U_s = p_{1,s}$, in which case it can be funded by both agents 1 and 2.
}

\changed{In all cases, the outcome $a^{\text{IP}}$ is decomposable.}

\subsection{Pareto Suboptimality of \GrdDec and \UtDec}\label{app:paretoGrdDecUtDec}

\begin{proposition} \label{prop:greedydecomp_pareto_dominated}
    Neither \GrdDec nor \UtDec is Pareto-optimal.
\end{proposition}

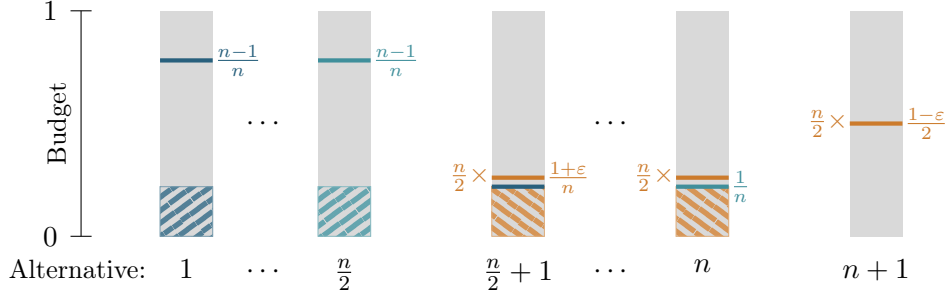
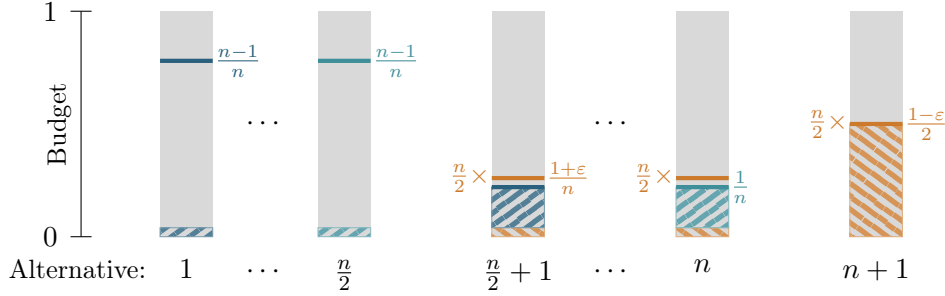
\begin{figure}
    \centering
    \def\coorddist{0.2}
    \def\coordwidth{0.06}
    \def\shift{0.0093}
    \newcommand{\drawvoteannotated}[6]{
        \draw[ultra thick, #3] (#1*\coorddist-\coordwidth,#2) -- (#1*\coorddist+\coordwidth,#2)
        node[pos=0, inner sep=0pt, anchor=east, yshift=#6] {\small #4}
        node[pos=1, inner sep=0pt, anchor=west, yshift=#6] {\small #5};
    }
    \newcommand{\drawcontribution}[5]{
        \fill[pattern={Lines[angle=#5*0.8,distance=4pt,line width=2pt]},pattern color=#4!80!white, draw=#4!60!white] (#1*\coorddist-\coordwidth,#2) rectangle (#1*\coorddist-\coordwidth+\coordwidth*2,#3);
    }
    \begin{subfigure}[t]{.48\textwidth}
        \begin{tikzpicture}[yscale=3,xscale=6]
            \draw (0,0) -- (0,1);
            \draw (-.02,0) -- (0.02,0);
            \foreach \x in {1,...,8}{
                \draw (-.01,\x/9) -- (0.01,\x/9);
            }
            \draw (-.02,1) -- (0.02,1);
            \node[anchor=east, xshift=-5px] at (0,0) {$0$};
            \node[anchor=east, xshift=-5px] at (0,1) {$1$};
            \node[anchor=south, rotate=90] at (0,0.5) {\small Budget};
            \filldraw[fill=alternativebarcolor,draw=none]
                (1*\coorddist-\coordwidth,0)
                rectangle
                (1*\coorddist+\coordwidth,1);
            \filldraw[fill=alternativebarcolor,draw=none]
                (2*\coorddist-\coordwidth,0)
                rectangle
                (2*\coorddist+\coordwidth,1);
            \filldraw[fill=alternativebarcolor,draw=none]
                (3*\coorddist-\coordwidth,0)
                rectangle
                (3*\coorddist+\coordwidth,1);
            \filldraw[fill=alternativebarcolor,draw=none]
                (4*\coorddist-\coordwidth,0)
                rectangle
                (4*\coorddist+\coordwidth,1);
            \node[anchor=north east, outer sep=4.4pt]
                at (1*\coorddist,0) {\small Alternative:};
            \node[anchor=north, outer sep=5pt]
                at (1*\coorddist,0) {$1$};
            \node[anchor=north, outer sep=5pt]
                at (2*\coorddist,0) {$2$};
            \node[anchor=north, outer sep=5pt]
                at (3*\coorddist,0) {$3$};
            \node[anchor=north, outer sep=5pt]
                at (4*\coorddist,0) {$4$};
            \drawcontribution{1}{0}{1/9}{color1}{45};
            \drawcontribution{2}{0}{1/9}{color1!50!black}{45};
            \drawcontribution{3}{0}{1/9}{color4}{45};
            \drawcontribution{4}{0}{3/9}{color2}{45};
            \drawcontribution{4}{3/9}{6/9}{color3}{-45};
            \drawvoteannotated
                {1}{1/9}{color2}{$3\times$}{}{0};
            \drawvoteannotated
                {4}{8/9+\shift}{color2}{$3\times$}{}{2.5pt};
            \drawvoteannotated
                {2}{1/9}{color3}{$3\times$}{}{0};
            \drawvoteannotated
                {4}{8/9-\shift}{color3}{$3\times$}{}{-2.5pt};
            \drawvoteannotated
                {1}{2/9}{color1}{}{}{0};
            \drawvoteannotated
                {3}{7/9-\shift}{color1}{}{}{0};
            \drawvoteannotated
                {2}{2/9}{color1!50!black}{}{}{0};
            \drawvoteannotated
                {3}{7/9+\shift}{color1!50!black}{}{}{0};
            \drawvoteannotated
                {3}{1}{color4}{}{}{0};
        \end{tikzpicture}
        \caption{Output of \GrdDec and \UtDec.}
        \label{fig:greedydecomp_pareto_dominated_greedydecomp}
    \end{subfigure}
    \hfill
    \begin{subfigure}[t]{.48\textwidth}
        \begin{tikzpicture}[yscale=3,xscale=6]
            \draw (0,0) -- (0,1);
            \draw (-.02,0) -- (0.02,0);
            \foreach \x in {1,...,8}{
                \draw (-.01,\x/9) -- (0.01,\x/9);
            }
            \draw (-.02,1) -- (0.02,1);
            \node[anchor=east, xshift=-5px] at (0,0) {$0$};
            \node[anchor=east, xshift=-5px] at (0,1) {$1$};
            \node[anchor=south, rotate=90] at (0,0.5) {\small Budget};
            \filldraw[fill=alternativebarcolor,draw=none]
                (1*\coorddist-\coordwidth,0)
                rectangle
                (1*\coorddist+\coordwidth,1);
            \filldraw[fill=alternativebarcolor,draw=none]
                (2*\coorddist-\coordwidth,0)
                rectangle
                (2*\coorddist+\coordwidth,1);
            \filldraw[fill=alternativebarcolor,draw=none]
                (3*\coorddist-\coordwidth,0)
                rectangle
                (3*\coorddist+\coordwidth,1);
            \filldraw[fill=alternativebarcolor,draw=none]
                (4*\coorddist-\coordwidth,0)
                rectangle
                (4*\coorddist+\coordwidth,1);
            \node[anchor=north east, outer sep=4.4pt]
                at (1*\coorddist,0) {\small Alternative:};
            \node[anchor=north, outer sep=5pt]
                at (1*\coorddist,0) {$1$};
            \node[anchor=north, outer sep=5pt]
                at (2*\coorddist,0) {$2$};
            \node[anchor=north, outer sep=5pt]
                at (3*\coorddist,0) {$3$};
            \node[anchor=north, outer sep=5pt]
                at (4*\coorddist,0) {$4$};
            \drawcontribution{3}{0}{2/9}{darkgray}{45};
            \drawcontribution{4}{0}{7/9}{darkgray}{45};
            \drawvoteannotated
                {1}{1/9}{color2}{$3\times$}{}{0};
            \drawvoteannotated
                {4}{8/9+\shift}{color2}{$3\times$}{}{2.5pt};
            \drawvoteannotated
                {2}{1/9}{color3}{$3\times$}{}{0};
            \drawvoteannotated
                {4}{8/9-\shift}{color3}{$3\times$}{}{-2.5pt};
            \drawvoteannotated
                {1}{2/9}{color1}{}{}{0};
            \drawvoteannotated
                {3}{7/9-\shift}{color1}{}{}{0};
            \drawvoteannotated
                {2}{2/9}{color1!50!black}{}{}{0};
            \drawvoteannotated
                {3}{7/9+\shift}{color1!50!black}{}{}{0};
            \drawvoteannotated
                {3}{1}{color4}{}{}{0};
        \end{tikzpicture}
        \caption{Pareto-dominating allocation.}
        \label{fig:greedydecomp_pareto_dominated_opt}
    \end{subfigure}
    \caption{Visualization of the profile and aggregates from the proof of \Cref{prop:greedydecomp_pareto_dominated}}
    \label{fig:greedydecomp_pareto_dominated}
\end{figure}

\proof
    Consider the profile $P$ with $n = 9$ and $m = 4$ defined as 
    \[
    P \;=\;
    \begin{pmatrix}
        \sfrac{1}{9} & 0 & 0 & \sfrac{8}{9}\\
        \sfrac{1}{9} & 0 & 0 & \sfrac{8}{9}\\
        \sfrac{1}{9} & 0 & 0 & \sfrac{8}{9}\\
        0 & \sfrac{1}{9} & 0 & \sfrac{8}{9}\\
        0 & \sfrac{1}{9} & 0 & \sfrac{8}{9}\\
        0 & \sfrac{1}{9} & 0 & \sfrac{8}{9}\\
        \sfrac{2}{9} & 0 & \sfrac{7}{9} & 0\\
        0 & \sfrac{2}{9} & \sfrac{7}{9} & 0\\
        0 & 0 & 1 & 0\\
    \end{pmatrix}.
    \] \Cref{fig:greedydecomp_pareto_dominated} illustrates this profile.
    
    \GrdDec first exhausts the budget of voters $1$ to $6$ on alternative 4, and then spends the budget of voters $7$ and $8$ on alternatives $1$ and $2$, respectively. Finally, the budget of voter $9$ is spent on alternative $3$. This leads to the allocation $\big(\frac{1}{9}, \frac{1}{9}, \frac{1}{9}, \frac{6}{9}\big)$ (see \Cref{fig:greedydecomp_pareto_dominated_greedydecomp}). Note that this is also the unique optimal decomposable allocation for this profile, as no decomposable mechanism can spend more than $\frac{6}{9}$ on alternative $4$ or more than $\frac{1}{9}$ on alternatives $1$ and $2$.
    However, this outcome is Pareto-dominated by the non-decomposable allocation \changed{$\big(0, 0, \frac{2}{9}, \frac{7}{9}\big)$} (see \Cref{fig:greedydecomp_pareto_dominated_opt}), where all voters get as much utility as before, except for voter $9$, whose utility strictly increases.
\endproof

\subsection{IP Formulation for Computing an Optimal Decomposable Allocation}\label{app:IPDecomp}

For a fixed profile $P=(p_{i,j})_{i\in [n],j\in [m]}\in \mathcal{P}_{n,m}$, we consider the following integer linear program $\mathrm{IP}(P)$ with variables $x\in \{0,1\}^{n\times m}$ and $c,u\in [0,1]^{n\times m}$: 
\allowdisplaybreaks
\begin{align}
    \max &\sum_{i \in [n]} \sum_{j \in [m]} u_{i,j} \nonumber\\ 
    &\sum_{j \in [m]} c_{i,j} = \frac{1}{n} & \text{ for all } i \in [n],\label{eq:ilp-decomp-1}\\ 
    & c_{i,j} \leq x_{i,j} & \text{ for all } i \in [n], j \in [m],\label{eq:ilp-decomp-2}\\ 
    &\sum_{i' \in [n]} c_{i',j} + (1-p_{i,j}) \cdot x_{i,j} \leq 1 & \text{ for all } i \in [n], j \in [m],\label{eq:ilp-decomp-3}\\ 
    & u_{i,j} \leq \sum_{i' \in [n]} c_{i',j} &\text{ for all } i \in [n], j \in [m],\label{eq:ilp-decomp-4}\\ 
    & u_{i,j} \leq p_{i,j} &\text{ for all } i \in [n], j \in [m].\label{eq:ilp-decomp-5}
\end{align}

\begin{proposition}\label{prop:ilp-decomp}
    For a profile $P\in \mathcal{P}_{n,m}$ and $c\in [0,1]^{n\times m}$, 
    $a\in \Delta^{(m-1)}$ \changed{defined as $a_j=\sum_{i\in [n]}c_{i,j}$ for every $j\in[m]$} is a decomposable allocation with voter contributions $c$ if and only if there exist $x\in \{0,1\}^{n\times m}$ and $u\in [0,1]^{n\times m}$ such that $(x,c,u)$ is a feasible solution  for $\mathrm{IP}(P)$. 
\end{proposition}

\proof
    We fix $P\in \mathcal{P}_{n,m}$ and $c\in [0,1]^{n\times m}$.
    For the forward direction, we let $a\in \Delta^{(m-1)}$ be a decomposable allocation with voter contributions $c\in [0,1]^{n\times m}$; i.e., $\sum_{j\in [m]}c_{i,j}=\frac{1}{n}$ for all $i\in [n]$, $\sum_{i\in [n]}c_{i,j}=a_j$ for all $j\in [m]$, and $a_j\leq p_{i,j}$ for all $i\in[n],j\in[m]$ with $c_{i,j}>0$. 
    We define $x\in \{0,1\}^{n\times m}$ and $u\in [0,1]^{n\times m}$ as follows:
    \[
        x_{i,j} = \begin{cases}
            1 & \text{if } c_{i,j} > 0,\\
            0 & \text{otherwise,}
        \end{cases} \qquad
        u_{i,j} = \min\{ a_{j}, p_{i,j} \}, \qquad \text{for all } i\in [n], j\in [m].
    \]
    We claim that $(x,c,u)$ is a feasible solution for $\mathrm{IP}(P)$.
    That this solution satisfies constraints \eqref{eq:ilp-decomp-1}, \eqref{eq:ilp-decomp-2},  \eqref{eq:ilp-decomp-4} and \eqref{eq:ilp-decomp-5} is straightforward from its definition.
    For \eqref{eq:ilp-decomp-3}, we fix $i\in [n]$ and $j\in [m]$, and distinguish two cases.
    If $x_{i,j}=1$, the constraint becomes $\sum_{i'\in [n]}c_{i',j}\leq p_{i,j}$. Since $c$ certifies decomposability, this is equivalent to $a_{j}\leq p_{i,j}$ and holds true.
    If $x_{i,j}=0$, the constraint becomes $\sum_{i'\in [n]}c_{i',j}\leq 1$, which is equivalent to $a_{j}\leq 1$ and thus holds trivially by the feasibility of $a$.

    For the backward direction, we let $(x,c,u)$ be a feasible solution for $\mathrm{IP}(P)$ and consider $a\in [0,1]^{m}$ defined by $a_j=\sum_{i\in [n]} c_{i,j}$ for every $j\in [m]$.
    From the fact that $c\in [0,1]^{n\times m}$ and constraint \eqref{eq:ilp-decomp-1}, it follows that $a\in \Delta^{(m-1)}$, i.e., $a$ is a feasible allocation. 
    We prove in what follows that $c$ are voter contributions that certify its decomposability. 
    Two of the three conditions follow directly from constraint  \eqref{eq:ilp-decomp-1} and the definition of $a$: that $\sum_{j\in [m]}c_{i,j}=\frac{1}{n}$ for all $i\in [n]$, and that $\sum_{i\in [n]}c_{i,j}=a_j$ for all $j\in [m]$.
    It only remains to show that $a_j\leq p_{i,j}$ for all $i\in[n],j\in[m]$ with $c_{i,j}>0$. 
    Fix $i\in [n]$ and $j\in [m]$ such that $c_{i,j}>0$. 
    Constraint \eqref{eq:ilp-decomp-2} implies $x_{i,j}=1$, and thus constraint \eqref{eq:ilp-decomp-3} yields $p_{i,j} \geq \sum_{i'\in [n]}c_{i',j}=a_j$, where the last equality follows from the definition of $a$.
\endproof

From our definition of voter utilities, captured precisely by $\sum_{j\in [m]}u_{i,j}$ for each voter $i$ due to constraints \eqref{eq:ilp-decomp-4} and \eqref{eq:ilp-decomp-5}, we obtain the following corollary, stating an IP-based characterization of welfare-optimal decomposable allocations. 
\begin{corollary}
    For any $P\in \mathcal{P}_{n,m}$, $a^\star\in \Delta^{(m-1)}$ is a welfare-optimal decomposable allocation if and only if there exists an optimal solution $(x^\star,c^\star,u^\star)$ for $\mathrm{IP}(P)$ such that $a^\star_j=\sum_{i\in [n]} c^\star_{i,j}$ for every $j\in [m]$.
\end{corollary}

\subsection{The Weighted Setting} \label{app:weights}

In the weighted setting, each voter $i \in [n]$ has a weight $\omega_i \in \N$, which intuitively indicates the voter's influence over the outcome. In the donor coordination interpretation (compare \citet{BBG+22a,brandt2023balanced}), the weights represent voters' contributions to a shared budget, which is then allocated across a set of charities. We denote the total weight (i.e., the total available budget) as $b = \sum_{i \in [n]} \omega_i$. In the remainder of this section, we introduce weighted versions of our mechanisms and axioms and show that all results from this paper extend to the weighted setting.

\paragraph{Weighted Mechanisms} 

Every mechanism from the unweighted setting can be transformed into a weighted mechanism with the following simple construction. Given a weighted profile $P$ with $n$ voters and weights $\omega$, we construct the unweighted profile $P'$ with $n' = b = \sum_{i \in [n]} \omega_i$ voters by duplicating each voter $i$ exactly $\omega_i$ times.\footnote{\citet{becker2025efficiently} use the same construction to lift their results for the computational complexity of computing Nash equilibria into the weighted setting. In our model, their results imply the existence and a pseudo-polynomial time computation of decomposable allocations in the weighted setting.} We then run the unweighted mechanism on $P'$ and use the resulting allocation as the output for the weighted profile $P$. For this section, let $\mathbb{P}(P, \omega)$ denote the unweighted profile constructed from a weighted profile $P$ with weights $\omega$.

We introduce a generalization of the class of moving-phantom mechanisms to the weighted setting, which we then show to be equivalent to the class obtained by applying moving-phantom mechanisms as explained above. This generalization is based on the following key observation. An unweighted profile $\mathbb{P}(P, \omega)$ constructed from a weighted profile potentially has a lot of duplicate voters. When computing the median of these voters and some moving phantoms on an alternative $j$, many of these phantoms can never be the median (unless they coincide with other phantoms or voters). As an example, consider a profile with one voter voting $0.2$ and two voters voting $0.4$ on some alternative $j$. On this alternative, there exists no placement of four phantoms $f_0$ to $f_3$, such that $f_2$ is the median (without being at the same position as another phantom or voter). By removing such phantoms from consideration, we can define weighted moving-phantom mechanisms as follows.

\paragraph{Weighted Moving Phantoms} A weighted moving-phantom mechanism is defined via one phantom
system for each value of the total budget $b$. For $b \in \N$, let $\F_b = \{f_\kappa : \kappa \in [b]_0 \}$ be a phantom system. Then, a family of phantom systems $\F = \{\F_b \mid b \in \N\}$ defines the weighted moving-phantom mechanism $\B^\F$ as follows. For a given profile $P = (p_1, \dots, p_n)$ with $n$ voters and weights $\omega = (\omega_1, \dots, \omega_n)$ totaling $\sum_{i \in [n]} \omega_i = b$, let $t^\star \in [0,1]$ be a value for which
\[\sum_{j \in [m]} \med(f_{\Omega_{0,j}}(t^\star), f_{\Omega_{1,j}}(t^\star), \dots, f_{\Omega_{n,j}}(t^\star), p_{1,j}, \dots, p_{n,j}) = 1,\]
where $\Omega_{i,j}$ is the total weight of voters with the first to $i^{\textrm{th}}$ lowest votes on alternative $j$ (with some arbitrary tie-breaking). The weighted moving-phantom mechanism $\B^\F$ returns the allocation $a$ with $a_j = \med(f_{\Omega_{0,j}}(t^\star), f_{\Omega_{1,j}}(t^\star), \dots, f_{\Omega_{n,j}}(t^\star), p_{1,j}, \dots, p_{n,j})$ for each alternative $j$.

We can show that running a weighted phantom mechanism induced by a family of phantom systems $\F$ on some profile $P$ with weights $\omega$ is indeed the same as running the unweighted phantom mechanism induced by $\F$ on the unweighted profile $\mathbb{P}(P, \omega)$. Note that, although the outcomes are identical, using weighted phantom mechanisms reduces the computation time from pseudo-polynomial to polynomial.

\begin{proposition}
    For any weighted profile $P$ with weights $\omega$, the weighted phantom mechanism $\B^\F$ induced by a family of phantom systems $\F$ returns the output of the unweighted moving-phantom mechanism $\A^\F$ for the profile $\mathbb{P}(P, \omega)$, that is, $\B^\F(P, \omega) = \A^\F(\mathbb{P}(P, \omega))$.
\end{proposition}

\proof
    Let $P$ be a weighted profile with weight vector $\omega$ and $n$ voters. Let $P' = \mathbb{P}(P, \omega)$ be the corresponding unweighted profile with $b = \sum_{i\in[n]} \omega_i$ voters. 
    Now, let $\F_{b} = \{f_k \mid k = 0, \dots, b\}$ be the phantom system inducing $\B^\F$. 
    
    We fix a time $t \in [0,1]$ and an alternative $j \in [m]$. Let $x = \med(f_{0}(t), \dots, f_{b}(t), p_{1,j}', \dots, p_{b,j}')$ be the median value of the votes from $P'$ and the phantoms and let $k'$ be the number of voters that vote strictly below $x$ on alternative $j$ in $P'$. Since $x$ is the median and there are $b-k'$ voters voting at least $x$ on $j$, there must be at least $k'+1$ phantoms on positions at least $x$ at time $t$, so in particular $f_{k'}(t) \ge x$. 
    Then, by the definition of $P'$, there must be an index $k \in [n]_0$ with $k' = \Omega_{k, j}$. Then, there are exactly $k$ voters voting strictly below $x$ in $P$ and thus $n-k$ voters voting at least $x$ on $j$. Additionally, we have the $k+1$ phantoms $f_{\Omega_{0,j}}(t) \ge \dots \ge f_{\Omega_{k,j}}(t) = f_{k'}(t) \ge x $  voting at least $x$ in $P$. 
    We can analogously show that the number of voters and phantoms voting below or equal $x$ in $P$ is at least $n+1$, showing that $x$ is indeed also the median of the phantoms $f_{\Omega_{0,j}}$ to $f_{\Omega_{n,j}}$ and votes in $P$ at time $t$.

    Since at any time $t$, the medians of both mechanisms are the same, they will reach normalization at the same time and return the same allocation.
\endproof

We define the weighted version of \GrdDec as the mechanism that first converts the weighted profile into an unweighted one and then applies \GrdDec. This algorithm can also be formulated independently of the unweighted version, which would involve assigning initial budgets proportional to voters' weights, running the outer loop for $k \in 1, \dots, b$ and charging voters in $N^+_j$ proportionally to their weights. Such a reformulation, however, would require redefining the values $\mu^k_j$, and since it does not yield many additional insights, we omit it.

Both the weighted \GrdDec mechanism and the weighted moving-phantom mechanisms admit a further generalization to real-valued weights $\omega_i$. 
A trivial adaptation would lead to an infinite number of outer iterations for \GrdDec and an infinite number of phantoms for moving-phantom mechanisms.
However, only a polynomial number of iterations/phantoms are actually relevant, and we can compute the outcomes by carefully controlling which outer iterations are executed in \GrdDec or, similarly, selecting the appropriate phantoms to consider in the moving-phantom mechanism.

\paragraph{Weighted Axioms}

We identified an easy way of constructing weighted mechanisms from unweighted ones. However, it is unclear whether the properties of these mechanisms carry over into the weighted setting. In this section, we show that this is indeed the case for the axioms considered in this paper. Note that this is not generally true. For instance, the axiom \textit{individual fair share (IFS)} \citep{bogomolnaia2005collective,brandl2021distribution}, which requires each voter's utility to be at least $\frac{1}{n}$, naturally generalizes to requiring that each voter's utility is at least $\frac{\omega_i}{b}$ in the weighted setting. Weighted IFS, however, is not necessarily satisfied when applying a mechanism satisfying IFS to the unweighted equivalent of a weighted profile. 

Since a major strength of moving-phantom mechanisms is truthfulness, we start by investigating whether truthfulness still holds in the weighted setting. To show that weighted moving-phantom mechanisms are indeed truthful, we 
observe via standard arguments that a slightly stronger property holds: 
No group of voters with the same preferences can profit from a manipulation, even if we allow these voters to have different misreports.

We say that a budget-aggregation mechanism $\A$ is \emph{unanimous-group truthful} if for all $n, m \in \N$ with $m \ge 2$, profile $P = (\p_1, \dots, \p_n) \in \mathcal{P}_{n,m}$, voter group $N \subset [n]$ with $p_i = p_h$ for all $i,h \in N$, and misreports $\p_i' \in \Delta^{(m-1)}$ for each $i \in N$, we have $u_i(\A(P)) \ge u_i(\A(P'))$ for all $i \in N$ where $P'$ is identical to $P$, except that each voter $i \in N$ votes $p_i'$.

\begin{proposition}
    Every truthful mechanism is unanimous-group truthful.
\end{proposition}

\proof
    Suppose there is a truthful, but not unanimous-group truthful mechanism $\A$. Then there must be $n,m \in \N$, a voter group $N$, a profile $P$, misreports $p_i'$ for $i \in N$ and the resulting profile $P'$ as in the definition of unanimous-group truthfulness, such that $u_i(\A(P)) < u_i(\A(P'))$ for some $i \in N$.

    We fix an arbitrary order over $N$ and consider the profiles $P^k$ for $k \in [|N|]_0$ in which only the first $k$ voters from $N$ misreport. Therefore, since $u_i(\A(P^0)) = u_i(\A(P)) < u_i(\A(P')) = u_i(\A(P^{|N|}))$, there must be some index $k^\star \in [|N|]$, where $u_i(\A(P^{k^\star-1})) < u_i(\A(P^{k^\star}))$. Since all voters in $N$ have the same utility function as $i$, this is also an improvement for the $k^\star$th voter in $N$, which contradicts the truthfulness of $\A$.
\endproof

Because each weighted voter corresponds to a unanimous group in the unweighted profile, any truthfulness violation in the weighted setting carries over to a unanimous-group truthfulness violation in the unweighted setting. We therefore obtain the following corollary.

\begin{corollary}
    Weighted moving-phantom mechanisms are strategy-proof.
\end{corollary}

While the truthfulness axiom did not require adaptation to the weighted setting, single-minded proportionality and decomposability require slight modifications.
For single-minded proportionality, we require the output of a mechanism to coincide with the \textit{weighted} mean of the votes in single-minded instances. It is straightforward to verify that any mechanism satisfying single-minded proportionality in the unweighted setting also satisfies it in the weighted setting. 
The definition of decomposability is modified only in that the contributions of each voter $i$ sum up to $\frac{\omega_i}{b}$ instead of $\frac{1}{n}$. This reflects the intuition that, in a donor coordination setting, each participant should be satisfied with how their contribution is spent. Again, it is easy to verify that every mechanism satisfying decomposability in the unweighted setting also does so in the weighted setting.

Since all relevant weighted axioms and mechanisms reduce to their unweighted counterparts when all voter weights are $1$, the incompatibility of neutrality, decomposability, and truthfulness established in \Cref{prop:decomp-truthfulness} remains valid in the weighted setting. The same holds for the fact that \GrdDec and \UtDec are not Pareto efficient, stated in \Cref{prop:greedydecomp_pareto_dominated}.

\paragraph{Social Welfare in Weighted Budget Aggregation}

We now discuss how the results in this paper can be lifted to the weighted setting. 
We start by defining the social welfare of an allocation $a$ for a weighted instance as $w(a) = \sum_{i \in [n]} \omega_i u_i(a)$. Each voter's effect on the social welfare is scaled by its weight, which matches the definition of utilitarian welfare in \citet{brandt2023balanced}. For any aggregate, its social welfare in a weighted instance equals its social welfare in the corresponding unweighted instance. Therefore, all social welfare bounds established in \Cref{sec:worst_case_price} extend to the weighted setting, with every dependence on the number of voters replaced by a dependence on the total voter weight $b$. In particular, the lower bound on the welfare approximation in terms of $n$ for single-minded proportional mechanisms from \Cref{prop:proportional_social_welfare_lower_bound_n} and the matching upper bounds for \UtP (\Cref{cor:util_prop_welfare_approximation_n}) and \GrdDec (\Cref{cor:greedy_decomp_welfare_approximation_n}) hold in the weighted setting with respect to the total weight of the voters $b$. We leave as an open question how these bounds change when expressed in terms of $n$ in the weighted setting. \changed{The corresponding lower bound (\Cref{prop:proportional_social_welfare_lower_bound_m}) and upper bounds (\Cref{cor:util_prop_greedy_decomp_welfare_approximation_m}) for the welfare approximation with respect to $m$ remain unchanged for the weighted setting.} 

For the welfare domination relation $\trianglerighteq$ defined in \Cref{sec:phantom_domination}, all our results carry over into the weighted setting without modification. In particular, \Cref{lem:phantom_domination} and the resulting \Cref{thm:phantom-domination} remain valid for the weighted moving-phantom mechanisms.

In \Cref{sec:optimal_decomposable}, we show that computing a welfare-optimal decomposable allocation is NP-hard, which remains true in the weighted setting, as it is more general. Finally, weighted \GrdDec provides an approximation to the welfare-optimal decomposable allocation, just as in the unweighted case (\Cref{thm:greedy-decomp-two}), with the dependence on the number of voters replaced by a dependence on $b$.

\subsection{Welfare Approximation of \PiecewiseUniform and \Ladder}\label{app:lb-im-pu}

We show that \PiecewiseUniform does not achieve the same welfare approximation as \UtP.

\begin{proposition}
    For every $n\in \N$ with $n\geq 4$ even, there are instances with $n$ voters in which, if $a$ is the outcome of \PiecewiseUniform, it holds $\frac{n+1}{3} w(a) = w(a^\opt)$.
\end{proposition}

\proof
    Let $n\in \N$ be an even number with $n\geq 4$, and consider the following profile with $m=\frac{n^2}{2}+1$ alternatives:
    Voters $i$ from $1$ to $\frac{n}{2}$ vote $\frac{1}{n}$ on every alternative 
    $ (i-1) \cdot n + \ell$ for $\ell\in [n]$;
    voters $i$ from $\frac{n}{2}+1$ to $n$ vote $1$ on alternative $\frac{n^2}{2}+1$.
    It is easy to see that the optimal allocation allocates all budget on the last alternative, i.e., $a^\opt=(0,\ldots,0,1)$, with $w(a^\opt)=\frac{n}{2}$.

    \PiecewiseUniform reaches normalization at time $t^\star=\frac{n+3}{2(n+1)}$, when the phantom positions are $f_k(t^\star)=\frac{(n-k)(2t^\star-1)}{n} = \frac{2(n-k)}{n(n+1)}$ for $k\geq \frac{n}{2}$, which are the relevant phantoms for this profile and time of normalization.
    In particular, we have $f_{n-1}(t^\star)=\frac{2}{n(n+1)}$ and $f_{\frac{n}{2}}(t^\star)=\frac{1}{n+1}$, so the allocation is $a=\big(\frac{2}{n(n+1)},\ldots,\frac{2}{n(n+1)},\frac{1}{n+1}\big)$ with welfare $w(a)=\frac{n^2}{2}\cdot \frac{2}{n(n+1)} + \frac{n}{2} \cdot \frac{1}{n+1} = \frac{3n}{2(n+1)}$.
    Since $\frac{w(a^\opt)}{w(a)}=\frac{n+1}{3}$, the result follows.
\endproof

\smallskip

In the following, we show that \Ladder almost guarantees a worst-case welfare guarantee of $\alpha^\star(n)=\max_{\ell\in[n]}\frac{n\ell}{n+\ell(\ell-1)}\in O(\sqrt{n})$. We write \emph{almost}, since the guarantee differs by a factor smaller than $\frac{\ell}{\ell-1}$. Moreover, we show in \Cref{prop:ladder-lower-bound} that this guarantee is asymptotically tight. 

\begin{proposition}
For every profile $P\in\mathcal P_{n,m}$, let $a$ be the outcome of \Ladder. Then, $$\frac{w(a^\opt)}{w(a)}
 \le \max_{\ell\in[n]}\frac{n\ell}{n+(\ell-1)^2}.$$
\end{proposition}

\proof
Let $t^\star$ be a normalization time for Ladder. We introduce the following two variables that will be useful throughout the proof: 
\[
 \begin{aligned}
 n_{+}&=\max_{j\in[m]}\bigl|\{i\in[n]:p_{i,j}>a_j\}\bigr|,\\
 n_{-}&=\min_{j\in[m]:\,a_j>0}\bigl|\{i\in[n]:p_{i,j}\ge a_j\}\bigr|.
 \end{aligned}
\]
Here, $n_{+}$ can be interpreted as the maximum number of voters that benefit from a small increase in one alternative. Similarly, $n_{-}$ is the minimum number of voters who lose utility from a small decrease in a funded alternative. 
Importantly, these numbers also bound larger changes since increasing funding for one alternative can
only reduce the number of voters who benefit from a further increase. Similarly, decreasing funding can only increase the number of voters who lose from a further
decrease. Thus each unit of added funding increases welfare by at most $n_{+}$, and each unit of removed budget loses at least $n_{-}$ welfare.

Every funded alternative has at least one voter whose report is weakly
above its allocation. Otherwise, all $n$ reports and the last phantom would lie below the median. Hence $n_{-}\geq 1$. If $n_{+}\le n_{-}$, no transfer of
budget can increase welfare, so $a$ is optimal and the claim follows.
Hence, assume from now on that $n_{+}>n_{-}$, and relabel the alternatives so that $a_1$ attains the maximum defining $n_{+}$.

\paragraph{Lower Bound on $a_1$.} We start by showing that
\begin{equation}\label{eq:median-counts}
 \begin{aligned}
 \bigl|\{i:p_{i,j}>a_j\}\bigr|
 &\le n(1-t^\star+a_j) &&\text{for every }j,\\
 \bigl|\{i:p_{i,j}\ge a_j\}\bigr|
 &\ge n(1-t^\star+a_j) &&\text{if }a_j>0.
 \end{aligned}
\end{equation}

First, note that $a_j\le t^\star$ because all $n+1$ phantoms are at most
$t^\star$. Moreover, since the positive phantoms are equally spaced with a distance of $\frac{1}{n}$, there are at least $n(t^\star-a_j)$ phantoms strictly above
$a_j$ and, when $a_j>0$, at most $n(t^\star-a_j)+1$ phantoms weakly above it. The median has at most $n$ entries strictly above it and at least $n+1$ entries weakly above it, proving \eqref{eq:median-counts}.

Now, note that $n_{-}>n(1-t^\star)$ since the right-hand side of the second inequality in
\eqref{eq:median-counts} exceeds $n(1-t^\star)$ for any alternative with $a_j>0$, so in particular also for the alternative defining $n_{-}$. 
Now, we apply the first inequality to alternative~$1$ and obtain that 
\begin{equation}\label{eq:allocation}
 a_1\ge\frac{n_{+}}{n}-(1-t^\star)>\frac{n_{+}-n_{-}}{n}.
\end{equation}

\paragraph{Welfare Bounds.}
Every funded alternative has at least $n_{-}$ voters whose reports are
weakly above its allocation, so every allocated unit contributes at
least $n_{-}$ welfare. On alternative~$1$, at least $n_{+}$ voters receive utility $a_1$. Its additional contribution is therefore
at least $(n_{+}-n_{-})a_1$, giving
\begin{equation}\label{eq:welfare}
 w(a)\ge n_{-}+(n_{+}-n_{-})a_1.
\end{equation}

Now, consider moving directly from $a$ to $a^{\opt}$. Budget removed from alternative~$1$ loses at least $n_{+}$ welfare per unit and gains at most $n_{+}$ at its destination, so such transfers cannot give a net gain.
Budget removed from any other alternative yields a net gain of at most $n_{+}-n_{-}$ per unit. The total amount that can be removed from those alternatives is at most $1-a_1$. Consequently,
\begin{equation}\label{eq:improvement}
 w(a^\opt)-w(a)\le(n_{+}-n_{-})(1-a_1).
\end{equation}

Applying first \Cref{eq:improvement}, then \Cref{eq:welfare}, and finally \Cref{eq:allocation} yields
\[
 \begin{aligned}
 \frac{w(a^\opt)}{w(a)}
 &\le 1+\frac{(n_{+}-n_{-})(1-a_1)}{w(a)}\\
 &\le 1+\frac{(n_{+}-n_{-})(1-a_1)}
                   {n_{-}+(n_{+}-n_{-})a_1}
 =\frac{n_{+}}{n_{-}+(n_{+}-n_{-})a_1} \le \frac{n_{+}}{n_{-}+(n_{+}-n_{-})^2/n}.
 \end{aligned}
\]

Now choose the index $\ell=n_{+}-n_{-}+1$. Since $1\le n_{-}<n_{+}\le n$, this index belongs to $[n]$.
Using $n_{+}=\ell+(n_{-}-1)$, rewrite:
\begin{equation}\label{eq:common-offset}
 \frac{n_{+}}{n_{-}+(n_{+}-n_{-})^2/n}
 =\frac{\ell+(n_{-}-1)}{1+(\ell-1)^2/n+(n_{-}-1)}
 \le\frac{\ell}{1+(\ell-1)^2/n}.
\end{equation}
The last inequality is purely algebraic. Since $\ell\in[n]$, we have
$1+(\ell-1)^2/n\le\ell$, so the fraction in the center is at least one. Deleting the same nonnegative amount $n_{-}-1$ from its numerator and denominator can only move it further away from one. 

Finally, the maximum in the theorem statement includes this particular $\ell$, proving the claim.
\endproof

\begin{proposition}\label{prop:ladder-lower-bound}
For every $n\ge3$, there exists a sequence of profiles $(P^{(r)})_{r \in \mathbb{N}}$ such that the corresponding outcomes of \Ladder, $(a_r)_{r \in \mathbb{N}}$, and the optimal outcomes, $(a^\opt_r)_{r \in \mathbb{N}}$, satisfy
$$
 \sup_{r \in \mathbb{N}} \frac{w(a^\opt_r)}{w(a_r)} = \max_{\ell\in[n]}\frac{n\ell}{n+(\ell-1)^2}.
$$
\end{proposition}

\proof
To construct $P^{(r)}$, consider $n$ to be fixed and choose $\ell$ attaining the maximum in the statement. For $n\ge3$, the value at $\ell=2$ exceeds the values at both $\ell=1$ and $\ell=n$, so $2\le\ell<n$. 
Let $\ell$ voters report their entire budget on alternative~$1$.
Each remaining voter spreads their report uniformly over $r$ private
alternatives. Thus,  the total number of
alternatives is $m=1+r(n-\ell)$. Allocating everything to alternative~$1$
is optimal: each budget unit contributes $\ell$ welfare there and at most one elsewhere. Thus, $w(a^\opt)=\ell$ on this profile.

In contrast, we claim that Ladder returns
$$
 a_1=\frac{\ell-1}{n}+\frac{n-\ell+1}{nm},
 \qquad
 a_j=\frac{n-\ell+1}{nm}\quad(j\ge2).
$$
To verify this, take $t^\star=(n-1)/n+(n-\ell+1)/(nm)$.
The $\ell$ unit reports on alternative~$1$ make its median
$f_{n-\ell}(t^\star)=t^\star - \frac{n-\ell}{n} = a_1$. On each private alternative, $a_j<1/m<1/r$,
so the $n-1$ zero reports and single report $1/r$ select the median
$f_{n-1}(t^\star)=t^\star - \frac{n-1}{n} = a_j$.
These medians sum to one, so they are indeed Ladder's allocation.
Each private alternative contributes exactly its allocation to the total welfare,
while alternative~$1$ contributes $\ell a_1$. Hence
$$
 w(a)=\ell a_1+\sum_{j=2}^m a_j=1+(\ell-1)a_1.
$$
Keeping $n$ and the chosen $\ell$ fixed, let $r\to\infty$.
Then $m\to\infty$ and $a_1\to(\ell-1)/n$, giving
\[
\frac{w(a^\opt)}{w(a)}
 =\frac{\ell}{1+(\ell-1)a_1}
 \longrightarrow\frac{n\ell}{n+(\ell-1)^2}.\qedhere
 \]
\endproof

\subsection{Lattice of Phantom Mechanisms for Two Alternatives}\label{app:latticeTwoAlternatives}

    For the case of two alternatives, we know that every moving-phantom mechanism is induced by a \emph{symmetric generalized median mechanism}~\citep[][Theorem 5]{freeman2021truthful}. More precisely, every moving-phantom mechanism can be expressed by a phantom system $\F_n=\{ f_k\mid k=0,\ldots,n\}$ satisfying the following two properties: (i) phantoms are constant, i.e., for all $k \in [n]_0$ there exists $f_k \in [0,1]$ with $f_k(t) = f_k$ for all $t \in [0,1]$, and (ii) $f_k=1-f_{n-k}$ for every $k\in \big\{0,\ldots,\big\lfloor\frac{n}{2}\big\rfloor\big\}$. It is then easy to see that, for any $P \in \mathcal{P}_{n,2}$, taking the median of the above phantoms and $P$ yields a normalized outcome. Note that symmetric generalized median mechanisms are well-defined only for $m=2$, in contrast to our main definition of a mechanism, which has to be defined for all $m$. 

We will show that the set of symmetric generalized median mechanisms forms a lattice under the dominance relation $\trianglerighteq$ for $m=2$. To state this precisely, define $\A \trianglerighteq_2 \mathcal{B}$ if and only if, for every profile $P \in \mathcal{P}_{n,2}$, the outcome of $\A$ yields at least as much social welfare as that of $\mathcal{B}$. The analogous statement for moving-phantom mechanisms fails because technically $\trianglerighteq_2$ is not antisymmetric on this set, as a lattice requires. This is merely a consequence of defining moving-phantom mechanisms simultaneously for all $n$ and $m$: distinct mechanisms may coincide for $m=2$ yet differ for larger $m$. However, an analogous statement could be formulated for equivalence classes of moving-phantom mechanisms that return identical outcomes for every $P \in \mathcal{P}_{n,2}$.

\begin{proposition}
    The dominance relation $\trianglerighteq_2$ forms a lattice among symmetric generalized median mechanisms. 
\end{proposition}

\proof
    We claim: For two symmetric phantom systems $\F=\{\F_n\mid n\in \N\}$ and $\G=\{\G_n\mid n\in \N\}$ given by $\F_n=\big\{f_k\mid k=0,\ldots,\big\lfloor\frac{n}{2}\big\rfloor\big\}$ and $\G_n=\big\{g_k\mid k=0,\ldots,\big\lfloor\frac{n}{2}\big\rfloor\big\}$ for each $n\in \N$, it holds that
    \[
        \A^\F \trianglerighteq \A^\G \Longleftrightarrow f_k\geq g_k \text{ for all } k\in \left\{0,\ldots,\left\lfloor\frac{n}{2}\right\rfloor\right\},\ n\in \N.
    \]
    Note that the claim immediately implies that the relation $\trianglerighteq$ is antisymmetric on symmetric generalized median mechanisms since  $\A^\F \trianglerighteq \A^\G$ and $\A^\G \trianglerighteq \A^\F$ implies that $f_k = g_k$ for all $k\in \left\{0,\ldots,\left\lfloor\frac{n}{2}\right\rfloor\right\}$, and therefore $\A^\F$ and $\A^\G$ are the same mechanisms. 
    
    To prove the claim, note that if the right-hand side of this equivalence is true, then we can apply \Cref{lem:phantom_domination} with $k^\star=\big\lfloor\frac{n}{2}\big\rfloor$, and conclude that $\A^\F \trianglerighteq \A^\G$.
    If the right-hand side is not true, there exists $n\in \N$ and $k\in \big\{0,\ldots,\big\lfloor\frac{n}{2}\big\rfloor\big\}$ such that $f_k<g_k$.
    Note that, when $n$ is even, $f_{\frac{n}{2}}=g_{\frac{n}{2}}=\frac{1}{2}$ by symmetry, so we have $k<\frac{n}{2}$.
    Consider the profile $P\in \mathcal{P}_{n,2}$ where the voters $i$ from $1$ to $n-k$ vote $(1,0)$ and the voters $n-k+1$ to $n$ vote $(0,1)$. 
    Then, the outcomes are $\A^\F(P)=(f_k,1-f_k)$ and $\A^\G(P)=(g_k,1-g_k)$, with respective welfares of
    \begin{align*}
        w(\A^\F(P)) = (n-k)f_k+k(1-f_k) & =k+(n-2k)f_k \\ & < k+(n-2k)g_k = (n-k)g_k + k(1-g_k) = w(\A^\G(P)),
    \end{align*}
    where the inequality follows from the facts that $k<\frac{n}{2}$ and $f_k<g_k$.
    
    Having established the claim, for any two families of symmetric phantom systems $\F=\{\F_n\mid n\in \N\}$ and $\G=\{\G_n\mid n\in \N\}$ given by $\F_n=\big\{f_k\mid k=0,\ldots,\big\lfloor\frac{n}{2}\big\rfloor\big\}$ and $\G_n=\big\{g_k\mid k=0,\ldots,\big\lfloor\frac{n}{2}\big\rfloor\big\}$ for each $n\in \N$, it follows that:
    \begin{enumerate}[label=(\roman*)]
        \item The least upper bound or \textit{join} of their associated moving-phantom mechanisms according to $\trianglerighteq$ is given by $\A^\F\vee \A^\G = \A^{\mathcal{H}}$, where $\mathcal{H}=\{\mathcal{H}_n\mid n\in \N\}$ is given, for each $n\in \N$, by the symmetric phantom system $\mathcal{H}_n=\big\{\max\{f_k,g_k\}\mid k=0,\ldots,\big\lfloor\frac{n}{2}\big\rfloor\big\}$;
        \item The greatest lower bound or \textit{meet} of their associated moving-phantom mechanisms according to $\trianglerighteq$ is given by $\A^\F\wedge \A^\G = \A^{\mathcal{H}}$, where $\mathcal{H}=\{\mathcal{H}_n\mid n\in \N\}$ is given, for each $n\in \N$, by the symmetric phantom system $\mathcal{H}_n=\big\{\min\{f_k,g_k\}\mid k=0,\ldots,\big\lfloor\frac{n}{2}\big\rfloor\big\}$.
    \end{enumerate}
\endproof

\clearpage
\bibliography{literature}

@article{caragiannis2022truthful,
  title={Truthful aggregation of budget proposals with proportionality guarantees},
  author={Caragiannis, Ioannis and Christodoulou, George and Protopapas, Nicos},
  journal={Artificial Intelligence},
  volume={335},
  pages={104178},
  year={2024},
  publisher={Elsevier},
  doi={10.1016/j.artint.2024.104178}
}

@article{bei2026individual,
  title={Individual Fairness in Budget Aggregation},
  author={Bei, Xiaohui and Brandt, Felix and Greger, Matthias and Peters, Jannik and Segal-Halevi, Erel and Suksompong, Warut},
  journal={arXiv preprint arXiv:2608.01228},
  year={2026}
}

@article{freeman2021truthful,
  title={Truthful aggregation of budget proposals},
  author={Freeman, Rupert and Pennock, David M and Peters, Dominik and Vaughan, Jennifer Wortman},
  journal={Journal of Economic Theory},
  volume={193},
  pages={105234},
  year={2021},
  publisher={Academic Press Inc.},
  doi={10.1016/j.jet.2021.105234}
}

@article{gonzalez1985clustering,
title = {Clustering to minimize the maximum intercluster distance},
journal = {Theoretical Computer Science},
volume = {38},
pages = {293--306},
year = {1985},
author = {Teofilo F. Gonzalez},
doi={10.1016/0304-3975(85)90224-5}
}

@techreport{nehring2019resource,
  author =        {Nehring, Klaus and Puppe, Clemens},
  institution =   {Karlsruhe Institute of Technology (KIT)},
  number =        {131},
  type =          {Working Paper Series in Economics},
  title =         {Resource allocation by frugal majority rule},
  year =          {2019},
}

@article{goel2019knapsack,
  title={Knapsack voting for participatory budgeting},
  author={Goel, Ashish and Krishnaswamy, Anilesh K and Sakshuwong, Sukolsak and Aitamurto, Tanja},
  journal={ACM Transactions on Economics and Computation},
  volume={7},
  number={2},
  pages={8:1--8:27},
  year={2019},
  publisher={ACM New York, NY, USA},
  doi={10.1145/3340230}
}

@inproceedings{lindner2008midpoint,
	title = {Allocating Public Goods via the Midpoint Rule},
	author = {Tobias Lindner and Klaus Nehring and Clemens Puppe},
	booktitle = {Proceedings of the 9th International Meeting of the Society for Social Choice and Welfare},
	year = {2008}
}

@article{bogomolnaia2005collective,
  title={Collective choice under dichotomous preferences},
  author={Bogomolnaia, Anna and Moulin, Herv{\'e} and Stong, Richard},
  journal={Journal of Economic Theory},
  volume={122},
  doi={10.1016/j.jet.2004.05.005},
  number={2},
  pages={165--184},
  year={2005},
  publisher={Elsevier}
}

@inproceedings{brandl2021distribution,
  title={Distribution Rules Under Dichotomous Preferences: Two Out of Three Ain't Bad},
  author={Brandl, Florian and Brandt, Felix and Peters, Dominik and Stricker, Christian},
  booktitle={Proceedings of the 22nd ACM Conference on Economics and Computation},
  pages={158--179},
  year={2021},
  doi={10.1145/3465456.3467653}
}

@inproceedings{MPS20a,
  title={Price of Fairness in Budget Division and Probabilistic Social Choice},
  author={Michorzewski, Marcin and Peters, Dominik and Skowron, Piotr},
  booktitle={Proceedings of the 34th AAAI Conference on Artificial Intelligence},
  pages={2184--2191},
  year={2020},
  doi={10.1609/aaai.v34i02.5594}
}

@article{elkind2026settling,
  title={Settling the score: Portioning with cardinal preferences},
  author={Elkind, Edith and Greger, Matthias and Lederer, Patrick and Suksompong, Warut and Teh, Nicholas},
  journal={Artificial Intelligence},
  pages={104487},
  volume={352},
  year={2026},
  publisher={Elsevier},
  doi={10.1016/j.artint.2026.104487}
}

@InProceedings{goyal2023low,
  author =	{Goyal, Mohak and Sakshuwong, Sukolsak and Sarmasarkar, Sahasrajit and Goel, Ashish},
  title =	{{Low Sample Complexity Participatory Budgeting}},
  booktitle =	{Proceedings of the 50th International Colloquium on Automata, Languages, and Programming},
  pages =	{70:1--70:20},
  year =	{2023},
  doi={10.4230/LIPIcs.ICALP.2023.70}
}

@inproceedings{freeman2024project,
	title = {Project-Fair and Truthful Mechanisms for Budget Aggregation},
	booktitle = {Proceedings of the 38th {AAAI} Conference on Artificial Intelligence},
	author = {Freeman, Rupert and Schmidt-Kraepelin, Ulrike},
    pages={9704--9712},
    year={2024},
    doi={10.1609/aaai.v38i9.28828}
}

@inproceedings{brandt2023balanced,
  title = {Balanced {{Donor Coordination}}},
  author = {Brandt, Felix and Greger, Matthias and {Segal-Halevi}, Erel and Suksompong, Warut},
  year = 2023,
  month = jul,
  pages = {299},
  booktitle = {Proceedings of the 24th {{ACM Conference}} on {{Economics}} and {{Computation}}},
  doi={10.1145/3580507.3597729}
}

@article{ABC+15a,
  title={Justified representation in approval-based committee voting},
  author={Aziz, Haris and Brill, Markus and Conitzer, Vincent and Elkind, Edith and Freeman, Rupert and Walsh, Toby},
  journal={Social Choice and Welfare},
  volume={48},
  number={2},
  pages={461--485},
  year={2017},
  publisher={Springer},
  doi={10.1007/s00355-016-1019-3}
}

@article{BBG+22a,
  title = {Funding Public Projects: {{A}} Case for the {{Nash}} Product Rule},
  author = {Brandl, Florian and Brandt, Felix and Greger, Matthias and Peters, Dominik and Stricker, Christian and Suksompong, Warut},
  year = 2022,
  journal = {Journal of Mathematical Economics},
  volume = {99},
  pages = {102585},
  doi={10.1016/j.jmateco.2021.102585},
  booktitle = {Journal of {{Mathematical Economics}}}
}

@inproceedings{becker2025efficiently,
author = {Becker, Patrick and Fries, Alexander and Greger, Matthias and Segal-Halevi, Erel},
title = {Efficiently Computing Equilibria in Budget-Aggregation Games},
year = {2026},
doi = {10.65109/CCMB1463},
booktitle = {Proceedings of the 25th International Conference on Autonomous Agents and Multiagent Systems},
pages = {1304--1313},
numpages = {10}
}

@inproceedings{BrPe24b,
  title = {Completing {{Priceable Committees}}: {{Utilitarian}} and {{Representation Guarantees}} for {{Proportional Multiwinner Voting}}},
  author = {Brill, Markus and Peters, Jannik},
  year = 2024,
  month = mar,
  pages = {9528--9536},
  doi = {10.1609/aaai.v38i9.28808},
  booktitle = {Proceedings of the 38th {{AAAI Conference}} on {{Artificial Intelligence}}}
}

@inproceedings{berg2024truthful,
  title = {Truthful {{Budget Aggregation}}: {{Beyond Moving Phantom Mechanisms}}},
  booktitle = {Proceedings of the 20th {{Conference}} on {{Web}} and {{Internet Economics}}},
  author = {{de Berg}, Mark and Freeman, Rupert and {Schmidt-Kraepelin}, Ulrike and Utke, Markus},
  year = {2026},
  pages={642},
  doi={10.1007/978-3-032-08560-3}
}

@article{EFI+24b,
  title = {The {{Price}} of {{Justified Representation}}},
  author = {Elkind, Edith and Faliszewski, Piotr and Igarashi, Ayumi and Manurangsi, Pasin and {Schmidt-Kraepelin}, Ulrike and Suksompong, Warut},
  year = 2024,
  journal = {ACM Transactions on Economics and Computation},
  volume = {12},
  number = {3},
  pages = {11:1--11:27},
  doi = {10.1145/3676953},
  booktitle = {{{ACM Transactions}} on {{Economics}} and {{Computation}}}
}

@article{bertsimas2011price,
  title={The Price of Fairness},
  author={Bertsimas, Dimitris and Farias, Vivek F and Trichakis, Nikolaos},
  journal={Operations Research},
  volume={59},
  number={1},
  pages={17--31},
  year={2011},
  publisher={INFORMS},
  doi = {10.1287/opre.1100.0865}
}

@article{caragiannis2012efficiency,
  title={The Efficiency of Fair Division},
  author={Caragiannis, Ioannis and Kaklamanis, Christos and Kanellopoulos, Panagiotis and Kyropoulou, Maria},
  journal={Theory of Computing Systems},
  volume={50},
  number={4},
  pages={589--610},
  year={2012},
  publisher={Springer},
  doi = {10.1007/s00224-011-9359-y}
}

@article{LaSk20b,
  title = {Utilitarian Welfare and Representation Guarantees of Approval-Based Multiwinner Rules},
  author = {Lackner, Martin and Skowron, Piotr},
  year = 2020,
  journal = {Artificial Intelligence},
  volume = {288},
  pages = {103366},
  issn = {00043702},
  doi = {10.1016/j.artint.2020.103366},
  urldate = {2026-01-14},
  langid = {english},
  booktitle = {Artificial {{Intelligence}}}
}

@inproceedings{TWZ20b,
  title={Price of Fairness in Budget Division for Egalitarian Social Welfare},
  author={Tang, Zhongzheng and Wang, Chenhao and Zhang, Mengqi},
  booktitle={Proceedings of the 14th International Conference on Combinatorial Optimization and Applications},
  pages={594--607},
  year={2020},
  doi={10.1007/978-3-030-64843-5_40}
}

@inproceedings{BrPe23a,
  title = {Robust and Verifiable Proportionality Axioms for Multiwinner Voting},
  author = {Brill, Markus and Peters, Jannik},
  year = 2023,
  pages = {301},
  booktitle = {Proceedings of the 24th {{ACM Conference}} on {{Economics}} and {{Computation}}},
  doi={10.1145/3580507.3597785}
}

@article{SEL+26,
title = {Proportional justified representation},
journal = {Artificial Intelligence},
volume = {353},
pages = {104503},
year = {2026},
issn = {0004-3702},
doi = {https://doi.org/10.1016/j.artint.2026.104503},
url = {https://www.sciencedirect.com/science/article/pii/S0004370226000299},
author = {Luis Sánchez-Fernández and Edith Elkind and Martin Lackner and Norberto {Fernández García} and Jesús A. Fisteus and Pablo {Basanta Val} and Piotr Skowron}
}

@article{brandt2025optimal,
  title={Optimal budget aggregation with star-shaped preference domains},
  author={Brandt, Felix and Greger, Matthias and Segal-Halevi, Erel and Suksompong, Warut},
  journal={Mathematics of Operations Research},
  year={2025},
  publisher={INFORMS},
  doi = {https://doi.org/10.1287/moor.2024.0723},
  note = {Forthcoming}
}
\bibliographystyle{abbrvnat} 

\end{document}